\documentclass[12pt,a4paper,reqno]{amsart}
\usepackage[pagewise]{lineno}
\usepackage[maxnames=50,sorting=nyt]{biblatex}
\usepackage{enumerate}
\usepackage{amsfonts, color}
\usepackage{amsthm}
\usepackage{amssymb}
\usepackage{rotating}
\usepackage{tabularx}
\usepackage{lscape}
\usepackage{dsfont}
\usepackage{mathtools}
\usepackage{float}
\usepackage{amsmath}
\usepackage{comment}
\usepackage{bm}
\usepackage{todonotes}
\usepackage{amscd}
\usepackage[utf8]{inputenc}
\usepackage{t1enc}
\usepackage[raggedrightboxes]{ragged2e}
\usepackage[mathscr]{eucal}
\usepackage{indentfirst}
\usepackage{graphicx}
\usepackage{tikz-cd}
\usepackage{hyperref}
\usepackage{graphics}
\usepackage[margin= 2.2cm]{geometry}
\usepackage{epstopdf} 
\usepackage{xcolor}
\hypersetup{backref=true,       
    pagebackref=true,               
    hyperindex=true,                
    colorlinks=true,                
    breaklinks=true,                
    urlcolor= black,                
    linkcolor= blue,                
    bookmarks=true,                 
    bookmarksopen=false,
    filecolor=black,
    citecolor=red,
    linkbordercolor=red
}

\newtheorem{prop}{Proposition}[section]
\newtheorem{lemma}{Lemma}[section]

\newtheorem{thm}{Theorem}[section]

\theoremstyle{definition}\newtheorem{defn}{Definition}[section]

\newtheorem{remark}{Remark}[section]

\newtheorem{?}[thm]{Problem}

\theoremstyle{plain}
\theoremstyle{definition}

\definecolor{mycol}{rgb}{0,0,1}

\numberwithin{equation}{section}

\usepackage{ragged2e}
\AtBeginDocument{\hypersetup{pdfborder={0 0 1}}}

\begin{document}

\title{\boldmath  Beauty And The Beast Part 2: \\
Apprehending The Missing Supercurrent}
\author[G. W. Moore]{Gregory W. Moore}
\author[R.K. Singh]
{Ranveer Kumar Singh}
\address[Gregory W. Moore]{NHETC and Department of Physics and Astronomy\\Rutgers University, 136 Frelinghuysen Rd, Piscataway, NJ 08854, USA}
\email{gwmoore@physics.rutgers.edu}
\address[Ranveer Kumar Singh]{NHETC and Department of Physics and Astronomy\\Rutgers University, 136 Frelinghuysen Rd, Piscataway, NJ 08854, USA}
\email{rks158@scarletmail.rutgers.edu}
%\subjclass[2010]{11F03, 11F11, 11F37.}

%\date{\today}

%\keywords{mock modular forms, polynomial growth, Eisenstein series, theta function}

\begin{abstract}
The Moonshine module is a $c=24$ conformal field theory (CFT) whose automorphism group is the Monster group. It was argued by Dixon, Ginsparg, and Harvey in \cite{Dixon:1988qd} that there exists a spin lift of the Moonshine CFT with superconformal symmetry. Reference  
\cite{Dixon:1988qd} did not provide an explicit construction of a superconformal current. The present paper provides an explicit construction of a supercurrent. In fact, we will construct several superconformal currents in a spin lift of the Moonshine CFT using techniques developed in \cite{Harvey:2020jvu}. In particular, our construction relies on error correcting codes. 
\end{abstract}
\maketitle
\tableofcontents
\newpage
\section{Introduction}

Moonshine phenomena have played an important role in the history of string theory. 
The explanation of Monstrous Moonshine by Borcherds, Frenkel, Lepowsky, and Meurman \cite{Borcherds:1983sq,Borcherds1992MonstrousMA,Borcherds:1997cx,Frenkel:1988xz} 
played an important role in the development of vertex operator algebra theory and string compactification on orbifolds \cite{Dixon:1986qv,Dixon:1985jw,Dixon:1986jc}. For a nice recent review see \cite{Harrison:2022zee}. In \cite{Harvey:1985fr} Harvey sketched out a clear physical interpretation of the construction of Frenkel, Lepowsky, and Meurman and a few years 
later, in a beautiful paper \cite{Dixon:1988qd}, Dixon, Ginsparg, and Harvey elaborated on the construction and suggested that there is an underlying superconformal algebra in the FLM construction. But reference \cite{Dixon:1988qd}, as was recently emphasized in \cite{Gaiotto:2018ypj}, did not provide an explicit construction of the underlying supercurrent. Although it was proved in \cite{huang_1996} that the Beauty and the Beast module does contain a vertex operator superalgebra (see Theorem \ref{thm:BBSVOA} for the precise statement) but the problem of constructing a supercurrent remains open. The aim of this paper is to provide explicit construction of a supercurrent. The key idea is to construct superconformal currents using the theory of (quantum and classical) error correcting codes, following ideas from \cite{Harvey:2020jvu}. When those ideas are applied to (a spin lift of) the Moonshine module\footnote{There are several terms used in the literature (and in this paper) for the Moonshine module, for example, FLM module, FLM Moonshine module, FLM VOA, Monster CFT, Monster module, Monster theory, Monster VOA, Moonshine CFT.} one is obliged to compute the operator products of bosonic twist fields. Such computations are technically difficult. Indeed the first such computation was a tour-de-force \cite{Dixon:1986qv}. In our computation we make heavy use of the formalism of Dolan, Goddard, and Montague \cite{Dolan:1989vr,Dolan:1994st} (which differs slightly from the standard formalism of vertex operator algebras) and arrive at sufficiently precise formulae that we can carry out the general ideas of \cite{Harvey:2020jvu} and reduce the construction of a superconformal current to the existence of certain error correcting codes related to sublattices of the Leech lattice. Along the way, we clarify the notion of spin lift of a CFT and show that the CFT in \cite{Dixon:1988qd} is a spin lift of the Monster CFT.\\\\
\par We give here an abbreviated technical summary of the paper aimed for physicists. 
We start with the Leech torus theory - the holomorphic  conformal field theory (CFT) described in Subsection \ref{subsec:untwistcft} based on the Leech lattice. Let $\Lambda_{\mathrm{L}}$ be the Leech lattice, considered as a subset of $\mathbb{R}^{24}$, with quadratic form inherited from the standard $\mathbb{R}^{24}$ inner product $\langle \cdot,\cdot\rangle$ (see Subsection \ref{subsec:lattice} for the detailed construction). Consider 24 free chiral bosons $X^j(z),~j=1,\dots,24$ with target space the Leech torus $\mathbb{T}^{24}:=\mathbb{R}^{24}/\Lambda_{\mathrm{L}}.$ The quotient $\mathbb{R}^{24}/\Lambda_{\mathrm{L}}$ is defined by the equivalence relation  
\begin{equation}
    \bm{x}\sim \bm{x}'\iff \bm{x}-\bm{x}'=2\pi\lambda
\label{eq:}
\end{equation}
for some $\lambda\in\Lambda_{\mathrm{L}}$. The field $X^j$ is expanded as 
\begin{equation}
X^{j}(z)=q^j+ip^j\ln z+i\sum_{n\neq 0}\frac{a^j_n}{n}z^{-n}.
\end{equation}
The modes satisfy the commutation relations 
\begin{equation}
[q^j,p^k]=i\delta^{jk},\quad [a^j_n,a^k_m]=n\delta^{jk}\delta_{n+m,0}.
\label{eq:untheisalg}
\end{equation}
We also require that $a^{j\dagger}_n=a^j_{-n}$. For each $\lambda\in\Lambda_{\mathrm{L}}$ we construct the Fock space $\mathscr{F}_{\lambda}$ over the highest weight vector $|\lambda\rangle$ defined by $a^j_{n}|\lambda\rangle=0$ for $n>0$ and $p^j|\lambda\rangle=\lambda^j|\lambda\rangle,$ see Subsection \ref{subsec:untwistcft} for details. The Hilbert space $\mathscr{H}(\Lambda_{\mathrm{L}})$ is then spanned by the Fock spaces. The stress tensor is given by 
\begin{equation}
T(z)=-\frac{1}{2}:\partial_z X\cdot\partial_z X:
\end{equation}
where the dot indicates contraction of $i$ index using $\delta_{ij}$.
The Virasoro modes are given by 
\begin{equation}
\begin{split}
&L_n=\frac{1}{2}\sum_{m\in\mathbb{Z}}a_{n-m}\cdot a_{m},\quad n\neq 0\\&L_0=\frac{1}{2}p^2+\sum_{n\in\mathbb{N}}a_{-n}\cdot a_n.
\end{split}
\end{equation}
They satisfy the Virasoro algebra with central charge $24$.
The torus partition function is then given by 
\begin{equation}
Z_{\Lambda_{\mathrm{L}}}(q)=\text{Tr} q^{L_0-c/24}=\frac{\Theta_{\Lambda_{\mathrm{L}}}(q)}{\eta^{24}(q)}=\frac{1}{q}+24+196884q+\dots=J(q)+24,\quad q=e^{2\pi i\tau}    
\end{equation}
with $J(q)=j(q)-744$ where $j(q)=1/q+744+\dots$ is Klein's famous weight-zero modular function generating the field of modular functions and $\tau$ is the complex moduli of the torus. In particular $J(q)$ has no constant term.\\\\
One can give a field theory description of the theory as follows: Consider fields $\mathcal{X}^{j}(z,\bar{z})$ satisfying the half lattice periodicity
\begin{equation}
    \mathcal{X}^j\sim \mathcal{X}^j+2\pi(\lambda^j/2),\quad \lambda=(\lambda^1,\dots,\lambda^{24})\in\Lambda_{\mathrm{L}}.
\end{equation}
Let $\{e_{\mu}\}_{\mu=1}^{24}$ be a generating set of $\Lambda_{\mathrm{L}}$ so that $\lambda=n^{\mu}e_{\mu}$ with $n^{\mu}\in\mathbb{Z}$ for $\lambda\in\Lambda_{\mathrm{L}}.$ We then start with the field theory action 
\begin{equation}\label{eq:curclyXaction}
    S=\frac{1}{2\pi}\int d^2z\left(\partial\mathcal{X}^j\Bar{\partial}\mathcal{X}^{j}+B_{jk}\partial\mathcal{X}^j\bar{\partial}\mathcal{X}^k\right)=\frac{1}{2\pi}\int d^2z\left(g_{\mu\nu}\partial\mathcal{X}^\mu\Bar{\partial}\mathcal{X}^{\nu}+B_{\mu\nu}\partial\mathcal{X}^\mu\bar{\partial}\mathcal{X}^\nu\right),
\end{equation}
where $g_{\mu\nu}=\langle e_{\mu},e_{\nu}\rangle$, $\mathcal{X}^{j}=(e_{\mu})^j\mathcal{X}^{\mu}$, $B_{jk}$ is a constant antisymmetric B-field and $B_{\mu\nu}=B_{jk}(e_{\mu})^j(e_{\nu})^k$. In general, the partition function for this field theory depends on both the holomorphic and antiholomorphic coordinates $q=e^{2\pi i \tau},\bar{q}=e^{2\pi i\bar{\tau}}$. If we require 
\begin{equation}
    B_{\mu\nu}\equiv \langle e_{\mu},e_{\nu}\rangle~\bmod~2,
    \label{eq:Bmunu}
\end{equation}
then the partition function takes the factorised form 
\begin{equation}
    Z(q,\bar{q})=\left|\frac{1}{\eta(q)^{24}}\sum_{\lambda\in\Lambda_{\mathrm{L}}}q^{\frac{\langle\lambda,\lambda\rangle}{2}}\right|^2,
\end{equation}
where $\eta$ is the Dedekind eta function. The Hilbert space of the theory also decomposes as a tensor product of holomorphic and antiholomorphic sectors. The restriction to the holomorphic sector gives us the Leech torus theory in the field theory description. We can relate the on-shell fields to the 
field $X^i(z)$ of the algebraic description above via   $\mathcal{X}^i(z,\bar{z})=\frac{1}{2}(X^i(z)+X^i(\bar{z}))$.  \\\\ 
\par The Monster CFT is constructed as a $\mathbb{Z}_2$-orbifold of the Leech torus CFT. To construct this orbifold CFT, we gauge the global $\mathbb{Z}_2$ symmetry defined by: 
\begin{equation}
    \theta X^j\theta^{-1}=-X^j,\quad \theta a^j_n\theta^{-1}=-a^j_n,\quad \theta |\lambda\rangle = |-\lambda\rangle
\end{equation}
where $\theta$ is the generator of the  $\mathbb{Z}_2$ automorphism. This defines an action on the Fock spaces and extends linearly to the Hilbert space. Now gauge it. The result is a theory of fields $X^j$ on the orbifold $\mathbb{T}^{24}/\mathbb{Z}_2$ where $\mathbb{Z}_2$ acts on $\mathbb{T}^{24}$ as $\bm{x}\mapsto -\bm{x}$. The fixed points of this action, $\frac{1}{2}\Lambda_{\mathrm{L}}/\Lambda_{\mathrm{L}}$,
form a group isomorphic to $T_2(\Lambda_{\mathrm{L}}):= \Lambda_{\mathrm{L}}/2\Lambda_{\mathrm{L}}$.  
%
%and it turns $\mathbb{T}/\mathbb{Z}_2$ into
%an orbifold. 
%
The B-field defines a bilinear form $\varepsilon$ on $T_2(\Lambda_{\mathrm{L}})$:
\begin{equation}\label{eq:sympformBfield}
    \varepsilon(\lambda_1,\lambda_2)=e^{\pi i\sum_{j<k}B_{jk}(\lambda_1)^j(\lambda_2)^k}.
\end{equation}
From eq.  \eqref{eq:Bmunu} we see that 
\begin{equation}\label{eq:cocycleBfield}
\frac{\varepsilon(\lambda_1,\lambda_2)}{\varepsilon(\lambda_2,\lambda_1)}=(-1)^{\langle \lambda_1,\lambda_2\rangle}.    
\end{equation}
Following the standard rules of gauge theory, the Hilbert space of the orbifold theory is a direct sum of spaces associated with the trivial and nontrivial $\mathbb{Z}_2$ bundles over the circle. The subspace associated with the trivial bundle is the  $\theta=1$ subspace $\mathscr{H}^+(\Lambda_{\mathrm{L}})$ of $\mathscr{H}(\Lambda_{\mathrm{L}})$, called the \textit{untwisted sector}. The subspace associated with the nontrivial bundle is known as the twisted sector.

To construct the twisted sector we start with fields $X^j(z)$ on $\mathbb{T}^{24}/\mathbb{Z}_2$ satisfying the antiperiodic boundary condition 
\begin{equation}
    X^j(e^{2\pi i}z)=-X^j(z) ~(\bmod~\Lambda_{\mathrm{L}}).
\end{equation}
These fields have mode expansion of the form 
\begin{equation}
X^j(z)=\Tilde{q}^j+i\sum_{k\in\mathbb{Z}+\frac{1}{2}}\frac{c^j_k}{k}z^{-k}
\end{equation}
where the modes  satisfy the same algebra as eq.  \eqref{eq:untheisalg}:
\begin{equation}
[c^j_k,c^m_{\ell}]=k\delta_{jm}\delta_{k+\ell,0}.
\end{equation}
Classically, the $\mathbb{Z}_2$-orbifold breaks the translation symmetry on the torus to the order two points $T_2(\Lambda_{\mathrm{L}})$. 
The unbroken translational symmetry is realised in the quantum theory via the representation of a central extension of $T_2(\Lambda_{\mathrm{L}})$, as is often the case in quantum theory. 
In the present case the central extension is the Heisenberg extension and fits in an exact sequence: 
\begin{equation}
   0\longrightarrow \mathbb{Z}_2\longrightarrow\Gamma(\Lambda_{\mathrm{L}})\xrightarrow[]{~~-~~}T_2(\Lambda_{\mathrm{L}})\longrightarrow 0 
   \label{eq:centextheis}
\end{equation}
where $-$ is the surjective homomorphism. This central extension is characterised by a cocycle $\varepsilon$ which is provided by the B-field (see eq.  \eqref{eq:sympformBfield}). We may choose a section $T:T_2(\Lambda_{\mathrm{L}})\longrightarrow \Gamma(\Lambda_{\mathrm{L}})$ such that
\begin{equation}
\begin{split}    &T([\lambda_1])T([\lambda_2])=\varepsilon(\lambda_1,\lambda_2)T([\lambda_1+\lambda_2])\\
%&T([\lambda_1])T([\lambda_2])=(-1)^{\langle \lambda_1,\lambda_2\rangle}T([\lambda_2])T([\lambda_1]).
\end{split}
\end{equation}
The commutator function is given by eq.  \eqref{eq:cocycleBfield}.
%
%This means that the cocycle must satisfy 
%\begin{equation}
%   \frac{\varepsilon(\lambda_1,\lambda_2)}{\varepsilon(\lambda_2,\lambda_1)}=(-1)^{\langle %\lambda_1,\lambda_2\rangle}
%\end{equation}
%
%which is satisfied, see 
%
There is a unique unitary irreducible real representation of the Heisenberg group $\Gamma(\Lambda_{\mathrm{L}})$ up to isomorphism. Let us denote it by $\mathcal{S}$. Since $|T_2(\Lambda_{\mathrm{L}})|=2^{24}$, $\mathcal{S}\cong \mathbb{R}^{2^{12}}$ as a vector space and can be constructed using gamma matrices (see the discussion below eq.  \eqref{eq:alphachi} for detailed construction). Let $\mathscr{F}_T$ be the Fock space constructed using the half-integral moded oscillators $c_r^j$. Then the twisted Hilbert space is
\begin{equation}
    \mathscr{H}_T(\Lambda_{\mathrm{L}})=\text{Span}_\mathbb{C}\{\mathscr{F}_T\otimes_{\mathbb{R}} \mathcal{S}\}.
\end{equation}
Note that $\mathscr{F}_T$ is complex and hence $\mathscr{H}_T(\Lambda_{\mathrm{L}})$ is a complex vector space. 
One can extend the action of the involution $\theta$ to $\mathscr{H}_T(\Lambda_{\mathrm{L}})$ by defining 
\begin{equation}
\theta c^j_k\theta^{-1}=-c^j_k,\quad \theta\chi_0=-\chi_0,\quad \chi_0\in\mathcal{S}. 
\label{eq:}
\end{equation}
Thus $\mathscr{H}_T(\Lambda_{\mathrm{L}})$ also decomposes as a direct sum of $\theta=\pm 1$ subspaces which we denote by $\mathscr{H}_T^{\pm}(\Lambda_{\mathrm{L}})$. The total Hilbert space of the Monster CFT is then 
\begin{equation}\label{eq:flmmod}
    \mathscr{H}_{\text{FLM}}:=\mathscr{H}^+(\Lambda_{\mathrm{L}})\oplus \mathscr{H}_T^+(\Lambda_{\mathrm{L}}).
\end{equation}
The vertex operators are defined precisely in Subsection \ref{subsec:twistcft} 
where the space $\mathscr{H}_{\text{FLM}}$ is also denoted 
$\widetilde{\mathscr{H}}(\Lambda_{\text{L}})$.   
This theory has the famous property that the automorphism group of the operator product algebra is the Monster group.
%
%
% To see this, we begin with the partition function obtained by tracing over %the subspace $\mathscr{H}^+(\Lambda_{\mathrm{L}})$. We use the projection operator %$\mathcal{P}=\frac{1}{2}(1+\theta)$. We get
%\begin{equation}
%\begin{split}
%    Z_{T}(q)=q^{-1}\text{Tr}\mathcal{P}q^{L_0}&=\frac{1}
%{2}q^{-1}\left(\text{Tr}q^{L_0}+\text{Tr}\theta q^{L_0}\right)\\&=\frac{1}
%%{2}Z_{\Lambda_{\mathrm{L}}}+\frac{1}{2}q^{-1}\text{Tr}\theta q^{L_0}.
%\end{split}    
%\end{equation}
%When we apply the modular transformation $S: \tau\to-1/\tau$ then the first %term in $Z_T$ is invariant since it is the partition function of the %modular invariant Leech torus theory, while the second term transforms to %trace over states 
%twisted by $\theta$ in the space direction but periodic in the time %direction. Thus to make the partition function modular invariant, we must %add a contribution from a twisted sector Hilbert space %%
%*************************************
%
The partition function for the theory is 
\begin{equation}
    Z_{\mathscr{H}_{\text{FLM}}}(q)=J(q).
\end{equation}

In section \ref{sec:scft}  below we give a careful 
explanation of a result of \cite{Lin:2019hks} that the  theory studied in \cite{Dixon:1988qd} is a spin lift
\footnote{Reference \cite{Dixon:1988qd} uses the term ``$\mathbb{Z}_2$ cover,'' without 
giving a precise definition. We make this notion precise by explaining the spin lift of a CFT, see Section \ref{sec:spin_lift}.}
of the Monster CFT corresponding to an involution $\iota$ of the Monster which acts as the identity on $\mathscr{H}^+(\Lambda_{\mathrm{L}})$ and as multiplication by $-1$ on $\mathscr{H}^+_T(\Lambda_{\mathrm{L}})$.   The Hilbert space of the theory is 
\begin{equation}
    \mathscr{H}_{\text{BB}}=\mathscr{H}(\Lambda_{\mathrm{L}})\oplus\mathscr{H}_T(\Lambda_{\mathrm{L}}).
\end{equation}
The main claim of \cite{Dixon:1988qd} is that this theory is superconformal. 
As we will discuss, some care must be taken here since $\mathscr{H}_{\text{BB}}$
is not actually a vertex operator superalgebra (SVOA). As in \cite{Dixon:1988qd}, we decompose 
\begin{equation}
\mathscr{H}_{\text{BB}}=\mathscr{H}_{\text{BB}}^{\text{NS}}\oplus\mathscr{H}_{\text{BB}}^{\text{R}} ~,   
\end{equation}
where 
\begin{equation}
\mathscr{H}_{\text{BB}}^{\text{NS}}:=\mathscr{H}(\Lambda_{\mathrm{L}})^+ \oplus\mathscr{H}_T(\Lambda_{\mathrm{L}})^-,\quad \mathscr{H}_{\text{BB}}^{\text{R}}:=\mathscr{H}(\Lambda_{\mathrm{L}})^- \oplus\mathscr{H}_T(\Lambda_{\mathrm{L}})^+~.    
\end{equation}
We will show in Theorem \ref{thm:BBSVOA} that the subspace 
\begin{equation}\label{eq:BB-SUBSVOA}
\mathscr{H}_{\text{BB}}^{\text{NS}}= \mathscr{H}(\Lambda_{\mathrm{L}})^+ \oplus\mathscr{H}_T(\Lambda_{\mathrm{L}})^-
\end{equation}
is an SVOA. Then we will show that   there exists a conformal dimension $\frac{3}{2}$ vertex operator that can serve as a supercurrent, thus making 
\eqref{eq:BB-SUBSVOA} a superconformal VOA (SCVOA). The remaining space
$\mathscr{H}_{\text{BB}}^{\text{R}}=\mathscr{H}(\Lambda_{\mathrm{L}})^- \oplus\mathscr{H}_T(\Lambda_{\mathrm{L}})^+$
is a module for this SCVOA. 

In some more detail, the construction of the superconformal current proceeds as follows. 
The spinor representation $\mathcal{S}$, which was projected out in the Monster CFT, contains states with conformal weight $24\cdot\frac{1}{16}=\frac{3}{2}$. These are the twist fields in the theory and have the correct conformal dimension of a supercurrent. The main new observation of the present paper is that there does in fact exist a spinor $\Psi\in\mathcal{S}$ such that the corresponding vertex operator $V(\Psi,z)$ can be identified with $T_F(z)$ in the OPE eq.  \eqref{eq:opetbtf}. It is not difficult to see that for any $\Psi \in \mathcal{S}$ the vertex operator $V(\Psi,z)$ is a Virasoro primary of weight $3/2$. With more effort, and making use of
the formalism of \cite{Dolan:1989vr,Dolan:1994st}, we show that the vertex operator $V(\Psi,z)$ satisfies the OPE
\begin{equation}
 V(\Psi,z)V(\Psi,w)\sim \frac{\langle\Psi|\Psi\rangle}{(z-w)^3}+\frac{1}{8}\frac{\langle\Psi|\Psi\rangle}{z-w}T(w)+\frac{1}{z-w}\sum_{\substack{\lambda\in\Lambda_{\mathrm{L}}\\\langle\lambda,\lambda\rangle=4}}\kappa_{\lambda}(\Psi)e^{i\lambda\cdot X(w)}   
\end{equation}
where 
\begin{equation}
    \kappa_{\lambda}(\Psi)=\langle\Psi|T([\lambda])|\Psi\rangle,\quad T([\lambda])\in\Gamma(\Lambda_{\mathrm{L}}).
\end{equation}
The problem of constructing a supercurrent then reduces to finding a spinor $\Psi$ such that $\kappa_{\lambda}(\Psi)=0$ for every $\lambda$ of norm-squared 4. Given such a non-zero vector we can normalise it so that $\langle\Psi|\Psi\rangle=4$ and then $V(\Psi,z)$   is a supercurrent with central charge 
$\hat c = 16$.  \\\\\

Note that for any Abelian subgroup $\hat{\mathcal{L}}\subset\Gamma(\Lambda_{\mathrm{L}})$, the operator 
\begin{equation}
P(\hat{\mathcal{L}}):=\mathcal{N}\sum_{[\lambda]\in\hat{\mathcal{L}}}T([\lambda])
\end{equation}
is a projection operator for a suitable normalization factor 
$\mathcal{N} = \vert \hat{\mathcal{L}}\vert^{-1}$. Moreover if $\hat{\mathcal{L}}$ is a maximal Abelian subgroup then $P(\hat{\mathcal{L}})$ is a rank one projection operator. We will show  that there exist maximal Abelian subgroups $\hat{\mathcal{L}}$ such that, for a suitably normalized $\Psi\in\text{Im}P(\hat{\mathcal{L}})$, the vertex operator $V(\Psi,z)$ is a supercurrent. Our proof relies on the existence of sublattices $\Lambda_{\text{SC}}\subset\Lambda_{\mathrm{L}}$ such that: 
\begin{enumerate}
    \item For every $\lambda_1,\lambda_2,\lambda\in\Lambda_{\text{SC}}$, $\langle\lambda_1,\lambda_2\rangle\equiv 0\bmod~2$ and $\langle\lambda,\lambda\rangle\equiv 0\bmod ~4$. 
    \item $\Lambda_{\text{SC}}$ does not contain any vector of norm-squared 4.
    \item $2\Lambda_{\mathrm{L}}\subset \Lambda_{\text{SC}}$ with index $2^{12}$. 
\end{enumerate}
Hypothesis (1) and (3) implies that $\widetilde{\Gamma}(\Lambda_{\text{SC}}):=\{T([\lambda]):[\lambda]\in\widetilde{\Lambda}_{\text{SC}}\}$ is an order $2^{12}$ \emph{Abelian} subgroup of $\Gamma(\Lambda_{\mathrm{L}})$, where $\widetilde{\Lambda}_{\text{SC}}:=\Lambda_{\text{SC}}/2\Lambda_{\mathrm{L}}$. Hypothesis (2) is a technical requirement crucial for the construction of the supercurrent.
We call such sublattices \textit{superconformal sublattices}.
We define a specific isomorphism\footnote{note that there is an obvious isomorphism between the two spaces once we choose a generating set for $\Lambda_{\mathrm{L}}$.} $T_2(\Lambda_{\mathrm{L}})\cong\mathbb{F}_2^{24}$ using the representation of $\Gamma(\Lambda_{\mathrm{L}})$ in terms of the Dirac gamma matrices (see \eqref{eq:isof} and the discussion below it). Then the chain of maps 
\begin{equation}
 \Gamma(\Lambda_{\mathrm{L}})\longrightarrow T_2(\Lambda_{\mathrm{L}})\longrightarrow\mathbb{F}_2^{24}   
\end{equation}
determines a $[24,12]$ binary code $\mathcal{C}_{\text{SC}}$ when restricted to $\widetilde{\Gamma}(\Lambda_{\text{SC}})$. Finally we show that
\begin{equation}
\langle\lambda,\lambda\rangle=4\implies \kappa_{\lambda}(\Psi)=0    
\end{equation}
if $\Psi\in\text{Im}\,P(\widetilde{\Gamma}(\Lambda_{\text{SC}}))$ because of the error-correcting properties of $\mathcal{C}_{\text{SC}}$.\\\\ 

An example of a superconformal lattice is given by an embedding of $\sqrt{2} \Lambda_{\mathrm{L}}$ into $\Lambda_{\mathrm{L}}$. 
In fact there can be several embeddings which cannot be rotated into each other by an element of the automorphism group $\mathrm{Co}_0$ of the Leech lattice. We investigate the different supercurrents that can be constructed in this way in Section \ref{sec:N>1susy}.  We give an explicit construction of two inequivalent  superconformal sublattices
using the theory of $\mathbb{Z}_4$ codes. One of these superconformal lattices gives rise to the binary Golay code, so that the corresponding supercurrent can be considered analogous to the Duncan spinor in Conway supermoonshine \cite{Duncan:2014eha}. In principle there might be further superconformal sublattices of $\Lambda_{\mathrm{L}}$ not obtained by embeddings of $\sqrt{2}\Lambda_{\mathrm{L}}$, but we have nothing to say about these.  Since there are inequivalent supercurrents one naturally asks if there is $\mathcal{N}\geq 2$ supersymmetry in the BB module. Unfortunately, the answer appears to be no, since there is no suitable $U(1)$ R-symmetry current. Nevertheless, there is an unusual algebraic structure in this model. A deeper understanding would most likely turn out to be fruitful.  \\\\

\par We comment here on some related recent literature.  Vertex operator superalgebras (SVOAs) of central charge $c=24$ were classified in \cite{Hohn:2023auw}. The SVOA of equation \eqref{eq:BB-SUBSVOA} is entry 968 of \cite[Table 9]{Hohn:2023auw}\footnote{We thank the reviewer for pointing this out.}.  Recently, it was also observed that most of the vertex operator superalgebras associated to the odd unimodular lattices of rank 24 have an $\mathcal{N}=4$ superconformal structure \cite{höhn2018vertex}. Moreover, one of the results of \cite{Gaiotto:2018ypj} was that the vertex operator superalgebra obtained by an orbifold of the VOA based on the odd Leech lattice has at least two inequivalent $\mathcal{N}=1$ superconformal structures. Indeed, \cite{Gaiotto:2018ypj} also shows that the orbifold of the VOA based on odd Leech lattice is isomorphic to the Beauty and the Beast SCFT, thus providing an indirect construction of a supercurrent in the Beauty and the Beast SCFT. As far as we understand, the isomorphism of their description to that used by \cite{Dixon:1988qd} is non-explicit. Our construction, which deals with the orbifold theory based on the even Leech lattice, gives a direct construction of the supercurrents in the Beauty and the Beast SCFT and hence is different from the considerations of \cite{Gaiotto:2018ypj,höhn2018vertex}. 
\\\\
\par Here is a brief outline of the remainder of the paper:  In Section \ref{sec:cftextension} we review the formalism of \cite{Dolan:1994st} for vertex operator algebras (VOAs) and their extensions by representations. In Section \ref{sec:latcft} we review the twisted and untwisted construction of lattice VOAs. We discuss the Monster VOA as the prime example. In Section \ref{sec:scft} we review the notion of a superconformal vertex operator algebra (SCVOA) and the spin lift of CFTs and apply it to Monster CFT in Subsection  \ref{subsec:BB} to show that the theory in \cite{Dixon:1988qd} is a spin lift of the Monster CFT. Finally in Section \ref{sec:supconfbb} we show how superconformal lattices give rise to supercurrents. In   Section \ref{sec:N>1susy} we construct two inequivalent supercurrents by constructing two inequivalent superconformal sublattices.   
There are two technical appendices. In Appendix  \ref{sec:codeslat} we review some aspects of binary and quaternary codes and their associated lattices that are used in the paper. A second Appendix \ref{App:TwistedLocality} gives a detailed proof of a crucial locality relation satisfied by twist fields.\\\\

\textbf{Acknowledgement.} GM warmly thanks L. Dixon, P. Ginsparg, and J. Harvey for discussions on various aspects of Moonshine over the years. We also thank Theo Johnson-Freyd, Gerald H\"{o}hn, Geoff Mason and Shu-Heng Shao  for some informative correspondence on aspects of the draft. R.K.S would like to thank Richard Borcherds, Scott Carnahan and James Lepowsky for some useful correspondence, Ashoke Sen for raising some interesting questions, Anindya Banerjee for many discussions on Moonshine, and Runkai Tao for help with mathematica code. Finally the authors would like to that the anonymous referees whose comments helped to improve the presentation of the paper. The research of G.M and  R.K.S is supported by the US Department of Energy under grant DE-SC0010008.

\section{Vertex Operator Algebras And Their Extensions}
\label{sec:cftextension}
In this section, we review vertex operator algebra in the formalism of \cite{Dolan:1994st}. See \cite{Frenkel:1988xz} for the mathematical definition of vertex operator algebras. We refer the reader to \cite{Dolan:1994st} for more details and the proofs of the results in this section.
\subsection{Definition Of A Vertex Operator Algebra}
\noindent We begin with the definition of a vertex operator algebra.
\begin{defn}
A \emph{vertex operator algebra} (VOA) is a tuple\footnote{\label{foot:VOA_CFT_spin} What we call a vertex operator algebra is called a conformal field theory in  \cite{Dolan:1994st}. In the physics literature, a VOA is the \textit{chiral} part of a CFT and the full CFT is \textit{constructed} from the irreducible representations of the VOA, see Subsection \ref{subsec:repvoa} for the definition of representation of a VOA. It is in this sense that we will use the word vertex operator algebra and conformal field theories in this paper, see \cite{Moore:1988qv,Moore:1989vd} for more details on this viewpoint. For the Monster VOA, there are no other irreducible representations other than itself. So we can (and do) use Monster CFT synonymously with Monster VOA. An attempt to formalize the full CFT in terms of the \textit{full non-chiral algebra} and its modules was made in \cite{huang_kong_2006,Moriwaki:2020cxf,Singh:2023mom}. Finally, we use the term spin CFT to mean a CFT whose state spaces on circle (a time slice of the cylinder) depends on the spin structure on the circle.} $\left(\mathscr{H}, \mathscr{F}, \mathbf{V},|0\rangle, \psi_{L}\right)$ where $\mathscr{H}$ is a Hilbert space of states, $\mathscr{F}$ a dense subspace called the \emph{Fock space} and $\mathbf{V}$ a set of linear operators 
on $\mathscr{H}$ called \emph{vertex operators} $V(\psi, z)$ which are linear in $\psi$ and depend meromorphically on $z\in\mathbb{C}$. They in one-to-one correspondence with the states $\psi \in \mathscr{F}$. The product of vertex operators $V(\psi_1,z_1)V(\psi_2,z_2)\cdots$ is defined, \emph{a priori}, only in the region $|z_1|>|z_2|>\dots$. 
There are two special states in $\mathscr{F}$, the \emph{vacuum} $|0\rangle$ and a \emph{conformal state} $\psi_{L}.$ The first five properties listed below must be satisfied. If in addition the last property (6) is satisfied, we call it an \emph{Hermitian vertex operator algebra} 
\begin{enumerate}
\item The \emph{moments} of the vertex operator of $\psi_{L}$ is given by by
\begin{equation}
V\left(\psi_{L}, z\right)=\sum_{n \in \mathbb{Z}} L_{n} z^{-n-2}
\end{equation}
and satisfies the Virasoro algebra
\begin{equation}
\left[L_{m}, L_{n}\right]=(m-n) L_{m+n}+\frac{c}{12} m\left(m^{2}-1\right) \delta_{m,-n}
\label{eq:viraalg}
\end{equation}
with $L_{n}^{\dagger}=L_{-n}$ and $L_{n}|0\rangle=0, n \geq-1$. The vertex operator $T(z)=V\left(\psi_{L}, z\right)$ is called the \emph{stress tensor}. 
\item The vertex operators satisfy
\begin{equation}
(i)\quad V(\psi, z)|0\rangle=e^{z L_{-1}} \psi\quad (ii) \quad V(\psi,z)V(\phi,w)=V(\phi,w)V(\psi,z)
\label{eq:VL-1,locrel}
\end{equation}
The relation (ii) is called the \emph{bosonic locality relation} and is interpreted as follows: the left hand side is \emph{a priori} well defined for $|z|>|w|$ and the right hand side is well defined for $|w|>|z|$. The locality relation means that the analytic continuation of the matrix elements of both sides to the whole $(z,w)$ plane must be equal except for poles at $z=w$ and $z,w=0,\infty$. A vertex operator is said to be \textit{local}  with respect to $\mathbf{V}$ if it satisfies the bosonic locality relation with every vertex operator in $\mathbf{V}$.
\item The operator $w^{L_0}:=e^{L_0\ln w}$ acts locally with respect to $\mathbf{V}$, that is $w^{L_0}V (\psi,z)w^{-L_0}$ is local with respect to $\mathbf{V}$. 
\item\label{it:L0lower_bound}  $L_0$ is diagonalizable and the spectrum of $L_0$ is bounded below. The eigenvalues $h_{\psi}$ of eigenstates $\psi$ of $L_0$ are called \emph{conformal weights}.
\item\label{it:cft_type} The vacuum is the only $\mathfrak{su}(1,1)$ invariant state in the theory, \textit{i.e.} the only state annihilated by $L_0$ and $L_{\pm 1}$.
\item \label{it:hermiticity} For all $\psi\in\mathscr{F}$, the vertex operator $V(e^{z^*L_1}z^{*-2L_0}\psi,1/z^*)^{\dagger}$, where $z^*$ is the complex conjugate of $z$, is local with respect to $\mathbf{V}.$
\end{enumerate}
\label{defn:cft}
\end{defn}
\begin{defn}
\begin{enumerate}
\item A state $\psi\in\mathscr{H}$ is called a \emph{quasi-primary state}  if $L_1\psi=0$. It is called a \emph{conformal primary state} if $L_n\psi=0$ for $n=1,2$ (and hence for all $n>0$ by Virasoro algebra). The corresponding vertex operator $V(\psi,z)$ is called \emph{quasi-primary field} and \emph{conformal primary field} respectively.
\item The moments of a vertex operator corresponding to an eigenstate $\psi$ of $L_0$ with conformal weight $h_{\psi}$ are defined to be the modes in the Laurent expansion around zero:
\begin{equation}
V(\psi,z)=\sum_{n\in \mathbb{Z}}V_n(\psi)z^{-n-h_{\psi}}
\end{equation}
where $z\not=0,\infty$. 
\end{enumerate}
\end{defn}
\noindent The following properties can be proved as a consequence of the definition of an Hermitian VOA (see \cite{Dolan:1994st} for proofs):
\begin{prop}
Let $\left(\mathscr{H}, \mathscr{F}, \mathbf{V},|0\rangle, \psi_{L}\right)$ be an Hermitian vertex operator algebra. Then the following statements hold:
\begin{enumerate}
\item We have that 
\begin{equation}\label{eq:transprop}
    \frac{d}{dz}V(\psi,z)=V(L_{-1}\psi,z).
\end{equation}
This is called the translation property.
\item The conformal weights of an Hermitian VOA are non-negative integers. 
\item Any conformal primary field $V(\psi,z)$ with $\psi$ an eigenstate of $L_0$ with conformal weight $h_{\psi}$ satisfies the commutator 
\begin{equation}\label{eq:LncommVpsiz}
    [L_n,V(\psi,z)]=z^{n}\left(z\frac{d}{dz}+(n+1)h_\psi\right)V(\psi,z).
\end{equation}
Equivalently, the moments satisfy
\begin{equation}
[L_m,V_n(\psi)]=((h_{\psi}-1)m-n)V_{n+m}(\psi).
\label{eq:momcomlm}
\end{equation}
For a quasi-primary field, \eqref{eq:LncommVpsiz} and \eqref{eq:momcomlm} is true only for $L_0,L_{\pm 1}$.
\item The \emph{operator product expansion} (OPE) of any two vertex operators  $V(\psi,z)$ and $V(\phi,w)$ corresponding to eigenstates of $L_0$ is given by\begin{equation}
\begin{split}
V(\psi,z)V(\phi,w)&=\sum_{n=0}^{\infty}(z-w)^{-n-h_{\psi}-h_{\phi}}V(\phi_n,w)\\&=\sum_{n=0}^\infty(z-w)^{-h_{\phi_n}-h_{\psi}-h_{\phi}}V(\phi_n,w),
\end{split}
\label{eq:genope}
\end{equation}
where
\begin{equation}
    \phi_n:=V_{h_{\phi}-n}(\psi)\phi.
\end{equation}
Equivalently, \footnote{The left hand side of \eqref{eq:dualityV} is defined in the domain $|z|>|\zeta|$ while the right hand side of \eqref{eq:dualityU} is defined in the domain $|\zeta|>|z-\zeta|$. So the equality in \eqref{eq:dualityV} must be understood in the sense of analytic continuation of the matrix elements of the two sides as in the bosonic locality relation (ii) in \eqref{eq:VL-1,locrel}.}
\begin{equation}\label{eq:dualityV}
    V(\psi,z)V(\phi,w) = V( V(\psi,z-w)\phi, w) ~ . 
\end{equation}
\item For each $\psi\in\mathscr{F}$, the locality of $V(e^{z^*L_1}z^{*-2L_0}\psi,1/z^*)^{\dagger}$ implies that we have an antilinear map
\begin{equation}
    \psi\longmapsto \overline{\psi},
\end{equation}
where 
\begin{equation}
    V(e^{z^*L_1}z^{*-2L_0}\psi,1/z^*)^{\dagger}=V(\overline{\psi},z)
\end{equation}
so that 
\begin{equation}
\overline{\psi}:=\lim_{z\to 0}V(e^{z^*L_1}z^{*-2L_0}\psi,1/z^*)^{\dagger}|0\rangle.
\end{equation}
\item For a quasi-primary vertex operator, the moments satisfy $V_n(\psi)^{\dagger}=V_{-n}(\overline{\psi}).$
\end{enumerate}
\label{prop:propcft}
\end{prop}
\begin{remark}
We would like to put the definition above into perspective. The notion of vertex (operator) algebras was introduced by Borcherds in \cite{Borcherds:1983sq} and Frenkel-Lepowsky-Meurman in \cite{Frenkel:1988xz}. The notion of vertex operator algebra in \cite{Frenkel:1988xz} differs from ours in the sense that \cite{Frenkel:1988xz} uses formal calculus to define vertex operators and the locality axiom is replaced by Jacobi identity. In \cite{frenkel_huang_lepowsky_1993}, it has been shown that locality (called commutativity in \cite{frenkel_huang_lepowsky_1993}) and OPE \eqref{eq:dualityV} (called associativity in \cite{frenkel_huang_lepowsky_1993}) of vertex operators is equivalent to the Jacobi identity. The lower boundedness of the spectrum of $L_0$ in item \ref{it:L0lower_bound} corresponds to a grading-restriction condition on the VOA. Moreover, the Hermiticity axiom in item \ref{it:hermiticity} above is equivalent to the statement that the VOA is isomorphic to its \textit{contragradient} as modules for itself, see \cite{frenkel_huang_lepowsky_1993} for more details.     
\end{remark}
\begin{remark}
Usually primary fields in a VOA are defined based on their OPE with the stress tensor. A field $\phi(z)$ is called a primary field of conformal dimension $h$ if it has the following OPE with $T(z)$: 
\begin{equation}
T(z)\phi(w)\sim\frac{h}{(z-w)^2}\phi(w)+\frac{\partial_w\phi(w)}{z-w}~.
\end{equation}
Using the standard contour integral manipulations, one can show that this definition is equivalent to commutator \eqref{eq:momcomlm}. 
\end{remark}
\begin{defn}
A sub-vertex operator algebra (sub-VOA) of a VOA $\mathscr{H}$ is defined to be a subspace $\mathscr{I}\subset\mathscr{H}$ such that 
\begin{enumerate}
\item $\mathscr{I}$ is an invariant subspace for all the vertex operators $V(\phi,z)$ for $\phi\in\mathscr{F}_{\mathscr{I}}:=\mathscr{F}\cap \mathscr{I}$.
\item $\mathscr{I}$ is invariant under the $\mathfrak{su}(1,1)$ algebra spanned by $L_{\pm 1},L_0$.
\item $\overline{\mathscr{F}_{\mathscr{I}}}:=\{\overline{\phi}:\phi\in\mathscr{F}_{\mathscr{I}}\}=\mathscr{F}_{\mathscr{I}}$. 
\end{enumerate}
\end{defn}  
One can show that the sub-VOA contains a conformal state and that 
a sub-VOA of a (Hermitian) VOA is itself a (Hermitian) VOA. 
\begin{defn}
We say that two VOAs $\mathscr{H}$ and $\mathscr{H}'$ with Fock spaces $\mathscr{F}$ and $\mathscr{F}'$ respectively and corresponding vertex operators $V(\psi,z)$ and $V'(\psi',z)$ are isomorphic if there exists a unitary map $u:\mathscr{H}\longrightarrow\mathscr{H}'$ such that 
\begin{equation}
V'(u\psi,z)=uV(\psi,z)u^{-1},\quad \psi\in\mathscr{F}.    
\end{equation}
An isomorphism  $u:\mathscr{H}\longrightarrow\mathscr{H}$ which preserves $|0\rangle,\psi_L$ is called an automorphism.  
The group of all automorphisms of a VOA $\mathscr{H}$ is called the automorphism group of the VOA.
\end{defn}
If $u:\mathscr{H}\longrightarrow\mathscr{H}'$ is an isomorphism then 
\begin{equation} u|0\rangle=\left|0^{\prime}\right\rangle, u \psi_L=\psi_L^{\prime}, 
\end{equation}
where $|0\rangle,\left|0^{\prime}\right\rangle$ are the vacuum states and $\psi_L$, $\psi_L^{\prime}$ are the conformal states in $\mathscr{H}$ and $\mathscr{H}^{\prime}$ respectively.  \subsection{Representation Of A Vertex Operator Algebra}\label{subsec:repvoa}
\begin{defn}
A representation $(\mathbf{U},\mathscr{K})$ of a vertex operator algebra $(\mathscr{H},\mathbf{V},\mathscr{F},|0\rangle, \psi_L)$ is a Hilbert space $\mathscr{K}$ and a set of linear operators $U(\psi,z):\mathscr{K}\longrightarrow\mathscr{K}$ meromorphic in $z$ and linear in $\psi$ for $\psi\in\mathscr{F}$ and for all $\psi,\phi\in\mathscr{F}$, we have:
\begin{equation}\label{eq:dualityU}
    U(\psi, z)U(\phi, \zeta) = U(V(\psi, z - \zeta)\phi, \zeta)
\end{equation}
with $U(|0\rangle,z)\equiv \mathds{1}$. The equality is understood as in \eqref{eq:dualityV}.
\label{def:repofcft}
\end{defn}
%The left hand side of \eqref{eq:dualityU} is defined in the domain $|z|>|\zeta|$ while the right hand side of \eqref{eq:dualityU} is defined in the domain $|\zeta|>|z-\zeta|$. So the equality in \eqref{eq:dualityU} must be understood in the sense of analytic continuation of the matrix elements of the two sides as in the bosonic locality relation (ii) in \eqref{eq:VL-1,locrel}. 
One can show that the operators $U(\psi,z)$ are local in the sense of bosonic locality relation. Moreover, the modes of $U(\psi_L,z)$ satisfies the Virasoro algebra and $U(\psi_L,z)$ also satisfies the translation property as in \eqref{eq:transprop}. For $\chi\in\mathscr{K}$ let us define a meromorphic family of operators $W(\chi,z):\mathscr{H}\longrightarrow\mathscr{K}$ by  
\begin{equation}
W(\chi,z)\phi := e^{zL_{-1}} U(\phi,-z)\chi. 
\label{eq:defW}
\end{equation}
Then we have the following proposition.
\begin{prop}\emph{\cite[Proposition 3.4]{Dolan:1994st}}
The existence of a representation is equivalent to the existence of ``intertwining'' operators $W(\chi,z):\mathscr{H}\longrightarrow\mathscr{K}$ for $\chi\in\mathscr{K}$ (or a dense subspace of $\mathscr{K}$) satisfying 
\begin{equation}
    U(\psi,z)W(\chi,\zeta)=W(\chi,\zeta)V(\psi,z)
\label{eq:}
\end{equation}
and
\begin{equation}
    \lim_{\zeta\to 0}W(\chi,\zeta)|0\rangle=\chi.
\label{eq:}
\end{equation}
\end{prop}
\subsection{Extension Of A Vertex Operator Algebra By A Real Hermitian
Representation}
\begin{defn}\label{def:hermrep}
A representation $\mathbf{U}$ is said to be Hermitian if 
\begin{equation}\label{eq:Uphibarz}
    U(\overline{\phi},z)=U(e^{z^*L_1}z^{*-2L_0}\phi,1/z^*)^{\dagger}
\end{equation}
for every $\phi\in\mathscr{F}$. Here $L_1$ is the Laurent mode of $U(\psi_L,z)$. The representation is said to be real if there is an antilinear map $\chi\mapsto \overline{\chi}$ in $\mathscr{K}$ satisfying 
\begin{equation}
    \begin{split}
\overline{\overline{\chi}}=\chi,\quad L_{-1}\overline{\chi}=-\overline{L_{-1}\chi},\\f_{\chi_1\phi\chi_2}=(-1)^{h_{\chi_1}+h_{\phi}+h_{\chi_2}}f_{\overline{\chi_1}\overline{\phi}\overline{\chi_2}}
\end{split}
\label{eq:}
\end{equation}
where 
\begin{equation}
    f_{\chi_1\phi\chi_2} :=\langle \overline{\chi_1}|U(\phi,1)|\chi_2\rangle.
\end{equation}
Moreover if $L_0\chi=h_{\chi}\chi$ then $L_0\overline{\chi}=h_{\chi}\overline{\chi}.$
\end{defn}
\begin{remark}\label{rem:Udaggerrep}
The above condition \eqref{eq:Uphibarz} is equivalent to 
\begin{equation}
    U(\overline{\phi},z)=z^{-2h_{\phi}}U(e^{z^*L_1}\phi,1/z^*)^{\dagger}
\label{eq:}
\end{equation}
for $\phi\in\mathscr{F}$ with conformal weight $h_{\phi}.$
\end{remark}
\begin{defn}
Define the conjugate $\overline{W}$ of the operator $W(\chi,z):\mathscr{H}\longrightarrow\mathscr{K}$ by the relation:
\begin{equation}
\overline{W}(\overline{\chi},z)=z^{-2h_{\chi}}W(e^{z^*L_1}\chi,1/z^*)^{\dagger}.
\label{eq:defWbar}
\end{equation}

\end{defn}
\begin{prop}\emph{\cite[Proposition 4.3]{Dolan:1994st}}
The operator $\overline{W}:\mathscr{K}\longrightarrow\mathscr{H}$ satisfies the following intertwining relations
\begin{equation}
    \begin{split}
\overline{W}(\chi,\zeta)U(\psi,z)=V(\psi,z)\overline{W}(\chi,\zeta)\\\overline{W}(\chi_1,\zeta)W(\chi_2,z)=\overline{W}(\chi_2,z)W(\chi_1,\zeta).
\end{split}
\label{eq:}
\end{equation}
\label{prop:reploc}
\end{prop}
We have the following important theorem:
\begin{thm}\emph{\cite[Proposition 4.4]{Dolan:1994st}}
Suppose that the operator $\overline{W}(\chi,z)$ satisfy the additional locality relation 
\begin{equation}
W(\chi_1,\zeta)\overline{W}(\chi_2,z)=W(\chi_2,z)\overline{W}(\chi_1,\zeta)
\end{equation}
and the spectrum of $L_0$ in the representation is strictly positive, then we may extend the vertex operator algebra $\mathscr{H}$ to a Hermitian vertex operator algebra $\widetilde{\mathscr{H}}:=\mathscr{H}\oplus\mathscr{K}$ with Fock space $\widetilde{\mathscr{F}}:=\mathscr{F}\oplus\mathscr{K}$  and vertex operators defined by 
\begin{equation}
    \widetilde{V}(\psi,z)=\begin{pmatrix}
V(\psi,z)&0\\0&U(\psi,z)
\end{pmatrix},\quad \widetilde{V}(\chi,z)=\begin{pmatrix}
0&\overline{W}(\chi,z)\\W(\chi,z)&0
\end{pmatrix}
\label{eq:}
\end{equation}
where we used the notation $\psi=(\psi,0)$ and $\chi=(0,\chi)$ for $\psi\in\mathscr{H}$ and $\mathscr{K}$ and the definition is extended to all of $\widetilde{\mathscr{F}}$ by linearity. The vacuum and conformal states are $(|0\rangle,0)$ and $(\psi_L, 0)$ respectively.
\label{thm:extensionofcft}
\end{thm}
\begin{remark}
The extended vertex operator algebra $\widetilde{\mathscr{H}}$ has an automorphism $\iota$ which acts as identity on $\mathscr{H}$ and multiplication by $-1$ on $\mathscr{K}.$
\end{remark}

\section{Lattice Vertex Operator Algebras And Their Orbifolds}\label{sec:latcft}
In this section, we study VOAs associated to lattices, first studied in 
\cite{Borcherds:1983sq,Frenkel:1988xz} Let $\Lambda$ be an even Euclidean lattice of dimension $r$. We can construct a VOA using this lattice. We can construct two types of VOAs, the untwisted or straight VOA and the twisted VOA. We describe this construction in this section. We follow \cite{Dolan:1994st} for this construction.
\subsection{Untwisted Vertex Operator Algebra}\label{subsec:untwistcft}
Consider the $r$-torus $\mathbb{T}^r:=\mathbb{R}^r/\Lambda$ where the equivalence relation is 
\begin{equation}
    \bm{x}\sim \bm{x}'\iff \bm{x}-\bm{x}'=2\pi\beta
\label{eq:}
\end{equation}
for some $\beta\in\Lambda$. We consider $r$ fields 
\begin{equation}
X^{j}(z)=q^j+ip^j\ln z+i\sum_{n\neq 0}\frac{a^j_n}{n}z^{-n}
\end{equation}
which we can view as in correspondence with the coordinates on the torus $\mathbb{T}^r.$ 
We impose the commutation relations on the modes:
\begin{equation}
[q^j,p^k]=i\delta^{jk},\quad [a^j_n,a^k_m]=n\delta^{jk}\delta_{n+m,0}. 
\end{equation}
We then introduce a family of orthogonal lines $L_\lambda$, 
one for each $\lambda\in\Lambda$. We require that $a^j_n$ for $n>0$ annihilates $L_\lambda$ and $p^j$ acts on $L_{\lambda}$ as scalar multiplication by $\lambda^j$. We may choose a family of 
basis vectors $\vert \lambda \rangle $ for $L_{\lambda}$ which are orthonormal: 
$\langle \lambda \vert \lambda' \rangle = \delta_{\lambda,\lambda'}$. Any two choices are related 
by a uniform multiplication by a phase. 
  With this inner product, we have the Hermitian conjugates $a^{j\dagger}_n=a^j_{-n}$. For each $\lambda\in\Lambda$, we then construct the Fock space $\mathscr{F}_{\lambda}$ by creating states using $a^j_n$. Thus a 
  basis is of the form: 
\begin{equation}
\psi=\left(\prod_{a=1}^M a^{j_a}_{-m_a}\right)|\lambda\rangle,
\label{psiflamda}
\end{equation}
where $m_a$ and $j_a$ are positive integers $1\leq j_a\leq r$ and $M\in\mathbb{Z}_+$ with the understanding that $\psi=|\lambda\rangle$ when $M=0$. The Hilbert space of states is then the (completion of the) span of the Fock spaces $\mathscr{F}_{\lambda}$ for $\lambda\in\Lambda$. The Fock space of the VOA is defined as 
\begin{equation}
    \mathscr{F}=\bigoplus_{\lambda\in\Lambda}\mathscr{F}_{\lambda}.
\label{eq:fockuntwist}
\end{equation}
We denote the Hilbert space by $\mathscr{H}(\Lambda)$.
For $\psi\in\mathscr{F}_{\lambda}$ given by \eqref{psiflamda}, we define the vertex operator $V(\psi,z)$ by the expression 
\begin{equation}
V(\psi,z)=~ :\left(\prod_{a=1}^{M} \frac{i}{\left(m_{a}-1\right) !} \frac{d^{m_{a}} X^{j_a}}{d z^{m_{a}}}(z)\right) \exp~ \{i \lambda \cdot X(z)\}: \sigma_{\lambda},
\label{eq:veropuntwisted}
\end{equation}
where $::$ denotes normal ordering and is defined as follows:
\begin{equation}
\begin{split}
&: a_n^j a_m^k: = : a_m^k \,  a_n^j: =\begin{cases}
a_m^k \, a_n^j& m\leq n,\\a_n^j \,  a_m^k& m\geq n,
\end{cases}
\\&: a^j_m p^k: =: p^k \,  a_m^j: = a_m^j p^k ,\\&: q^{j} \, a_n^k : = : a_n^k q^{j}:=a_n^k q^{j},\\&: q^j p^k: =: p^k \,  q^j: = q^j p^k. 
\end{split}    
\end{equation}
In particular 
\begin{equation}
:\exp~ \{i \lambda \cdot X(z)\}:= \exp~ \left\{i \lambda \cdot X_<(z)\right\}e^{i\lambda\cdot q}z^{i\lambda\cdot p}\exp~ \left\{i \lambda \cdot X_>(z)\right\}   
\end{equation}
where 
\begin{equation}
    X_{\lessgtr}^j(z)=i \sum_{n \lessgtr 0} \frac{a_n^j}{n} z^{-n}.
\end{equation}
The operator $\sigma_{\lambda}$ is a cocycle operator such that
\begin{equation}
\widehat{\sigma}_{\lambda}\widehat{\sigma}_{\mu}=(-1)^{\lambda\cdot\mu}\widehat{\sigma}_{\mu}\widehat{\sigma}_{\lambda}
\end{equation}
with $\widehat{\sigma}_{\lambda}=e^{i\lambda\cdot q}\sigma_{\lambda}$. Also the operator $q$ acts on the states $|\lambda\rangle$ only via the exponential 
\begin{equation}\label{eq:mudorqactonlamb}
e^{i\mu\cdot q}|\lambda\rangle=|\lambda+\mu\rangle,\quad\lambda,\mu\in\Lambda. 
\end{equation}
This cocycle is introduced to ensure the mutual locality and associativity of the operator product expansion.
The vacuum state is $|0\rangle$ corresponding to $0\in\Lambda$ and the conformal state $\psi_L$ is defined as 
\begin{equation}
\psi_L=\frac{1}{2}a_{-1}\cdot a_{-1}|0\rangle.
\label{eq:confelem}
\end{equation}
It can be shown that the energy momentum tensor $T(z)$ is given by 
\begin{equation}
T(z)=V(\psi_L,z)=-\frac{1}{2}:\partial_z X\cdot\partial_z X:
\end{equation}
and the Virasoro modes are given by 
\begin{equation}\label{eq:virgenuntwisted}
\begin{split}
L_n=\frac{1}{2}\sum_{m\in\mathbb{Z}}a_{n-m}\cdot a_{m},\quad n\neq 0,\\L_0=\frac{1}{2}p^2+\sum_{n\in\mathbb{N}}a_{-n}\cdot a_n.
\end{split}
\end{equation}
One can also show that the Virasoro generators satisfy the Virasoro algebra with central charge $c=r$.
The partition function of the theory is given by 
\begin{equation}\label{eq:torpartfunc}
Z_{\Lambda}(\tau)=q^{-c/24}\text{Tr}_{\mathscr{H}(\Lambda)} q^{L_0}=\frac{\Theta_{\Lambda}(\tau)}{\eta^r(q)},\quad q=e^{2\pi i\tau},
\end{equation}
where 
\begin{equation}
\Theta_{\Lambda}(\tau)=\sum_{\lambda\in\Lambda}q^{\frac{1}{2}\lambda^2}
\end{equation}
is the theta function for the lattice $\Lambda$ and $\tau\in\mathbb{H}:=\{z=x+iy\in\mathbb{C}:y>0\}$. The partition function is said to be modular invariant if it is invariant under the modular action of $\mathrm{SL}(2,\mathbb{Z})$:
\begin{equation}
Z_{\Lambda}\left(\frac{a\tau+b}{c\tau+d}\right)=Z_{\Lambda}(\tau),\quad \begin{pmatrix}
    a&b\\c&d
\end{pmatrix}\in\mathrm{SL}(2,\mathbb{Z}).    
\end{equation}
The partition function \eqref{eq:torpartfunc} is modular invariant if and only if the lattice $\Lambda$ is selfdual. 
\begin{remark}\label{rem:genpartfunc}
In physics, \eqref{eq:torpartfunc} defines what is called the torus partition function and is the partition function of a conformal field theory on the torus. In general, partition function can be defined on any general Riemann surface via a path integral. Due to conformal symmetry, the path integral on the torus simplifies to the above expression involving trace over the Hilbert space.    
\end{remark}
\begin{thm}\emph{\cite{Dolan:1994st}}\label{thm:lattice_VOA}
If $\Lambda$ is even, self-dual with dimension $r\equiv 0 (\bmod~24)$, then the Hilbert space $\mathscr{H}(\Lambda)$ with the vertex operators and structures defined above forms a Hermitian vertex operator algebra with central charge $c=r$ and modular invariant partition function. Moreover the Hermitian structure is given by 
\begin{equation}
    \overline{\phi}=e^{-i\pi L_0}\theta\phi
\label{eq:}
\end{equation}
where $\phi$ is a real linear combination of states of the form \eqref{psiflamda} and can be extended to complex linear combinations using the antilinearity of the above map and $\theta$ is an involution given by 
\begin{equation}
    \theta a^j_n\theta^{-1}=-a^j_n,\quad \theta|\lambda\rangle=|-\lambda\rangle,\quad \theta X^i\theta^{-1}=-X^i.
\label{eq:}
\end{equation}
\end{thm}
\begin{remark}
It can easily be checked that the basis states \eqref{psiflamda} have conformal dimensions 
\begin{equation}
    h_{\psi}=\frac{1}{2}\lambda^2+\sum_{a=1}^Mm_a
\label{eq:}
\end{equation}
which is an integer since $\Lambda$ is even.
\end{remark}
\subsection{Twisted Vertex Operator Algebra}\label{subsec:twistcft}
Given a modular invariant vertex operator algebra with a discrete subgroup $G$ of the automorphims group, one can construct a new vertex operator algebra by ``modding out'' or ``twisting'' or ``orbifolding'' by $G$ given that $G$ satisfies certain restrictions. This simply corresponds to gauging the symmetry, and the restriction is simply the condition that $G$ is not anomalous. For the lattice vertex operator algebra constructed above, the automorphism group $\mathbb{Z}_2$, consisting of the involution $\theta$ in Theorem \ref{thm:lattice_VOA} satisfies all the restrictions given that $\sqrt{2}\Lambda^*$ is even. Thus we can construct the twisted vertex operator algebra. The essential idea is to construct a representation of a sub-vertex operator algebra  which satisfies the hypotheses of Theorem \ref{thm:extensionofcft} and then consider the extended vertex operator algebra given in Theorem \ref{thm:extensionofcft}. We now describe the construction in detail. \\\\
Let $\Lambda,\mathscr{H}(\Lambda)$ and $\theta$ be as in Theorem \ref{thm:extensionofcft}. Then $\theta$ is an involution of $\mathscr{H}(\Lambda)$. We have the direct sum decomposition 
\begin{equation}
\mathscr{H}(\Lambda)=\mathscr{H}^+(\Lambda)\oplus \mathscr{H}^-(\Lambda)
\end{equation}
where 
\begin{equation}
\mathscr{H}^{\pm}(\Lambda):=\{\psi\in\mathscr{H}(\Lambda):\theta\psi=\pm \psi\}.
\end{equation}  
\begin{lemma}\emph{\cite{Dolan:1994st}}
$\mathscr{H}^+(\Lambda)$ is a sub-vertex operator algebra of $\mathscr{H}(\Lambda)$.
\end{lemma}
Note that $\mathscr{H}^+(\Lambda)$ is not a modular invariant VOA. To get a modular invariant VOA, we construct a representation of this VOA and then extend it as described in Section \ref{sec:cftextension}. To construct the representation of the VOA $\mathscr{H}^+(\Lambda)$, we first construct a Hilbert space $\mathscr{H}_T(\Lambda)$, 
known as the space of twisted states. To do so, we consider the bosonic oscillators $c^j_\ell$ indexed by $\ell\in\mathbb{Z}+\frac{1}{2}$ satisfying the same algebra as $a^j_n$:
\begin{equation}
[c^j_\ell,c^k_{m}]=k\delta^{jk}\delta_{\ell+m,0}.
\label{eq:heisalg}
\end{equation}
%

%%%%%%%%%%%%%%%%%%%%%%%%%%%%%%%%%%%%%%%% To be rewritten
To construct the space over which these oscillators can act, we start with the 
$\mathbb{Z}_2$ extension of the (abelian group) lattice $\Lambda$ with commutator function 
$(\lambda,\mu)\mapsto (-1)^{\lambda\cdot\mu}$. We choose  a 2-cocycle  $\varepsilon(\lambda,\mu)=\pm 1$ which the describes the extension:
\begin{equation}
    0\longrightarrow \mathbb{Z}_2\longrightarrow\widehat{\Lambda}\xrightarrow[]{~~-~~}\Lambda\longrightarrow 0
\end{equation}
where $-$ denotes the projection map. 
Consider the subgroup\footnote{It takes a few lines of algebra to verify that $K$ is a subgroup. It helps to choose (as one may) $\varepsilon$ to be a bihomomorphism.} $K$ of $\widehat{\Lambda}$ given by 
\begin{equation}
    K=\{a^2(-1)^{\frac{\bar{a}^2}{2}} : a\in\widehat{\Lambda}\}.
\end{equation}
Then \cite[eq. 10.3.18]{Frenkel:1988xz}
\begin{equation}
    \bar{K}\cong 2\Lambda
    \label{eq:Kbar}
\end{equation}
and we get the central extension 
\begin{equation}
   0\longrightarrow \mathbb{Z}_2\longrightarrow\widehat{\Lambda}/K\xrightarrow[]{~~-~~}\Lambda/2\Lambda\longrightarrow 0 
   \label{eq:centextheis}
\end{equation}
with commutator map $([\lambda],[\mu])\mapsto(-1)^{\lambda\cdot\mu}$. Clearly the commutator map does not depend on the choice of representatives. 
\begin{lemma}
The extension \eqref{eq:centextheis} is a Heisenberg extension.
\end{lemma}
\begin{proof}
We want to show that the commutator map of the extension is nondegenerate \cite[Section 15.5]{Moore:2022}. Consider the bilinear form 
\begin{equation}
\begin{split}
    b:&\Lambda/2\Lambda\times \Lambda/2\Lambda\longrightarrow\mathbb{Z}_2\\&([\lambda],[\mu])\longmapsto \lambda\cdot\mu\bmod 2.
\end{split}    
\end{equation}
Pick a basis $\{\lambda_i\}$ of $\Lambda$ and the corresponding basis $\{[\lambda_i]\}$ of $\Lambda/2\Lambda$. Then the matrix of the bilinear form $b$ in the basis $\{[\lambda_i]\}$ is the mod 2 reduction of the gram matrix of the basis $\{\lambda_i\}$ of $\Lambda$. Then since $\Lambda$ is unimodular, the determinant of the latter matrix is $\pm 1$. Hence the determinant of $b$ is  $1$ and hence $b$ is nondegenerate. This implies that the commutator map is nondegenerate, that is for every $[\lambda]\neq 0$, there exists $[\mu]$ such that $(-1)^{\lambda\cdot\mu}=-1$. 
\end{proof}
We denote by $\Gamma(\Lambda)$ the Heisenberg group $\widehat{\Lambda}/K$. 
Let $[\lambda]\mapsto e_{[\lambda]}$ be a section of this extension. For a function $u:\Lambda/2\Lambda\longrightarrow\{\pm 1\}$, $[\lambda]\mapsto\widetilde{e}_{[\lambda]}:=u([\lambda])e_{[\lambda]}$ defines another section of the extension. Moreover the cocycle changes by a coboundary:
\begin{equation}
 \label{eq:}\widetilde{e}_{[\lambda]}\widetilde{e}_{[\mu]}=\widetilde{\varepsilon}(\lambda,\mu)\widetilde{e}_{[\lambda+\mu]}   
\end{equation}
where 
\begin{equation}\label{eq:cocyclecoboundaryrel}
\widetilde{\varepsilon}(\lambda,\mu)=\frac{u([\lambda])u([\mu])}{u([\lambda+\mu])}\varepsilon(\lambda,\mu).    
\end{equation}
Using this ``gauge freedom'' we  
can choose\footnote{Note that \eqref{eq:Kbar} implies that $e_{[0]}=\pm\mathds{1}$.}  $\varepsilon(0,\lambda)=1$ for all $\lambda$ and  $e_{[0]}=\mathds{1}\in\Gamma(\Lambda)$.
To simplify notations, we write $e_{[\lambda]}=:e_{\lambda}$. We can then write 
\begin{equation}
\begin{split}
&\Gamma(\Lambda)=\{\pm e_{\lambda}:\lambda\in\Lambda\}\\
&e_{\lambda}e_{\mu}=(-1)^{\lambda\cdot\mu}e_{\mu}e_{\lambda}=\varepsilon(\lambda,\mu)e_{\lambda+\mu},\quad e_{\lambda}^2=(-1)^{\frac{1}{2}\lambda^2},
\end{split}
\label{eq:gammamatalg}
\end{equation} 
The last relation $e_\lambda^2=(-1)^{\lambda^2/2}$ comes from the fact that $K$ is modded out in the Heisenberg group $\Gamma(\Lambda)$. Note that these relations do not depend on the choice of representatives of $[\lambda],[\mu]$ when calculating $\lambda\cdot\mu$. 
\begin{lemma}
The cocycle $\varepsilon$ satisfies 
\begin{equation}
    \varepsilon(\lambda,\lambda)=(-1)^{\frac{\lambda^2}{2}}.
\end{equation}
\end{lemma}
\begin{proof}
To prove this, we use a result of \cite{Frenkel:1988xz}. For a central extension of an elementary abelian 2-group\footnote{For a prime number $p$, an elementary abelian $p$-group is a direct product of cyclic groups of order $p$.} $E$ by $\mathbb{Z}_2$
\begin{equation}
0\longrightarrow \mathbb{Z}_2\longrightarrow\widehat{E}\xrightarrow[]{~~-~~}E\longrightarrow 0    
\end{equation}
with commutator map $(x,y)\mapsto (-1)^{c(x,y)}$, the \textit{squaring map} $s:E\longrightarrow \mathbb{Z}_2$ is defined by 
\begin{equation}
    a^2=(-1)^{s(\bar{a})},\quad a\in \widehat{E}.
\end{equation}
The bilinear form corresponding to $s$ is the commutator map:
\begin{equation}
    c(x,y)=s(x+y)-s(x)-s(y).
\end{equation}
Then by \cite[Proposition 5.3.3]{Frenkel:1988xz}, for any central extension $\widehat{E}$ of an abelian 2-group $E$ by $\mathbb{Z}_2$ with squaring map $s$, the cocycle satisfies 
\begin{equation}
    (-1)^{s(x)}=\varepsilon(x,x).
\end{equation}
Moreover, by \eqref{eq:cocyclecoboundaryrel} cohomologous cocycles do not change this relation. For the central extension \eqref{eq:centextheis} at hand, it is easy to check that the squaring map is 
\begin{equation}
    [\lambda]\mapsto\frac{\lambda^2}{2}\bmod 2.
\end{equation}
The required relation thus holds.
\end{proof}
Let $\mathcal{S}(\Lambda)$ be the carrier space of the unique (up to isomorphism) $2^{r/2}$-dimensional irreducible unitary representation of $\Gamma(\Lambda)$. Such a representation can be constructed using Dirac gamma matrices, see  \cite[Appendix C]{Dolan:1989vr} for an explicit construction. When the representation is real we choose it to be a real vector space of dimension $2^{r/2}$. The vector space $\mathcal{S}$ of the introduction is $\mathcal{S}(\Lambda_{\mathrm{L}})$ in our present notation.

%%%%%%%%%%%%%%%%%%%%%%%%%%%%%%%%%%%%%%%%%%%%%%%%%

The algebra of oscillators $\{c^j_\ell\}$ is represented on a Fock space $\mathscr{F}_T$ with a unique vacuum line spanned by $\vert 0 \rangle_T$, so  $c^j_\ell\vert 0 \rangle_T=0$ for $\ell>0$. 
Then the twisted sector Hilbert space is the (completion of): 
\begin{equation}
    \mathscr{H}_T(\Lambda)=\mathscr{F}_T\otimes_{\mathbb{R}}\mathcal{S}(\Lambda) ~ . 
\end{equation}
%
%
%and $\chi_0\in\mathcal{S}(\Lambda).$
%We then construct $\mathscr{H}_T(\Lambda)$ as the span of the Fock spaces $\mathscr{F}$ given by
%\begin{equation}
%    \mathscr{F}:=\bigsqcup_{\chi_0\in\mathcal{S}(\Lambda)}\mathscr{F}_{\chi_0}
%\label{eq:focksptwist}
%\end{equation}
%where $\mathscr{F}_{\chi_0}$ is the set of states of the form
%
A general Fock space state is a linear combination of vectors of the form: 
\begin{equation}
\chi=\left(\prod_{a=1}^N c^{j_a}_{-k_a}\right)\vert 0\rangle_T \otimes \chi_0,
\label{eq:chitwistcft}
\end{equation}
where $N\in\mathbb{Z}_+$ and $k_a=m_a+\frac{1}{2}$ with $m_a\in\mathbb{Z}_+$
and $1\leq j_a\leq r$, 
and $\chi_0 \in \mathcal{S}(\Lambda)$.
 We endow $\mathscr{H}_T(\Lambda)$ with an Hermitian structure such that $(c^{j})^{\dagger}_k=c^j_{-k}$. The Virasoro modes are given by 
\begin{equation}
\begin{split}
L_{n}^T=\frac{1}{2} \sum_{k\in\mathbb{Z}+\frac{1}{2}}: c_{k} \cdot c_{n-k}:+\frac{r}{16} \delta_{n, 0}.
\end{split}
\end{equation}
In particular, the ground states $\mathcal{S}(\Lambda)$ have conformal weight $r/16$ and the conformal weight of the state $\chi$ given by \eqref{eq:chitwistcft} is 
\begin{equation}
h_{\chi}=\frac{r}{16}+\sum_{a=1}^Mk_a.
\label{eq:confwttwist}
\end{equation}
To construct the vertex operators corresponding to states in $\mathscr{H}_T(\Lambda)$, we start with the an expression analogous to the untwisted vertex operator \eqref{eq:veropuntwisted}. We  will then need to modify the naive expression in order to satisfy the translation property \eqref{eq:transprop} and the conformal transformation property \eqref{eq:LncommVpsiz}. Thus, we consider the fields
\begin{equation}
R^j(z)=i\sum_{k\in\mathbb{Z}+\frac{1}{2}}\frac{c^j_k}{k}z^{-k}
\end{equation}
and  we first define the naive vertex operator $V^0_T(\psi,z)$, corresponding to the state $\psi$ in \eqref{psiflamda}, to be: 
\begin{equation}
V^0_T(\psi,z) := \left(\prod_{a=1}^{M} \frac{i}{\left(m_{a}-1\right) !} \frac{d^{m_{a}} R^{j_a}}{d z^{m_{a}}}(z)\right) \exp~ \{i \lambda \cdot R(z)\}: e_{\lambda},
\end{equation}
where the normal ordering $::$ is defined analogous to the untwisted case. This operator fails to satisfy the $L_n^T$ commutator \eqref{eq:LncommVpsiz}. Including a factor of $(4z)^{-\frac{\lambda^2}{2}}$ rectifies the commutator with $L_n^T$ \cite{Dolan:1989vr} but now the translation property \eqref{eq:transprop} is not satisfied. To remedy this, we introduce the vertex operator 
\begin{equation}
V_T(\psi,z) :=V_T^0(e^{A(-z)}\psi,z),
\label{eq:twistverop}
\end{equation}
and try to solve for $A(z)$ demanding the required properties. The constraints imposed on $A(z)$ by the $L_n^T$ commutators and the translation property \eqref{eq:transprop} is solved by the operator \cite{Dolan:1989vr,Frenkel:1988xz}
\begin{equation}
A(z)=\frac{1}{2} \sum_{n, m \geq 0 \atop m+n>0}\left(\begin{array}{c}
-\frac{1}{2} \\
m
\end{array}\right)\left(\begin{array}{l}
\frac{1}{2} \\
n
\end{array}\right) \frac{(-z)^{-m-n}}{m+n} a_{m} \cdot a_{n}-\frac{1}{2} a_{0} \cdot a_{0} \ln (-4 z),
\label{eq:A(z)exp}
\end{equation}
where $a_0^i=p^i$.
Note that the last term corresponds to the factor $(4z)^{-\frac{\lambda^2}{2}}$ which is needed for the $L_n^T$ commutator \eqref{eq:LncommVpsiz}.
\begin{remark}\label{rem:deponsection}
\normalfont
The definition of the vertex operator $V_T(\psi,z)$ depends on the choice of the section in \eqref{eq:gammamatalg}. Indeed if we pick a different section $\lambda\mapsto e'_{\lambda}$, then the new cocycle $\varepsilon'(\lambda,\mu)$ is related to the old cocycle by a coboundary $b:\Lambda/2\Lambda\mapsto\{\pm 1\}$:
\begin{equation}
    \varepsilon'(\lambda,\mu)=\varepsilon(\lambda,\mu)\frac{b(\lambda)b(\mu)}{b(\lambda+\mu)}.
\end{equation}
The unitary operators $e'_{\lambda}$ and $e_{\lambda}$ acting on $\mathcal{S}(\Lambda)$ 
 are related by a unitary operator $S$ \cite[eq. 6.10]{Dolan:1994st}
\begin{equation}
    Se'_{\lambda}S^{\dagger}=b(\lambda)e_{\lambda}.
\end{equation}
\end{remark}
For the conformal weights \eqref{eq:confwttwist} to be integers, we need both the terms to be either strictly half-integers or integers simultaneously. This immediately implies that $r$ has to be a multiple of 8 \cite[Chapter 6]{Hamilton:2017gbn},\cite{Moore:2018}. When the second term is an integer (half-integers), the state $\chi$ is obtained by acting an even (odd) number of oscillators on a twisted ground state. To quantify this, we extend the action of the $\mathbb{C}$-linear involution $\theta$ to $\mathscr{H}(\Lambda)\oplus\mathscr{H}_T(\Lambda)$ by 
\begin{equation}
\theta c^i_k\theta^{-1}=-c^i_k,\quad \theta\chi_0=(-1)^{r/8}\chi_0,\quad \theta R^j\theta^{-1}=-R^j. 
\label{eq:thetatwisteddef}
\end{equation}
It is now clear that when $r$ is a multiple of 16, in which case the first term is an integer, $h_{\chi}$ can be integer only when $\chi$ is obtained by acting an even number of oscillators on a twisted ground state. Since the twisted ground states are even under $\theta$ for $r$ a multiple of 16, the state $\chi$ is also even under $\theta$. On the other hand if $r\equiv 8\bmod ~16$, the twisted ground states are odd under $\theta$ and if we take $\chi$ to be an odd number of oscillators acting on a twisted ground state, $\chi$ will be even under $\theta$. In this case the two terms in $h_{\chi}$ in \eqref{eq:confwttwist} are strictly half-integers. Thus we see that the only way the conformal weights can be integers is when the following two conditions are satisfied:
\begin{enumerate}
\item The dimension $r$ is a multiple of $8$.
\item The state $\chi\in\mathscr{H}_T^+(\Lambda)$ where $\mathscr{H}_T^+(\Lambda)$ is  $\theta=1$ subspace with the action of $\theta$ as defined in \eqref{eq:thetatwisteddef}.
\end{enumerate}
We now need to put an Hermitian structure on the twisted sector, see Definition \ref{def:hermrep}. Choose an orthonormal basis of $\mathcal{S}(\Lambda)$ and let $\gamma_\lambda$ be the matrix of $e_\lambda$ with respect to this basis. Let $M$ be a symmetric, unitary matrix satisfying 
\begin{equation}
M\gamma_{\lambda}^*=\gamma_{\lambda}M.
\label{eq:mdef}
\end{equation}
Such a matrix $M$ exists because the unitary representation of  the Clifford algebra is real in dimensions divisible by 8 and the unique (up to isomorphism) unitary representation of the Heisenberg extension \eqref{eq:gammamatalg} can be constructed using Dirac gamma matrices. Moreover one can choose $M$ such that $M^2=\mathds{1}$ when $r\equiv 0\bmod 8$, see \cite[Table 1]{Kuusela:2019iok} for more details and \cite[Appendix C]{Dolan:1989vr} for an explicit construction. Define the action of $M$ on $\mathscr{H}_T(\Lambda)$ by:
\begin{equation}
M:\chi=\left(\prod_{a=1}^N c^{j_a}_{-k_a}\right)\chi_0\longmapsto M(\chi)=\left(\prod_{a=1}^N c^{j_a}_{-k_a}\right)M\chi_0^*,\quad  \chi_0\in\mathcal{S}(\Lambda)   
\end{equation}
where $\chi_0^*$ is defined with respect to the chosen basis of $\mathcal{S}(\Lambda)$.
We now have the following theorem.
\begin{thm}\emph{\cite{Dolan:1989vr,Dolan:1994st}}
Let $\Lambda$ be an even lattice and $\sqrt{2}\Lambda^*$ be even. Then the following statements are true. 
\begin{enumerate}
\item The vertex operators $V_T(\psi,z)$ and the Hilbert space $\mathscr{H}_T^+(\Lambda)$ define a real, Hermitian representation of the vertex operator algebra $\mathscr{H}^+(\Lambda)$ with the conjugation map on the twisted sector $\mathscr{H}_T^+(\Lambda)$ given on states of the form \eqref{eq:chitwistcft} by
\begin{equation}
    \overline{\chi}=e^{-i\pi L_0^T} \theta M(\chi),
\label{eq:}
\end{equation}
with extension to all states of the twisted sector defined by antilinearity.
\item The extended vertex operator algebra $\widetilde{\mathscr{H}}(\Lambda):=\mathscr{H}^+(\Lambda)\oplus\mathscr{H}_T^+(\Lambda)$ is a real, Hermitian vertex operator algebra and is modular invariant if $\Lambda$ is self-dual. This is called the $\mathbb{Z}_2$-orbifold model.
\end{enumerate} 
\label{thm:latvoaext}
\end{thm}
\begin{remark}\label{rem:deponsectioniso}
The extended vertex operator algebra $\widetilde{\mathscr{H}}(\Lambda)$ is independent of the choice of section $e_\lambda$ used to define the operators $V_T(\psi,z)$ in \eqref{eq:twistverop} up to isomorphism. Indeed, if we choose a different section $\lambda\mapsto e'_\lambda$ of the Heisenberg extension and denote the operators in \eqref{eq:twistverop} with $e_\lambda$ replaced by $e'_\lambda$ by $V_T'(\psi,z)$, then the VOA $(\widetilde{\mathscr{H}}(\Lambda),\widetilde{V})$ is isomorphic to the VOA $(\widetilde{\mathscr{H}}(\Lambda),\widetilde{V}')$ where for $\psi\in\mathscr{H}^+(\Lambda),$ and $\chi\in\mathscr{H}^+_T(\Lambda)$
\begin{equation}
\begin{split}
&\widetilde{V}(\psi,z)=\begin{pmatrix}
V(\psi,z)&0\\0&V_T(\psi,z)
\end{pmatrix},\quad \widetilde{V}(\chi,z)=\begin{pmatrix}
0&\overline{W}(\chi,z)\\W(\chi,z)&0
\end{pmatrix},\\&\widetilde{V}'(\psi,z)=\begin{pmatrix}
V(\psi,z)&0\\0&V'_T(\psi,z)
\end{pmatrix},\quad \widetilde{V}'(\chi,z)=\begin{pmatrix}
0&\overline{W}'(\chi,z)\\W'(\chi,z)&0
\end{pmatrix}~,
\end{split}
\end{equation}
and $\overline{W}'(\chi,z)$ is defined as in \eqref{eq:defW} using $V_T'$. The explicit isomorphism $u:\widetilde{\mathscr{H}}(\Lambda)\longrightarrow \widetilde{\mathscr{H}}(\Lambda)$ is given by 
\begin{equation}
    \begin{pmatrix}
        |\lambda\rangle\\|\chi\rangle
    \end{pmatrix}\longmapsto \begin{pmatrix}
        b(\lambda)|\lambda\rangle\\S|\chi\rangle
    \end{pmatrix}~,\quad |\chi\rangle\in \mathcal{S}(\Lambda)~,
\end{equation}
where $b(\lambda),S$ is as in Remark \ref{rem:deponsection} and $u$ commutes with $a_n^i,c_k^i$. 
\end{remark}
Applying this theorem to the Leech lattice defines the Monster module, 
also known as the FLM module $\mathscr{H}_{\text{FLM}}:=\widetilde{\mathscr{H}}(\Lambda_{\text{L}})$, see \eqref{eq:flmmod} above. 
\begin{remark}\label{rem:Mi}
Note that we can take $M=\mathds{1}$ in \eqref{eq:mdef} by choosing the spin representation $\mathcal{S}(\Lambda)$ of the Heisenberg group $\Gamma(\Lambda)$ to be real. In that case the representation matrices $\gamma_\lambda$
for $e_{\lambda}$ acting on the  twisted sector Hilbert space  
\begin{equation}
    \mathscr{H}_T(\Lambda)\cong\mathscr{F}_T\otimes_{\mathbb{R}}\mathcal{S}(\Lambda)
\end{equation}
are real, i.e., $\gamma_{\lambda}^*=\gamma_{\lambda}$.
%where $\mathscr{F}_T$ is the Fock space generated by the complex linear combinations of states %\eqref{eq:focksptwist} and $\mathcal{S}(\Lambda)$ is
%the real $2^{12}$-dimensional representation of $\Gamma(\Lambda)$. 
%
We will use this to simplify the discussion in Subsection \ref{subsec:ConstructTF}. 
\end{remark}

\section{$\mathcal{N}=1$ Superconformal Vertex Operator Algebra And Spin Lifts }\label{sec:scft}
We now describe a superconformal vertex operator algebra (SCVOA). Roughly speaking, an $\mathcal{N}=1$ SCVOA is a super-VOA with, in addition, $\mathcal{N}=1$ conformal supersymmetry. 
\begin{defn}
An $\mathcal{N}=1$ superconformal vertex operator algebra (SCVOA) is a tuple $\left(\mathscr{H}, \mathscr{F}, \mathbf{V},|0\rangle, \psi_{L}\right)$ where $\mathscr{H}=\mathscr{H}_{\bar{0}}\oplus\mathscr{H}_{\bar{1}}$ is a super Hilbert space of states, $\mathscr{F}=\mathscr{F}_{\bar{0}}\oplus\mathscr{F}_{\bar{1}}$ with $\mathscr{F}_{\alpha}\subset\mathscr{H}_{\alpha}$ a dense subspace called the \emph{Fock space} and $\mathbf{V}$ a set of linear operators called \emph{vertex operators} $V(\psi, z)$ which are meromorphic in $z$ and linear in $\psi$. They are in one-to-one correspondence with the states $\psi \in \mathscr{F}$. We will denote the degree of $\psi\in\mathscr{H}_\alpha$ with respect to the $\mathbb{Z}_2$-grading by $|\psi|$ and call it the \textit{parity} of $\psi$. There are three special states in $\mathscr{F}$, the \emph{vacuum} $|0\rangle\in\mathscr{F}_{\bar{0}}$, the \emph{conformal state} $\psi_{L} \in\mathscr{F}_{\bar{0}}$ and the \emph{superconformal state} $\psi_{G}\in\mathscr{F}_{\bar{1}}$. The theory must satisfy the properties listed below, and is said to be a \emph{Hermitian $\mathcal{N}=1$ superconformal vertex operator algebra} if it satisfies in addition the last property (7).
\begin{enumerate}
\item The \emph{moments} of the vertex operator $V(\psi_L,z)$ and $V(\psi_G,z)$ are given\footnote{Whenever we have strictly half-integral powers of a complex number, we will choose the principal branch of square root to define it: in the principal branch, we choose $\sqrt{-1}=e^{i\pi/2}$, the square root of a complex number  in the upper half plane is mapped to the first quadrant
and the lower half plane is mapped to fourth quadrant.} by
\begin{equation}
\begin{split}
T_B(z):=V\left(\psi_{L}, z\right)=\sum_{n \in \mathbb{Z}} L_{n} z^{-n-2}\\
T_{F}(z):=V(\psi_G,z)=\sum_{r}G_rz^{-r-\frac{3}{2}}
\end{split}
\end{equation}
and satisfy the $\mathcal{N}=1$ superconformal algebra 
\begin{equation}
\begin{split}
\left[L_{m}, L_{n}\right]=(m-n) L_{m+n}+\frac{\widehat{c}}{8} m\left(m^{2}-1\right) \delta_{m,-n}\\
[G_r,L_n]=\left(r-\frac{n}{2}\right)G_{r+n}\\
\{G_r,G_s\}=2L_{r+s}+\frac{\widehat{c}}{2}\left(r^2-\frac{1}{4}\right)\delta_{r,-s}\\
[L_m,\widehat{c}]=[G_r,\widehat{c}]=0
\end{split}
\label{eq:NSsuperalg}
\end{equation}
with $L_{n}^{\dagger}=L_{-n}$ and $L_{n}|0\rangle=0, n \geq-1$. Equivalently, the operators $T_B$ and $T_F$ satisfy the OPE
\begin{equation}\label{eq:opetbtf}
    \begin{split}
        &T_{B}(z) T_{B}\left(w\right) \sim \frac{3 \widehat{c} / 4}{\left(z-w\right)^{4}}+\frac{2}{\left(z-w\right)^{2}} T_{B}\left(w\right)+\frac{1}{\left(z-w\right)} \partial_{w} T_{B}\left(w\right)\\
&T_{B}\left(z\right) T_{F}\left(w\right) \sim \frac{3 / 2}{\left(z-w\right)^{2}} T_{F}\left(w\right)+\frac{1}{\left(z-w\right)} \partial_{w} T_{F}\left(w\right)\\&T_{F}\left(z\right) T_{F}\left(w\right) \sim \frac{\widehat{c} / 4}{\left(z-w\right)^{3}}+\frac{1 / 2}{\left(z-w\right)} T_{B}\left(w\right)
    \end{split}
\end{equation}
where $\sim$ denotes that the relation is true up to a holomorphic term.  When the operators $G_r$ are integer moded, \textit{i,e} $r\in\mathbb{Z}$, the superconformal algebra is called the \textit{Ramond} (R) algebra and for half-integer moding, $r\in\mathbb{Z}+\frac{1}{2}$, the superconformal algebra is called the \textit{Neveu-Schwartz} (NS)  algebra. 
\item For homogenous $\psi$, the moments of vertex operators $V(\psi,z)$ are linear operators of parity $|\psi|,$ where $|\psi|$ is the $\mathbb{Z}_2$-grading of $\psi\in\mathscr{H}$. 
\item The vertex operators satisfy
\begin{equation}
(i)\quad V(\psi, z)|0\rangle=e^{z L_{-1}} \psi\quad (ii) \quad V(\psi,z)V(\phi,w)=(-1)^{|\psi||\phi|}V(\phi,w)V(\psi,z)~,
\label{eq:VL-1,locrel}
\end{equation}
The relation (ii) is defined for vertex operators corresponding to homogeneous states and is called the \emph{locality relation}, see Definition \ref{defn:cft} for details. We say that the vertex operator $V(\psi,z)$ is local with degree $|\psi|$ with respect to $\mathbf{V}$. 
\item The operator $w^{L_0}:=e^{L_0\ln w}$ acts locally with respect to $\mathbf{V}$, that is $w^{L_0}V (\psi,z)w^{-L_0}$ is local with parity $|\psi|$ with respect to $\mathbf{V}$.
\item  $L_0$ is diagonalisable and its spectrum is bounded below. The eigenvalues $h_{\psi}$ of eigenstates $\psi$ of $L_0$ are called \emph{conformal weights}.
\item The vacuum is the only $\mathfrak{su}(1,1)$ invariant state in the theory, that is the only state annihilated by $L_0$ and $L_{\pm 1}$.
\item For all $\psi\in\mathscr{F}$, the vertex operator $V(e^{z^*L_1}z^{*-2L_0}\psi,1/z^*)^{\dagger}$ is local with parity $|\psi|$ with respect to $\mathbf{V}.$
\end{enumerate}
\label{defn:scft}
\end{defn}
\begin{remark}
If we only require the existence of a conformal vector and the associated Virasoro algebra in the above definition, it is called a \textit{vertex operator superalgebra} (SVOA). So an SCVOA is in particular an SVOA.       
\end{remark}
\begin{defn}\label{def:scvoa_hom}
We say that two SCVOAs $\mathscr{H}$ and $\mathscr{H}'$ with Fock spaces $\mathscr{F}$ and $\mathscr{F}'$ respectively and corresponding vertex operators $V(\psi,z)$ and $V'(\psi',z)$ are isomorphic if there exists a unitary map $u:\mathscr{H}\longrightarrow\mathscr{H}'$ such that 
\begin{equation}
V'(u\psi,z)=uV(\psi,z)u^{-1},\quad \psi\in\mathscr{F}~,\quad u\psi_G=\psi_G'~,    
\end{equation}
where $\psi_G,\psi_G'$ are the superconformal states in $\mathscr{H}$ and $\mathscr{H}^{\prime}$ respectively. 
An isomorphism  $u:\mathscr{H}\longrightarrow\mathscr{H}$ which preserves $|0\rangle,\psi_L,\psi_G$ is called an automorphism of $\mathscr{H}$. 
\end{defn}
Note that an isomorphism $u:\mathscr{H}\longrightarrow\mathscr{H}'$ of SCVOAs (or simply VOAs) automatically satisfies \cite[Proposition 2.14]{Dolan:1994st}
\begin{equation} u|0\rangle=\left|0^{\prime}\right\rangle,\quad u \psi_L=\psi_L^{\prime}~,  
\end{equation}
where $|0\rangle,\left|0^{\prime}\right\rangle$ are the vacuum states and $\psi_L,\psi_L'$ are the conformal states of the two SCVOAs.  

As will be shown in the proposition below, the conformal weights of an $\mathcal{N}=1$ superconformal vertex operator algebra are elements of $\frac{1}{2}\mathbb{Z}$. Note that the locality relation forces the singularities in the OPEs of vertex operators to be integral powers of $(z-w)$.\\\\ In superconformal field theories in physics, one often encounters OPEs with half-integer power singularities. This clearly violates the locality relation. Indeed, the meromorphic function $(z-w)^{-\frac{1}{2}}$ defined on $|z|>|w|$ when analytically continued to $|w|>|z|$ picks up a fourth root of unity and violates the locality relation above which only allows $\pm 1$ factor. Such fields are called non-local fields. A superconformal field theory is a spin CFT (see Footnote \ref{foot:VOA_CFT_spin}). The states in the spin CFT associated to the nonbounding circle\footnote{\label{foot:spin_st_circle} Spin structures on a circle are classified by $H^1(S^1,\mathbb{Z}_2)\cong\mathbb{Z}_2$. The circle with spin structure corresponding to the trivial class is alled nonbounding circle while the circle with the spin structure corresponding to the nontrivial class is called bounding circle.} are called the \textit{Ramond} (R) sector states while the states associated to the bounding circle are called the \textit{Neveu-Schwartz} (NS) sector states. The subspaces of R and NS states in $\mathscr{H}$ is denoted by $\mathscr{H}_{\text{R}}$ and $\mathscr{H}_{\text{NS}}$ respectively. The corresponding vertex operators are called the R fields and NS fields respectively. Consequently the R fields have square root singularity in their OPE with the superconformal field $T_F$ and the NS fields have integer power singularities in their OPE with the superconformal field $T_F$. We will show that the ``Beauty and the Beast'' SCFT has a subspace which forms an SCVOA while the full space is a non-local superconformal field theory in the sense that it contains non-local fields. \\The vertex operators can be Laurent-expanded as before
\begin{equation}
    V(\psi,z)=\sum_{n}V_n(\psi)z^{-n-h_{\psi}}
\label{eq:}
\end{equation}
where the summation index is such that $n+h_{\psi}\in\mathbb{Z}$ for R vertex operators and  $n+h_{\psi}\in\mathbb{Z}+\frac{1}{2}$ for NS vertex operators. Note that if $V(\psi,z)$ is an NS vertex operator then $\psi\in\mathscr{H}_{\Bar{1}}$. The moments again satisfy 
\begin{equation}
    V_{-h_{\psi}}(\psi)|0\rangle=\psi,\quad V_{-n}(\psi)|0\rangle=0, \quad n>h_{\psi}.
\label{eq:}
\end{equation}
A state $\psi\in\mathscr{H}$ with conformal weight $h$ is called a \textit{superconformal primary} if 
\begin{equation}
    L_0\psi=h\psi,\quad L_n\psi=0,\quad G_r\psi=0,\quad n,r>0.
\end{equation}
\begin{prop}
\begin{enumerate}
\item If $U(z)$ satisfies $U(z)|0\rangle=e^{zL_{-1}}\phi$ and is local with degree $|\phi|$ with respect to $\mathbf{V}$, then $U(z)=V(\phi,z)$. In particular 
\begin{equation}
    \frac{d}{dz}V(\psi,z)=V(L_{-1}\psi,z).
\label{eq:}
\end{equation}
\item The duality relation $V(\psi,z)V(\phi,w)=V(V(\psi,z-w)\phi,w)$ is true.
\item Skew-symmetry: $V(\psi,z)\phi=(-1)^{|\psi||\phi|}e^{zL_{-1}}V(\phi,-z)\psi.$ 
\item For each $\psi\in\mathscr{F}$, the locality of $V(e^{z^*L_1}z^{*-2L_0}\psi,1/z^*)^{\dagger}$ implies that we have a grading-preserving antilinear map
\[
\psi\longmapsto \overline{\psi},
\]  
where 
\begin{equation}
    V(e^{z^*L_1}z^{*-2L_0}\psi,1/z^*)^{\dagger}=V(\overline{\psi},z)
\label{eq:}
\end{equation}
so that 
\begin{equation}
\overline{\psi}=\lim_{z\to 0}V(e^{z^*L_1}z^{*-2L_0}\psi,1/z^*)^{\dagger}|0\rangle.
\end{equation}
Moreover, for a quasi-primary vertex operator, the moments satisfy $V_n(\psi)^{\dagger}=V_{-n}(\overline{\psi}).$
\item The conformal weights are non-negative.
\item For any $\phi,\psi\in\mathscr{F}$, the following holds:
\begin{enumerate}
\item if $L_0\psi=h_{\psi}\psi$ then $L_0\overline{\psi}=h_{\psi}\overline{\psi}$.
\item $\overline{L_{-1}\psi}=-L_{-1}\psi$.
\item $\overline{\overline{\psi}}=\psi$.
\item $f_{\phi_1\phi_2\phi_3}=(-1)^{|\phi_1||\phi_2|+|\phi_2||\phi_3|+|\phi_3||\phi_1|}e^{-i\pi(h_{\phi_1}+h_{\phi_2}+h_{\phi_3})}(f_{\overline{\phi}_1\overline{\phi}_2\overline{\phi}_3})^{*}$ where $\phi_{1}, \phi_{2}, \phi_{3} \in \mathscr{F}$ with $L_{0} \phi_{j}=h_{j} \phi_{j}$ for $1 \leq j \leq 3$ and
\begin{equation}
    f_{\phi_{1} \phi_{2} \phi_{3}}=\left\langle\overline{\phi}_{1}\left|V\left(\phi_{2}, 1\right)\right| \phi_{3}\right\rangle.
\label{eq:fabcdefcor}
\end{equation}
\item $\overline{L_1\psi}=-L_1\overline{\psi}.$
\item the spectrum of $L_0$ is non-negative half-integers $\frac{1}{2}\mathbb{Z}$.
\item $\langle \overline{\psi}|\phi\rangle=\langle\overline{\phi}|\psi\rangle.$
\end{enumerate}
\item The \emph{operator product expansion} of any two vertex operators  $V(\psi,z)$ and $V(\phi,w)$ corresponding to eigenstates of $L_0$ is given by
\begin{equation}
\begin{split}
    V(\psi,z)V(\phi,w)&=\sum_{n}(z-w)^{n-h_{\psi}-h_{\phi}}V(\phi_n,w)\\&=\sum_{n}(z-w)^{-h_{\phi_n}-h_{\psi}-h_{\phi}}V(\phi_n,w)
\end{split}
\label{eq:scftgenope}
\end{equation}
where
\begin{equation}
    \phi_n:=V_{h_{\phi}-n}(\psi)\phi.
\label{eq:}
\end{equation}
\end{enumerate}\label{prop:propscft}
\begin{proof}
We will only prove (1), (3) and (6) since the proof of others is similar to the VOA case and can be found in \cite{Dolan:1994st}.
\\\\
(1) We have 
$$
\begin{aligned}
U(z) e^{\zeta L_{-1}} \psi &=U(z) V(\psi, \zeta)|0\rangle=(-1)^{|\phi||\psi|}V(\psi, \zeta) U(z)|0\rangle \\
&=(-1)^{|\phi||\psi|}V(\psi, \zeta) e^{z L_{-1}} \phi=(-1)^{|\phi||\psi|}V(\psi, \zeta) V(\phi, z)|0\rangle \\
&=V(\phi, z) V(\psi, \zeta)|0\rangle=V(\phi, z) e^{\zeta L_{-1}} \psi
\end{aligned}
$$
Thus, taking $\zeta\to 0$, we deduce $U ( z ) = V(\phi,z)$. The second part follows from the fact that $\frac{d}{dz} V(z,\psi)$ is local with respect to $\mathbf{V}$ and satisfies 
\begin{equation}
    \frac{d}{dz}V(z,\psi)|0\rangle=e^{zL_{-1}}L_{-1}\psi.
\end{equation}\\\\
(3) Using duality and (3)(i) of Definition \ref{defn:scft}, we have 
\begin{equation}
V(\psi,z)V(\phi,w)|0\rangle=V(V(\psi,z-w)\phi,w)|0\rangle=e^{wL_{-1}}V(\psi,z-w)\phi.
\end{equation}
Using the locality (3) (ii) of Definition \ref{defn:scft} first, we obtain 
\begin{equation}
V(\psi,z)V(\phi,w)|0\rangle=(-1)^{|\phi||\psi|}V(\phi,z)V(\psi,w)|0\rangle=(-1)^{|\phi||\psi|}e^{zL_{-1}}V(\phi,z-w)\psi.
\end{equation}
Taking $w\to 0$ in both expressions and comparing, we get the claimed skew-symmetry.\\\\
(6) The proof of (a), (b), (c), (e) and (g) is exactly as in \cite[Proposition 2.9]{Dolan:1994st}. We now prove (d). We have  
\[
\begin{split}
f_{\phi_{1} \phi_{2} \phi_{3}} &=\left\langle\overline{\phi}_{1}\left|V\left(\phi_{2}, 1\right)\right| \phi_{3}\right\rangle  \\
&=(-1)^{|\phi_2||\phi_3|}\left\langle\overline{\phi}_{1}\right| e^{L_{-1}} V\left(\phi_{3},-1\right)\left|\phi_{2}\right\rangle \quad \text { (by skew-symmetry) } \\
&=(-1)^{|\phi_2||\phi_3|}\left\langle\phi_{2}\left|V\left(\phi_{3},-1\right)^{\dagger} e^{L_{1}}\right| \overline{\phi}_{1}\right\rangle^{*}  \\
&=(-1)^{|\phi_2||\phi_3|-2h_{\phi_3}}\left\langle\phi_{2}\left|V\left(e^{-L_{1}} \overline{\phi}_{3},-1\right) e^{L_{1}}\right| \overline{\phi}_{1}\right\rangle^{*} \text { (by (4) and (b),(c) of (6))} \\&=(-1)^{|\phi_2||\phi_3|-2h_{\phi_3}+|\phi_1||\phi_3|}\left\langle\phi_{2}\left|e^{-L_{-1}} V\left(e^{L_{1}} \overline{\phi}_{1}, 1\right) e^{-L_{1}}\right| \overline{\phi}_{3}\right\rangle^{*}\quad\text{(by skew-symmetry)}\\&=(-1)^{|\phi_2||\phi_3|-2h_{\phi_3}+|\phi_1||\phi_3|}\left\langle\phi_{2}\left|e^{-L_{-1}} V\left(\phi_{1}, 1\right)^{\dagger} e^{-L_{1}}\right| \overline{\phi}_{3}\right\rangle^{*}\\
&=(-1)^{|\phi_2||\phi_3|-2h_{\phi_3}+|\phi_1||\phi_3|+|\phi_1||\phi_2|}\left\langle\phi_{1}\left|V\left(e^{-L_{1}} \phi_{2},-1\right)^{\dagger}\right| \overline{\phi}_{3}\right\rangle^{*}\quad\text { (by skew-symmetry) }\\
&=(-1)^{|\phi_2||\phi_3|-2h_{\phi_3}+|\phi_1||\phi_3|+|\phi_1||\phi_2|-2h_{\phi_2}}\left\langle\phi_{1}\left|V\left(\overline{\phi}_{2},-1\right)\right| \overline{\phi}_{3}\right\rangle^{*}\\
&=(-1)^{|\phi_2||\phi_3|-2h_{\phi_3}+|\phi_1||\phi_3|+|\phi_1||\phi_2|-2h_{\phi_2}}e^{i\pi(-h_{\phi_1}+h_{\phi_2}+h_{\phi_3})}\left(f_{ \overline{\phi}_{1} \overline{\phi}_{2} \overline{\phi}_{3}}\right)^{*}\quad \\&=(-1)^{|\phi_1||\phi_2|+|\phi_2||\phi_3|+|\phi_3||\phi_1|}e^{-i\pi(h_{\phi_1}+h_{\phi_2}+h_{\phi_3})}(f_{\bar{\phi_1}\bar{\phi_2}\bar{\phi_3}})^*
\end{split}
\]
where in the last step we used 
\begin{equation}
w^{L_0}V(\psi,z)w^{-L_0}=w^{h_{\psi}}V(\psi,zw)\quad \text{for $w=-1$}.
\label{eq:zVz}
\end{equation}
This is proved using the axiom (4) of Definition \ref{defn:scft} as follows: we have 
\[
w^{L_0}V(\psi,z)w^{-L_0}|0\rangle=w^{L_0}e^{zL_{-1}}\psi,
\]
since $L_0|0\rangle=0$ and $V(\psi,z)|0\rangle=e^{zL_{-1}}\psi$. We now use the following version of the Baker-Campbell-Hausdorff formula:
\begin{thm}\label{thm:bch}
If $[ X , Y ] = s Y$, where $s$ is a complex number with $s \neq 2 \pi i n$ for all integers $n$, then we have
\begin{equation}
    e^{X}e^{Y}=e^{\exp(s)Y}e^{X}.
\label{eq:bch}
\end{equation} 
\end{thm}
We can write 
\[
w^{L_0}e^{zL_{-1}}=e^{L_0\ln w}e^{zL_{-1}}.
\]
Using the fact that $[L_0,L_{-1}]=L_{-1}$, we see that the above theorem implies 
\[
e^{L_0\ln w}e^{zL_{-1}}=e^{\exp(\ln w)zL_{-1}}w^{L_0}.
\]
Thus we get 
\[
w^{L_0}e^{zL_{-1}}=e^{zwL_{-1}}w^{L_0}.
\]
Thus we get 
\[
w^{L_0}V(\psi,z)w^{-L_0}|0\rangle=e^{zwL_{-1}}w^{h_{\psi}}\psi.
\]
Part (1) of this proposition then gives 
\begin{equation}
    w^{L_0}V(\psi,z)w^{-L_0}=w^{h_{\psi}}V(\psi,zw).
\label{eq:wL0Vw-L0}
\end{equation}
To prove (f), we follow \cite[Proof of Proposition 2.9 (vi)]{Dolan:1994st} to get 
\[
\langle e^{L_1}\overline{\psi}|e^{L_1}\psi\rangle=(-1)^{2h_{\psi}+|\psi||\overline{
\psi}|}\langle e^{L_1}\psi|e^{L_1}\overline{\psi}\rangle\implies (-1)^{2h_{\psi}+|\psi||\overline{
\psi}|}=1. 
\]
If $\psi$ is Fermionic then $|\psi|=|\overline{
\psi}|=1$ and $h_{\psi}\in\mathbb{Z}+\frac{1}{2}.$ In other cases $h_{\psi}\in\mathbb{Z}$. Thus $h_{\psi}\in\frac{1}{2}\mathbb{Z}$.
\end{proof}
\end{prop}
Isomorphisms of SCVOAs and SVOAs is defined analogous to those for VOAs. A consequence of the definition is that isomorphisms also preserve the superconformal vector in addition to the conformal vector and the vacuum.   
%Note that an SCVOA $\mathscr{H}=\mathscr{H}_{\bar{0}}\oplus\mathscr{H}_{\bar{1}}$ has an obvious $\mathbb{Z}_2$ automorphism which acts as $+1$ on $\mathscr{H}_{\bar{0}}$ and $-1$ on $\mathscr{H}_{\bar{1}}$. 
%The partition function of the SCVOA $\mathscr{H}$ is defined as 
%
%\begin{equation}\label{eq:partfuncscvoa} Z_{\mathscr{H}}(\tau):=\text{Tr}_{ \mathscr{H}_{NS}}Pq^{L_0-\frac{\widehat{c}}{16}}+\text{Tr}_{ \mathscr{H}_{R}}Pq^{L_0-\frac{\widehat{c}}{16}}\end{equation}
%where $P=\frac{1+(-1)^F}{2}$ and 
%\\\\ In general, any spin quantum field theory\footnote{A spin-conformal field theory %is a conformal field theory which is sensitive to the spin structure of the Riemann %surface on which it is defined.} in 2d has a $\mathbb{Z}_2$ symmetry called the %\textit{fermion number symmetry} and is denoted by $(-1)^F$. It acts as $(-1)^F=1$ on %the Hilbert space corresponding to the bounding spin structure on the spatial circle %and as $(-1)^F=-1$ on the Hilbert space corresponding to the nonbounding spin %structure on the spatial circle. In an SCVOA it is characterised by the property that %it commutes with $T_B(z)$ and anticommutes with $T_F(z)$.  

\subsection{Spin Lifts Of A Conformal Field Theory}\label{sec:spin_lift}

Given a modular invariant holomorphic conformal field theory $\mathscr{H}$ with a non-anomalous
%
%\footnote{By non-anomalous, we mean that if we gauge the global symmetry by %introducing a background connection, then the resulting partition function is %invariant under gauge transformation.} 
%
$\mathbb{Z}_2$ automorphism, there is a natural way to construct a spin CFT. The essential idea goes back to the GSO projection applied to string worldsheet superconformal field theories to obtain a consistent string theory \cite{Gliozzi:1976qd} and played an important role in discussions of bosonization on arbitrary Riemann surfaces \cite{Alvarez-Gaume:1987wwg,Alvarez-Gaume:1986nqf}. More general versions of the construction of \cite{Alvarez-Gaume:1987wwg,Alvarez-Gaume:1986nqf} have been discussed in the recent literature \cite{Gaiotto:2015zta,Ji:2019ugf,Kapustin:2017jrc,Karch:2019lnn,Lin:2019hks,Shao:2023gho}. This spin-conformal field theory is called a spin lift of $\mathscr{H}$. 
The idea is to couple the non-anomalous $\mathbb{Z}_2$ automorphism of the CFT to an invertible topological field theory sensitive to spin structure known as the Arf
theory \cite{Debray:2018wfz}, and then gauge the diagonal global $\mathbb{Z}_2$ automorphism of the product theory. The resulting theory depends on the spin structure and in general the vertex operators do not satisfy the standard locality conditions of a super-vertex operator algebra. Nevertheless the GSO projection of the spin lift reproduces the original CFT. The spin lift also has an emergent $\mathbb{Z}_2$ symmetry which we will identify with the fermion number symmetry $(-1)^F$. Let us describe the construction in more detail.   
\subsubsection{Formal Description Of Gauging A Discrete Symmetry}\label{subsec:discretegauge}
Let us recall some well-known general facts. We are following \cite[Section 7.1]{Moore:2006dw}.
Let $\mathcal{F}$ be a $d$-dimensional quantum field theory,  that is, a functor $\mathcal{F}:\text{Bord}_d\longrightarrow\text{Vect}$ from the bordism category, whose objects are $(d-1)$-dimensional manifolds (possibly with extra structure) and whose morphisms are $d$-dimensional oriented bordisms, to the linear category of vector spaces and linear maps. We say that  $\mathcal{F}$ has global symmetry $G$ when,  for every $(d-1)$-manifold $Y$, the vector space $\mathscr{H}(Y):=\mathcal{F}(Y)$, called the space of states, is a representation space of $G$ and for any bordism $X:Y_1\longrightarrow Y_2$, the linear map $\mathcal{F}(X):\mathscr{H}(Y_1)\longrightarrow \mathscr{H}(Y_2)$, called the transition amplitude, is $G$-equivariant. Gauging the symmetry group involves two steps:
\begin{enumerate}
    \item Construct the ``equivariant theory'' by enlarging the number of fields. 
    In the present example our new fields are a principal $G$-bundle with connection. In our case $G$ is a finite group so the bundle has a unique connection.  
    \footnote{We use the term ``field'' in the generalized sense explained in Section 3 of 
    \cite{Freed:2013gc}. This step is the mathematical version of ``coupling the theory to background gauge fields'' in physics.}
    \item ``Sum'' or ``integrate'' over the principal $G$-bundles with connection with a suitable measure. 
\end{enumerate}
To construct the equivariant theory, we define a new bordism category $\text{Bord}_d^{G}$ whose objects are $(d-1)$-manifolds $Y$ equipped with principal $G$-bundles $P\to Y$ with connection and morphisms are $d$-dimensional bordisms $X:Y_1\longrightarrow Y_2$ also equipped with principal $G$-bundles with connection $P'\to X$ such that the restriction of $P'\to X$ to the ingoing boundary is isomorphic to $P_1 \to Y_1$ and the
restriction to the outgoing boundary is isomorphic to $P_2 \to Y_2$. 
%
%\footnote{When gauging discrete symmetries in a spin quantum field theory, one also %needs to specify the spin structure of the base manifolds. This will be the case in %Subsection } 
%%
%
Note that $\text{Bord}_d$ can be regarded as a subcategory of $\text{Bord}_d^{G}$ with trivial bundles associated to the objects and bordisms. 
Now, the equivariant theory $\mathcal{F}^{\text{eq}}:\text{Bord}_d^{G}\longrightarrow \text{Vect}$ assigns to every object $P\to Y$ a vector space $\mathscr{H}^{\text{eq}}(P):=\mathscr{H}(P\to Y)$ and for every bordism $(P'\to X)$ between $P_1 \to Y_1$ and $P_1 \to Y_1$ a transition amplitude 
\begin{equation}
    \mathcal{F}^{\text{eq}}(P'):=\mathcal{F}^{\text{eq}}(P'\to X):\mathscr{H}^{\text{eq}}(P_1)\longrightarrow \mathscr{H}^{\text{eq}}(P_2)
\end{equation}
such that when $P\to Y$ is the trivial bundle $Y\times G$, $\mathscr{H}^{\text{eq}}(Y\times G)=\mathscr{H}(Y)$ and for trivial bordism $X\times G$ between trivial objects $Y_1\times G$ and $Y_1\times G$, $\mathcal{F}^{\text{eq}}(P')$ reduces to the original transition amplitude $\mathcal{F}(X)$. The automorphisms of the bundle $P\to Y$ must act on $\mathscr{H}^{\text{eq}}(P)$. For the trivial bundle, this corresponds to the original global symmetry $G$-action on $\mathscr{H}(Y)$. The amplitudes are equivariant maps. 
 
Once we have the equivariant theories we can perform the second step of gauging, that is sum over bundles with connection. We need to take into account the fact that the isomorphic bundles have the same space of states. More precisely, if we have a bundle isomorphism $f:P\longrightarrow Q$ of bundles $P\to Y$ and $Q\to Y$ then we have an isomorphism $\tilde{f}:\mathscr{H}^{\text{eq}}(P)\longrightarrow\mathscr{H}^{\text{eq}}(Q)$ such that $\tilde{f}(\psi(P))=\psi(Q)$ for states $\psi(P)\in \mathscr{H}^{\text{eq}}(P)$. In particular, when $f:P\longrightarrow P$ is a bundle automorphism of $P\to Y$ then $\tilde{f}(\psi(P))=\psi(P)$. This is the mathematical equivalent of gauge invariant states. The gauged theory thus contains only gauge invariant states. The space of states of the gauged theory is then defined as a direct sum over isomorphism classes of principal bundles:
\begin{equation}
\mathscr{H}^{\text{gd}}(Y):=\bigoplus_{[P]\in\text{Prin}_G(Y)}\mathscr{H}^{\text{eq}}(P)^{\text{Aut}(P)}    
\end{equation}
where $\mathscr{H}^{\text{eq}}(P)^{\text{Aut}(P)}$ is the subspace of $\text{Aut}(P)$-invariant states of $\mathscr{H}^{\text{eq}}(P)$ and $\text{Prin}_G(Y)$ denotes the set of isomorphism classes of principal $G$-bundles over $Y$.
For a a $d$-dimensional bordism without boundary, the partition function $Z(X)$ of the theory is the complex number
%\footnote{A bordism $X$ without boundary can be thought of as a morphism between empty manifolds %$X:\emptyset\longrightarrow\emptyset$ and hence $\mathcal{F}(X)$ is a linear map between complex %numbers since $\mathscr{H}(\emptyset)=\mathbb{C}$.}
%
$Z(X):=\mathcal{F}(X)$. We also have partition function $Z^{\text{eq}}(P')$ for equivariant theories corresponding to principal bundle $P'\to X$. The partition function of the gauged theory is defined as 
\begin{equation}\label{eq:gaugedpartfunc}
    Z^{\text{gd}}(X)=\sum_{[P']\in\text{Prin}_G(X)}\frac{1}{|\text{Aut}(P')|}Z^{\text{eq}}(P').
\end{equation}

We now want to apply the above general procedure to $(1+1)$d quantum field theory. We take the spatial slices to be disjoint unions of circles to determine the gauged space of states. We start by noting that  any $g\in G$ we can construct a principal $G$-bundle $P_g\to S^1$ whose isomorphism class only depends on the conjugacy class of $g$, and moreover there is a 1-1 correspondence of conjugacy classes of $G$ and isomorphism classes of principal $G$-bundles over the circle.  More explicitly we have 
\begin{equation}
    P_g:=\mathbb{R}\times_\mathbb{Z}G\longrightarrow S^1
\end{equation}
where the equivalence relation is 
\begin{equation}
    (s+1,g')\sim (s,gg').
\end{equation}
The projection map is $[s,g']\mapsto e^{2\pi is}$ and the right action of $G$ is given by $[s,g']\cdot g''=[s,g'g'']$. The automorphism group of this bundle is the centraliser $Z(g)$ of $g$ in $G$. Thus for each $[g]\in\text{Conj}(G)$ we have an equivariant space of states $\mathscr{H}^{\text{eq}}(P_g)$ which is a representation of the centraliser $Z(g)$
and the space of states of the gauged theory is 
\begin{equation}\label{eq:gaugedhilbertspace}
    \mathscr{H}^{\text{gd}}(S^1):=\bigoplus_{[g]\in\text{Conj}(G)}\mathscr{H}^{\text{eq}}(P_g)^{Z(g)}.
\end{equation}
Using the structure \eqref{eq:gaugedhilbertspace} of the gauged space of states, we obtain the partition function on the torus 
\begin{equation}
    Z^{\text{gd}}(\tau)(T^2)=\text{Tr}_{\mathscr{H}^{\text{gd}}(S^1)}e^{-2\pi \text{Im}(\tau)H}e^{2\pi i\text{Re}(\tau)P}
\end{equation}
where $\tau$ is the complex moduli of the torus and $H$ and $P$ are the Hamiltonian and momentum operators generating time and space translations respectively. By expanding the sum, one can also write the partition function as 
\begin{equation}\label{eq:toruspartfuncgen}
Z^{\text{gd}}(T^2)(\tau)=\frac{1}{|G|} \sum_{[g_1,g_2]=1}Z^{\text{eq}}(P_{g_1,g_2})(\tau)   
\end{equation}
where 
\begin{equation}\label{eq:partfuncdefecthilb}
Z^{\text{eq}}(P_{g_1,g_2})(\tau)=\text{Tr}_{\mathscr{H}^{\text{eq}}(P_{g_1})} \tilde g_2 e^{-2\pi \text{Im}(\tau)H}e^{2\pi i\text{Re}(\tau)P}    
\end{equation}
The $g_2$ action on $\mathscr{H}$ must be extended, in a consistent way, 
to an action by $\tilde g_2$ on the spaces $\mathscr{H}^{\text{eq}}(P_{g_1})$. How this is done 
is part of the definition of the gauged theory, but it is constrained by, e.g. modular invariance. On the 
left hand side of \eqref{eq:partfuncdefecthilb} 
$P_{g_1, g_2}$ denotes the bundle over $T^2$ with monodromy $g_1, g_2$ around the two cycles. To be explicit, we have
\begin{equation}
P_{g_1, g_2}=\mathbb{R}^2 \times_{\mathbb{Z} \oplus \mathbb{Z}} G
\end{equation}
and the equivalence relation is given by
\begin{equation}
(s+1, t, g) \sim\left(s, t, g_1 g\right), \quad(s, t+1, g) \sim\left(s, t, g_2g\right).
\end{equation}

A standard fact is that when the gauge group $G$ is Abelian, the gauged theory has a global symmetry isomorphic to the Pontrjagin dual $\tilde{G}:=\text{Hom}(G,\mathrm{U}(1))$ of $G$. An element $\chi\in \tilde{G}$ acts on the equivariant Hilbert space $\mathscr{H}^{\text{eq}}(P_{g})$ by multiplying $\chi(g)\in \mathrm{U}(1)$. Note that this action is well defined since $\chi$ being a group homomorphism only depends on the conjugacy class of $g\in G$. This defines the action of $\chi$ on $\mathscr{H}^{\text{gd}}(S^1)$. We can write the action schematically as: 
\begin{equation}\label{eq:pontdualactonHgd}
    \chi\cdot\mathscr{H}^{\text{gd}}(S^1)=\bigoplus_{[g]\in\text{Conj}(G)}\chi(g)\cdot\mathscr{H}^{\text{eq}}(P_g)^{Z(g)},\quad \chi\in\Tilde{G}.
\end{equation}
For a conformal field theory, the partition function takes the form 
\begin{equation}
Z^{\text{gd}}(T^2)=\frac{1}{|G|} \sum_{[g_1,g_2]=1}\text{Tr}_{\mathscr{H}^{\text{eq}}(P_{g_1})}\tilde g_2q^{L_0-\frac{c}{24}}\bar{q}^{\bar{L}_0-\frac{\bar{c}}{24}},\quad q=e^{2\pi i\tau},~~\bar{q}=e^{-2\pi i\bar{\tau}}    
\end{equation}
after recognising that 
\begin{equation}
    H=L_0+\bar{L}_0-\frac{c}{24}-\frac{\bar{c}}{24},\quad P=i(L_0-\bar{L}_0).
\end{equation}
For vertex operator algebras (\textit{chiral algebra} of conformal field theories), we only have $q$ dependence in the partition function. 
We will be mostly interested in the case when $G=\mathbb{Z}_2$. In this case, we only have two isomorphism classes of principal bundles since we only have two conjugacy classes. 
The space of states corresponding to the conjugacy class of $-1$ is called the \textit{twisted sector}.
\subsubsection{The Arf Theory}
The Arf theory is a 2 dimensional \textit{spin topological quantum field theory} (spin TQFT) defined on a generic Riemann surface $\Sigma$. See \cite{Debray:2018wfz} for a nice discussion of this theory. Recall that all
orientable Riemann surfaces have vanishing second Stiefel-Whitney class \cite{Nakahara:2003nw} and hence are spin manifolds. Moreover the set of spin structures on $\Sigma$ is in 1-1 correspondence with $H^1(\Sigma,\mathbb{Z}_2)\cong\mathbb{Z}_2^{2g}$ \cite{Atiyah:1971} where $g$ is the genus of $\Sigma$. In fact, the set of spin structure on $\Sigma$ is an $H^1(\Sigma,\mathbb{Z}_2)$-torsor over $\mathbb{Z}_2$. 
Thanks to a beautiful paper of M. Atiyah \cite{Atiyah:1971} we can  
define the Arf invariant \cite{Arf1941UntersuchungenQ} of a spin structure to be the mod-two index of the Dirac operator:
\begin{equation}
 \text{Arf}[\rho]=\text{ind}_2(D_\rho).   
\end{equation}
Thus there is a dichotomy of even and odd spin structures: 
\begin{equation}
    \text{Arf}[\rho]:=\begin{cases}
        1&\rho \text{ odd}\\0&\rho \text{ even}.
    \end{cases}
\end{equation}
On a genus $g$ orientable surface  $\Sigma$ there are  $2^{g-1}(2^g - 1)$   odd and $2^{g-1}(2^g+1)$ even spin structures. 

Using the Arf invariant, one can define a 2 dimensional invertible
topological field theory
\footnote{The notion of an invertible TQFT is discussed in \cite{Freed:2014iua,Freed:2016rqq,Freed:2004yc}.}
called the Arf theory, see \cite{Debray:2018wfz} for the detailed definition. The partition function of the Arf theory on a closed oriented Riemann surface $\Sigma$ with spin structure $\rho$ is: 
\begin{equation}\label{eq:partfuncarf}
\text{Arf}(\Sigma_\rho):=
Z_{\text{Arf}}[\Sigma_\rho]:=e^{i\pi \text{Arf}[\rho]}.    
\end{equation}
Note that the partition function is topological in that it does not depend on the conformal structure of the Riemann surface but it is sensitive to the choice of spin structure on $\Sigma$. To describe the Hilbert space of the Arf theory, we consider the Arf theory on a cylinder $\mathbb{R}\times S^1$ where $\mathbb{R}$ is the time direction. The (complex) Hilbert space of the Arf theory on the spatial circle is one dimensional. 
%
%spanned by the ground state. This can be realised as the low energy limit of the nontrivial %topological phase of
%the Kitaev chain \cite{Kitaev:2000nmw}. This is a fermionic chain whose lowest energy state on $S^1$
%is non-degenerate and spans a one dimensional Hilbert space. 
%
In fact, the Hilbert space is   $\mathbb{Z}_2$-graded and the grading is defined by the 
action of the spin automorphism on the space. In the physical literature this operator is denoted   $(-1)^F$. The Hilbert space for the NS (bounding) spin structure has $(-1)^F=1$ and for the R (nonbounding) spin structure has $(-1)^F=-1$. Thus, as supervector spaces we have 
\footnote{For any field $\mathbb{F}$ we denote by $\mathbb{F}^{p|q}$ the supervector space $\mathbb{F}^{p|q}:=\mathbb{F}^{p}\oplus \mathbb{F}^{q}$ where $\mathbb{F}^{p}$ is the even subspace ($\mathbb{Z}_2$-parity $\bar{0}$) and $\mathbb{F}^{q}$ is the odd subspace ($\mathbb{Z}_2$-parity $\bar{1}$).
%see \cite[Section 23]{Moore:2021} for more details.
} 
$\text{Arf}(S^1_{\text{NS}}) \cong \mathbb{C}^{1|0}$   and $\text{Arf}(S^1_{\mathrm{R}}) \cong \mathbb{C}^{0|1}$ and $(-1)^F$ is just the parity operator. 
%
%We will denote tha Arf theory on a Riemann surface $\Sigma_\rho$ equipped with spin %structure $\rho$ by $\text{Arf}[\Sigma_\rho]$.
%
\subsubsection{Coupling A CFT To The Arf Theory}
Let $\mathscr{C}$ be a modular invariant CFT (i.e. a functor from Bord${}_2$ to Vect) with statespace on a circle 
$\mathscr{C}(S^1):=\mathscr{H}$ and suppose $\mathscr{C}$ has 
a non-anomalous $\mathbb{Z}_2$ automorphism
$\sigma$. We start by taking the tensor product of the $\sigma$-equivariant lift of 
$\mathscr{C}$ with the $(-1)^F$-equivariant lift of the Arf theory.
The product theory has as domain a bordism category $\text{Bord}_2^{\mathbb{Z}_2, \text{Spin}}$ 
consisting of $1$-manifolds and their bordisms equipped with a spin structure and, independently, with a principal $\mathbb{Z}_2$-bundle. Let  $P_{\bar 0}$ and $P_{\bar 1}$  be the trivial and nontrivial principal $\mathbb{Z}_2$ bundles over the circle corresponding to the trivial and nontrivial double-covers, respectively. 

We need to define the ``equivariant lift'' of the Arf theory on the category $\text{Bord}_2^{\mathbb{Z}_2, \text{Spin}}$. We do this as follows: Isomorphism classes of principal $\mathbb{Z}_2$ bundles over any manifold $X$ are in $1-1$ correspondence with cohomology classes $H^1(X,\mathbb{Z}_2)$.
\footnote{Recall that for a discrete group $G$, ismorphism classes of principal $G$-bundles over a connected space $X$ are classified by $\text{Hom}(\pi_1(X,x_0),G)/\sim$ where the equivalence relation is given by conjugation. For $G$ Abelian we have: 
\begin{equation}
H^1(X,G)\cong \text{Hom}(H_1(X),G)\cong \text{Hom}(\pi_1(X,x_0),G)/\sim ~ .  
\end{equation}
Taking $G=\mathbb{Z}_2$, an alternate way to prove this result is to observe that the classifying space of $\mathbb{Z}_2$ is $\mathbb{RP}^{\infty}$ and the set of homotopy classes of maps from a space $X$ to $\mathbb{RP}^{\infty}$ is isomorphic to $H^1(X,\mathbb{Z}_2)$. }
For each $t\in H^1(X,\mathbb{Z}_2)$ choose a representative bundle $\pi:P_t\to X$.
As noted above, the set of spin structures on $X$ is a torsor for $H^1(X,\mathbb{Z}_2)$. If $\rho$ denotes a spin structure and $t\in H^1(X,\mathbb{Z}_2)$ we write the action as $\rho \to \rho +t$. We \underline{define} the equivariant lift of the $\text{Arf}$ theory to be 
\begin{equation}\label{eq:Arf-Eq-Def}
    \text{Arf}^{\text{eq}}[P_t,X_\rho]:= \text{Arf}[X_{\rho+t}] 
\end{equation}
where $X_\rho$ is the manifold $X$ with spin structure $\rho$. 
In particular,  
\begin{equation}
\begin{split}
    \text{Arf}^{\rm{eq}}[P_{\bar 0},S^1_{\mathrm{NS}}]  & = \mathbb{C}^{1\vert 0}  \\ 
    \text{Arf}^{\rm{eq}}[P_{\bar 1},S^1_{\mathrm{NS}}]  & = \mathbb{C}^{0\vert 1} \\ 
    \text{Arf}^{\rm{eq}}[P_{\bar 0},S^1_{\mathrm{R}}]  & = {\mathbb{C}^{0\vert 1}}   \\ 
    \text{Arf}^{\rm{eq}}[P_{\bar 1},S^1_{\mathrm{R}}]  & = \mathbb{C}^{1\vert 0} \\ 
\end{split}
\end{equation}
We define the lifted global symmetry $(-1)^F_{\text{Arf}^{\text{eq}}}$ on $\text{Arf}^{\text{eq}}$ again to be the parity operator.  
 
The equivariant theory $\mathscr{C}^{\text{eq}}$ inherits a $\mathbb{Z}_2$ symmetry $\tilde\sigma$ from the action of $\sigma$ on $\mathscr{C}$.  The action of $\tilde\sigma$  on $\mathscr{H} = \mathscr{C}^{\text{eq}}(P_{\bar 0}) $ is just the original action of $\sigma$. The action of $\tilde\sigma$ on $\mathscr{H}^{\text{eq}} := \mathscr{C}^{\text{eq}}(P_{\bar 1})$ is part of the definition of the gauging. In our case, as we will see presently, the action is constrained by the modular covariance of the torus partition functions when combined with the Hilbert space interpretation of torus partition functions.  
\footnote{In \underline{topological} theories this action must be the identity. See Theorem 4 of \cite{Moore:2006dw}. But in general CFT this need not be so. }
Now, the   product theory $\mathscr{C}^{\text{eq}}\otimes \text{Arf}^{\text{eq}}$ has a $\mathbb{Z}_2$ automorphism given by the diagonal $\sigma_d:= \tilde\sigma \otimes (-1)^F_{\text{Arf}^{\text{eq}}}$. The \emph{spin lift} $\mathscr{C}^{\text{SL}}$ 
of $\mathscr{C}$ is 
the spin theory, with domain category $\text{Bord}_2^{\text{Spin}}$, obtained by gauging the group of 
automorphisms generated by $\sigma_d$. Thus, the Hilbert spaces on the two spin circles ($S^1_{\text{NS}}$ is the circle with bounding spin structure and $S^1_{\text{R}}$ is the circle with nonbounding spin structure, see Footnote \ref{foot:spin_st_circle}) are 
\begin{equation}\label{eq:hilspSL}
\begin{split} 
    \mathscr{H}^{\text{SL}}_{\text{NS}}:=\mathscr{C}^{\text{SL}}(S^1_{\text{NS}})  & =\left( \mathscr{H}(S^1)\otimes \mathbb{C}^{1\vert 0}\right)^{\sigma_d=+1} \oplus \left( \mathscr{H}^{\text{eq}}\otimes \mathbb{C}^{0\vert 1} \right)^{\sigma_d = +1} \\ 
    & \cong \mathscr{H}^+ \oplus \mathscr{H}^{\text{eq},-} \\
    \mathscr{H}^{\text{SL}}_{\text{R}}:=\mathscr{C}^{\text{SL}}(S^1_{\text{R}})& =\left( \mathscr{H}(S^1)\otimes \mathbb{C}^{0\vert 1}\right)^{\sigma_d=+1} \oplus \left( \mathscr{H}^{\text{eq}}\otimes \mathbb{C}^{1\vert 0} \right)^{\sigma_d = +1} \\ 
    & \cong \mathscr{H}^- \oplus \mathscr{H}^{\text{eq},+}, 
    \end{split}
\end{equation}
where $\mathscr{H}^{\text{eq},\pm}$ is the $\pm 1$ eigenspace of $\mathscr{H}^{\text{eq}}$ under the action of $\tilde\sigma$. 
Now, whenever we gauge a theory by an Abelian group $A$ the gauged theory has a global symmetry group isomorphic to the Pontryagin dual $\tilde A$. In our case $A$ is generated by $\sigma_d$ so 
the Pontryagin dual is $\tilde A\cong \mathbb{Z}_2$. We denote the nontrivial element by 
$\check{\sigma}_d$. From eq.  \eqref{eq:pontdualactonHgd} we see that the action of $\check{\sigma}_d$ on the spin lift is given by
\begin{equation}\label{eq:FactonSL}
\check{\sigma}_d=\begin{cases}
    +1&\text{on }~\mathscr{H}^+\oplus\mathscr{H}^-\\-1&\text{on }~\mathscr{H}^{\text{eq},+}\oplus\mathscr{H}^{\text{eq},-}.
\end{cases}    
\end{equation}  

Now we turn to the partition functions of the spin-lifted theory. As we remarked above, 
the principal $\mathbb{Z}_2$ bundles over a Riemann surface $\Sigma$ are in 1-1 correspondence with 
$H^1(\Sigma,\mathbb{Z}_2)$. When $\Sigma$ is connected the automorphism group of these bundles is isomorphic to $\mathbb{Z}_2$. 
When we consider the $\sigma_d$-equivariant lift of $\mathscr{C}\otimes \text{Arf}$ there are a  
total of $|H^1(\Sigma,\mathbb{Z}_2)|=2^{2g}$ equivariant partition functions indexed by ``background $\mathbb{Z}_2$-connections'' $T$.  The partition function for the equivariant lift of the product theory is given by (see eq.  \eqref{eq:partfuncarf} and eq.  \eqref{eq:Arf-Eq-Def})
\begin{equation}
    Z^{\text{eq}}[P_T,\Sigma_\rho]:=Z_{\mathscr{C}^{\text{eq}}}[P_T]Z_{\text{Arf}^{\text{eq}}}[P_T,\Sigma_{\rho}]=Z_{\mathscr{C}^{\text{eq}}}[P_T]e^{ i\pi \text{Arf}[\rho+T]}.
\end{equation}
Next, we gauge the equivariant theory, so we promote the background connection $T$ to a dynamical connection $t$. That is, we sum over isomorphism classes of principal $\mathbb{Z}_2$-bundles over $\Sigma$. Applying 
\eqref{eq:gaugedpartfunc} the partition function of the gauged theory is given by 
\footnote{A version of  \eqref{eq:partfuncZ2cover} appeared as equation $(2.5)$ of \cite{Lin:2019hks}, but our \eqref{eq:partfuncZ2cover} differs from equation $(2.5)$ of \cite{Lin:2019hks} in a few minor ways.  First we have chosen $S=0$ relative to that paper. This is inconsequential since spin structures form a torsor for $H^1(\Sigma, \mathbb{Z}_2)$. Second, there is no factor of $\exp[i \pi \text{Arf}[\rho]]$ in our version. This is merely a choice of ``gravitational counterterm,'' making our formula more in line with the description in terms of gauging the  $\mathbb{Z}_2$ symmetry 
$\langle \sigma_d \rangle$. Finally we have an overall factor ${1/2}$, again because in the description in terms of gauging the automorphism group of a principal $\mathbb{Z}_2$ bundle over a connected surface is simply $\mathbb{Z}_2$. Equation 
$(2.5)$ of \cite{Lin:2019hks} has a factor of $2^{-g}$, giving a factor of   $2^{1-g}$ relative to our normalization. Again this can be interpreted as coupling to a background invertible TQFT (an Euler character theory) so the difference is not essential. The choice $2^{-g}$ made in \cite{Lin:2019hks}  has the nice property that if the original bosonic theory couples trivially to the background $\mathbb{Z}_2$ gauge theory then the spin lift has the same partition function.  
}
\begin{equation}\label{eq:partfuncZ2cover}
    \begin{aligned} Z^{\text{SL}}[\Sigma_\rho]&=\frac{1}{2} \sum_{t\in H^1(\Sigma,\mathbb{Z}_2)} Z^{\text{eq}}[P_t,\Sigma_\rho]\\&= \frac{1}{2} \sum_{t\in H^1(\Sigma,\mathbb{Z}_2)} Z_{\mathscr{C}^{\text{eq}}}[P_t]\exp \left(i \pi\operatorname{Arf}[\rho+t]\right) .\end{aligned}
\end{equation}

Now we would like to check compatibility of \eqref{eq:partfuncZ2cover} with the 
Hilbert space description \eqref{eq:hilspSL} by examining the special case where  $\Sigma$ is a torus. There are four spin structures on the torus  and we can denote them by $T^2_{00},T^2_{01},T^2_{10},T^2_{11}$ where 0 and 1 denotes the bounding and nonbounding spin structure on a cycle of the torus. The corresponding partition functions of the spin-lifted theory are denoted by 
$Z^{\text{SL}}[00], \dots, Z^{\text{SL}}[11]$, respectively. 
Interpreting the first circle factor of $T^2=S^1 \times S^1$ as the spatial circle 
we can expect that   the partition functions \eqref{eq:partfuncZ2cover} can be interpreted as the NS and R partition functions with and without $\check{\sigma}_d$ insertion:
\begin{equation}\label{eq:torpartfunc01}
\begin{split}
    &Z^{\text{SL}}[00]=\text{Tr}_{\mathscr{H}^{\text{SL}}_{\text{NS}}}q^{L_0-\frac{c}{24}},\\&Z^{\text{SL}}[01]=\text{Tr}_{\mathscr{H}^{\text{SL}}_{\text{NS}}}\check{\sigma}_dq^{L_0-\frac{c}{24}},\\&Z^{\text{SL}}[10]=\text{Tr}_{\mathscr{H}^{\text{SL}}_{\text{R}}}q^{L_0-\frac{c}{24}},\\&Z^{\text{SL}}[11]=\text{Tr}_{\mathscr{H}^{\text{SL}}_{\text{R}}}\check{\sigma}_dq^{L_0-\frac{c}{24}},\quad q=e^{2\pi i\tau}.
\end{split}    
\end{equation}
We now verify that the relations in eq.  \eqref{eq:torpartfunc01} is consistent with eq.  \eqref{eq:partfuncZ2cover} and eq.  \eqref{eq:hilspSL}.
%
%we obtain four terms in the sum corresponding to the four bundles $P_{g_1,g_2}$ as in %\eqref{eq:partfuncdefecthilb}. We obtain four different partition functions %corresponding to four spin structures on the torus. 
%Let us denote the four spin structures on the torus by 00,01,10,11 where  Out of these %00,01,10 are even spin structures and 11 is odd spin structure on the torus. Since the %set of spin structures is a torsor for $H^1(\Sigma,\mathbb{Z}_2)$, the sum in eq. %\eqref{eq:partfuncZ2cover} involves the evaluation of the Arf invariant over all four %spin structures on the torus $\Sigma$.
%
%
%Next we have (see eq.  \eqref{eq:partfuncdefecthilb}) 
%
%\begin{equation}\label{eq:ZtZeqgg}
%\begin{split}
%Z_{\mathscr{C}^{\text{eq}}}[00]=Z^{\text{eq}}(P_{\bar{0},\bar{0}}),\quad %Z_{\mathscr{C}^{\text{eq}}}[01]=Z^{\text{eq}}%%(P_{\bar{0},\bar{1}}),\\Z_{\mathscr{C}^{\text{eq}}}[10]=Z^{\text{eq}}%%%(P_{\bar{1},\bar{1}}).
%\end{split}
%\end{equation}
%Here we write $Z_{\mathscr{C}^{\text{eq}}}[t]$ for $Z_{\mathscr{C}^{\text{eq}}}[P_t]$ %for similicity of notation. 
%
%
From \eqref{eq:partfuncZ2cover}   we have 
\begin{equation}
\begin{split}
&Z^{\text{SL}}[00]=\frac{1}{2}\left[Z^{\text{eq}}(P_{\bar{0},\bar{0}})+Z^{\text{eq}}(P_{\bar{0},\bar{1}})+Z^{\text{eq}}(P_{\bar{1},\bar{0}})-Z^{\text{eq}}(P_{\bar{1},\bar{1}})\right] \\&Z^{\text{SL}}[01]=\frac{1}{2}\left[Z^{\text{eq}}(P_{\bar{0},\bar{0}})+Z^{\text{eq}}(P_{\bar{0},\bar{1}})-Z^{\text{eq}}(P_{\bar{1},\bar{0}})+Z^{\text{eq}}(P_{\bar{1},\bar{1}})\right] \\&Z^{\text{SL}}[10]=\frac{1}{2}\left[Z^{\text{eq}}(P_{\bar{0},\bar{0}})-Z^{\text{eq}}(P_{\bar{0},\bar{1}})+Z^{\text{eq}}(P_{\bar{1},\bar{0}})+Z^{\text{eq}}(P_{\bar{1},\bar{1}})\right] \\&Z^{\text{SL}}[11]=\frac{1}{2}\left[-Z^{\text{eq}}(P_{\bar{0},\bar{0}})+Z^{\text{eq}}(P_{\bar{0},\bar{1}})+Z^{\text{eq}}(P_{\bar{1},\bar{0}})+Z^{\text{eq}}(P_{\bar{1},\bar{1}})\right]
\end{split}
\end{equation}
Using   \eqref{eq:partfuncdefecthilb} we see that 
\begin{equation}
\begin{split}
Z^{\text{SL}}[00]&=\frac{1}{2}\left[Z^{\text{eq}}(P_{\bar{0},\bar{0}})+Z^{\text{eq}}(P_{\bar{0},\bar{1}})+Z^{\text{eq}}(P_{\bar{1},\bar{0}})-Z^{\text{eq}}(P_{\bar{1},\bar{1}})\right]\\&=\text{Tr}_{\mathscr{H}}\left(\frac{1+\tilde{\sigma}}{2}\right)q^{L_0-\frac{c}{24}}+\text{Tr}_{\mathscr{H}^{\text{eq}}}\left(\frac{1-\tilde{\sigma}}{2}\right)q^{L_0-\frac{c}{24}}\\&=\text{Tr}_{\mathscr{H}^+}q^{L_0-\frac{c}{24}}+\text{Tr}_{\mathscr{H}^{\text{eq}-}}q^{L_0-\frac{c}{24}}=\text{Tr}_{\mathscr{H}^{\text{SL}}_{\text{NS}}}q^{L_0-\frac{c}{24}}.
\end{split}
\end{equation}
Similarly one can check all other relations in  \eqref{eq:torpartfunc01}.

We now review how modular invariance helps us determine the equivariant Hilbert space $\mathscr{H}^{\text{eq}}(P_{\bar 1})$ as a $\langle \tilde \sigma\rangle$-module. To this end, note that since $g_1,g_2$ are involutions 
\begin{equation}
    Z^{\text{eq}}(P_{g_1,g_2})(\tau)=Z^{\text{eq}}(P_{g_2,g_1})(-1/\tau)
\end{equation}
since exchanging the two cycles on the torus changes $\tau\to-1/\tau$ and $P_{g_1,g_2}\to P_{g_2,g_1}$. Note that this in particular implies that $Z^{\text{eq}}(P_{g,g})(\tau)$ is invariant under $\tau\to-1/\tau$. For $g=\bar{0}$, this is obvious since it is the partition function of the original theory. But for $g\neq \bar{0}$ we obtain a nontrivial identity. Taking $g_2=\bar{0}$, we see that the trace over\footnote{In modern terminology, $\mathscr{H}^{\text{eq}}(P_{g_1})$ is called the \textit{defect Hilbert space} corresponding to the $\mathbb{Z}_2$ line defect inserted along the longitude of the cylinder \cite{Huang:2021zvu,Pal:2020wwd}.} $\mathscr{H}^{\text{eq}}(P_{g_1})$ can be obtained from the trace over the original Hilbert space $\mathscr{H}$ with the insertion of $g_1$ and performing the modular transformation $\tau\to-1/\tau$: 
\begin{equation}\label{eq:defecthilmodtrans}
\text{Tr}_{\mathscr{H}^{\text{eq}}(P_{g_1})}q^{L_0-\frac{c}{24}}=Z_{\mathscr{H}}^{g_1}(-1/\tau)    
\end{equation}
where 
\begin{equation}
Z_{\mathscr{H}}^{g_1}(\tau):= \text{Tr}_{\mathscr{H}}g_1q^{L_0-\frac{c}{24}}.   
\end{equation}
This relation for $g_1=\bar{1}$ in principle can be used to determine the nontrivial equivariant Hilbert space with the understanding that
\begin{equation}
Z^{\bar{1}}_{\mathscr{H}^{\text{eq}}}(\tau)=\text{Tr}_{\mathscr{H}^{\text{eq}}}\tilde\sigma q^{L_0-\frac{c}{24}}.    
\end{equation}
where $\tilde\sigma$ acts on the equivariant Hilbert space. Moreover modular invariance also implies that 
\begin{equation}\label{eq:P11fromP10}
Z^{\text{eq}}(P_{\bar{1},\bar{1}})(\tau)=Z^{\text{eq}}(P_{\bar{1},\bar{0}})(\tau+1).    
\end{equation}
Using the equivariant Hilbert space, one can determine the $\check{\sigma}_d=\pm 1$ R and NS sectors as follows:
\begin{equation}\label{eq:trNS-R-}
\begin{split}
    &Z_{\text{NS}^-}^{\text{SL}}(\tau):=\text{Tr}_{\mathscr{H}_{\text{NS}}^{\text{SL}-}}q^{L_0-\frac{c}{24}}=\frac{1}{2}\left[Z^{\text{eq}}(P_{\bar{1},\bar{0}})-Z^{\text{eq}}(P_{\bar{1},\bar{1}})\right],\\&Z_{\text{R}^+}^{\text{SL}}(\tau):=\text{Tr}_{\mathscr{H}_{\text{R}}^{\text{SL}+}}q^{L_0-\frac{c}{24}}=\frac{1}{2}\left[Z^{\text{eq}}(P_{\bar{1},\bar{0}})+Z^{\text{eq}}(P_{\bar{1},\bar{1}})\right]
\end{split}    
\end{equation}
where $\mathscr{H}_{\text{NS}}^{\text{SL}\pm}$ is the $\check{\sigma}_d=\pm 1$ subspace of $\mathscr{H}_{\text{NS}}^{\text{SL}}$ and so on. More explicitly
\begin{equation}
\mathscr{H}_{\text{NS}}^{\text{SL}}:=\mathscr{H}^{\text{SL}+}_{\text{NS}}\oplus\mathscr{H}^{\text{SL}-}_{\text{NS}},\quad \mathscr{H}_{\text{R}}^{\text{SL}}:=\mathscr{H}^{\text{SL}+}_{\text{R}}\oplus\mathscr{H}^{\text{SL}-}_{\text{R}}.    
\end{equation} 
One can also check that 
\begin{equation}\label{eq:NSR+-partfunc}
\begin{split}
&Z_{\text{NS}^+}^{\text{SL}}(\tau):=\text{Tr}_{\mathscr{H}_{\text{NS}}^{\text{SL}+}}q^{L_0-\frac{c}{24}}=\frac{1}{2}\left[Z^{\text{eq}}(P_{\bar{0},\bar{0}})+Z^{\text{eq}}(P_{\bar{0},\bar{1}})\right],\\&Z_{\text{R}^-}^{\text{SL}}(\tau):=\text{Tr}_{\mathscr{H}_{\text{R}}^{\text{SL}-}}q^{L_0-\frac{c}{24}}=\frac{1}{2}\left[Z^{\text{eq}}(P_{\bar{0},\bar{0}})-Z^{\text{eq}}(P_{\bar{0},\bar{1}})\right]
\end{split}    
\end{equation}
As we will show with the explicit example of the Monster CFT, these relations determine the action of $\tilde{\sigma}$ on $\mathscr{H}^{\text{eq}}$.\\\\

Finally, if we gauge the $\check{\sigma}_d$ symmetry of the spin conformal field theory we just obtained, we recover the original theory. Let $Z_{\mathscr{H}}$ denote the gauged torus partition function of the resulting bosonic theory. By eq.  \eqref{eq:toruspartfuncgen} we have
\begin{equation}
\begin{split}
    Z_{\mathscr{H}}&=\frac{1}{2}\sum_{g_1,g_2\in\mathbb{Z}_2}Z^{\text{eq}}(P_{g_1,g_2})=Z_{\mathrm{NS}^+}^{\text{SL}}+Z_{\mathrm{R}^+}^{\text{SL}}\\&=\text{Tr}_{\mathscr{H}^+}q^{L_0-\frac{c}{24}}+\text{Tr}_{\mathscr{H}^-}q^{L_0-\frac{c}{24}}=\text{Tr}_{\mathscr{H}}q^{L_0-\frac{c}{24}}
\end{split}
\end{equation}
where we used eq.  \eqref{eq:trNS-R-} and \eqref{eq:NSR+-partfunc} in the second equality and eq.  \eqref{eq:FactonSL} and eq.  \eqref{eq:hilspSL} for the final equality. Thus we obtain the partition function for the original theory. On the other hand using eq.  \eqref{eq:torpartfunc01} we see that 
\begin{equation}
\begin{split}
\frac{1}{2}\sum_{\rho\in H^1(T^2,\mathbb{Z}_2)}Z^{\text{SL}}[\rho]&=\text{Tr}_{\mathscr{H}^{\text{SL}}_{\text{NS}}}\left(\frac{1+\check{\sigma}_d}{2}\right)q^{L_0-\frac{c}{24}}+\text{Tr}_{\mathscr{H}^{\text{SL}}_{\text{R}}}\left(\frac{1+\check{\sigma}_d}{2}\right)q^{L_0-\frac{c}{24}}\\&=\text{Tr}_{\mathscr{H}^{\text{SL}+}_{\text{NS}}}q^{L_0-\frac{c}{24}}+\text{Tr}_{\mathscr{H}^{\text{SL}+}_{\text{R}}}q^{L_0-\frac{c}{24}}=Z_{\mathscr{H}}.
\end{split}
\end{equation}
This is the usual GSO projection which involves summing over spin structures with the insertion of $\frac{1+\check{\sigma}_d}{2}$ in the traces over NS and R sectors \cite{Gliozzi:1976qd}. From this trace formula for the partition function, we see that the gauged Hilbert space is the $\check{\sigma}_d=1$ space of the spin theory but without the spin dependence, i.e. we forget the supervector space structure of the RHS below:
\begin{equation}
\mathscr{H}^{\text{SL}+}_{\text{NS}}\oplus\mathscr{H}^{\text{SL}+}_{\text{R}}\cong\mathscr{H}^+\oplus\mathscr{H}^-\cong\mathscr{H}.
\end{equation}
\\\\
%Given a CFT with a $\mathbb{Z}_2$ automorphism, its spin lift with respect to the %automorphism as explained above defines a spin quantum field theory. 

In mathematical terms, the spin lift of a CFT is an abelian intertwing algebra, see Remark \ref{rem:abelian_int} below. In general the spin lift of a holomorphic CFT defined by a VOA is not an SCVOA or even an SVOA. But in some cases, one might be able to identify a subspace of the spin lift which forms an SVOA and view the spin lift as a non-local extension of the SVOA. Suppose such an SVOA subsector exists in a spin lift of a CFT, then from Proposition \ref{prop:propscft}, the SVOA contains vertex operators of half-integral conformal dimensions. It is possible that the SVOA contains a conformal weight $\frac{3}{2}$ primary which satisfies the OPE \eqref{eq:opetbtf} and hence is in fact an SCVOA. Our goal in this paper is to show that the Monster CFT carries precisely that structure. In Subsection \ref{subsec:BB} we will show that the ``Beauty and the Beast'' SCFT is the spin lift of the Monster CFT with respect to a $\mathbb{Z}_2$ automorphism of the Monster CFT and has a dimension $\frac{3}{2}$ primary satisfying the OPE \eqref{eq:opetbtf}. We will also identify a subsector of the ``Beauty and the Beast'' SCFT which has an SCVOA structure. 
\begin{comment}
******************

\textcolor{red}{it has an SVOA structure} (some explanation required). 

MY VIEW IS THAT THE STATEMENT IS FALSE AND CANNOT BE PROVEN. THERE MIGHT BE A SUBSPACE OF $H_{BB}$ WHICH IS A SVOA, BUT $H_{BB}$ IS THEN SOME KIND OF EXTENSION. I DO NOT THINK WE NEED THE STATEMENT FOR ANYTHING IMPORTANT WE DO BELOW. 

RANVEER: IF YOU AGREE PLEASE WRITE IN THE DETAILS HERE. 
\textcolor{blue}{Yea, I agree, I am adding the details.}
*******************
\end{comment}

\section{Superconformal Symmetry And The Moonshine Module}\label{sec:supconfbb}

In this section, we will first review the Monster VOA $\widetilde{\mathscr{H}}(\Lambda_{\mathrm{L}})$ as the $\mathbb{Z}_2$-orbifold model of the Leech lattice $\Lambda_{\mathrm{L}}$ and then, following \cite{Lin:2019hks}, interpret 
 the ``Beauty and the Beast'' theory of \cite{Dixon:1988qd} as a spin lift of the Monster theory. We next identify a SVOA subalgebra of $\mathscr{H}_{BB}$ and then, 
 in section \ref{subsec:ConstructTF} we 
 proceed to the main result of the paper: The explicit construction of a dimension $\frac{3}{2}$ primary inside that SVOA which can serve as a superconformal current.

\subsection{The Beauty And The Beast Theory As The Spin Lift Of The FLM VOA }\label{subsec:BB}

The FLM Moonshine module is constructed, as we have reviewed, by applying 
 Theorem \ref{thm:latvoaext} part (2) to the Leech lattice. The construction 
 can be interpreted as gauging the Leech torus theory using the $\mathbb{Z}_2$ automorphism group generated by $\theta$. The value of the gauged theory on 
 the bundle $P_{\bar 1}$ is the $\theta = +1$ subspace of $\mathscr{H}_T(\Lambda_L)$. 
The quantum automorphism group dual to $\langle \theta \rangle$ (see discussion around \eqref{eq:pontdualactonHgd}) is the $\mathbb{Z}_2$ automorphism group of $\widetilde{\mathscr{H}}(\Lambda_{\mathrm{L}})$ generated by 
$\iota$, where $\iota$ is the identity on the untwisted sector $\mathscr{H}^+(\Lambda_{\mathrm{L}})$ and acts as $-1$ on the twisted sector $\mathscr{H}^+_T(\Lambda_{\mathrm{L}})$.

%
%
%the twisted VOA $\widetilde{\mathscr{H}}(\Lambda_{\mathrm{L}})$, as described  on the %Leech lattice $\Lambda_{\mathrm{L}}:=\widetilde{\Lambda}_{\mathcal{G}}$ (see eq. %\eqref{eq:twistlatticecode}) where $\mathcal{G}\subset\mathbb{F}_2^{24}$ is the Golay %code. States with half-integral conformal weight and the conformal dimension one %states are all projected out by the restriction to the $\theta=1$ subspace.   The %partition function of the theory is the famous Klein $J$ function with constant term %zero, i.e. $J(q)=j(q)-744$ where 
%\begin{equation}
%    j(q)=\frac{1}{q}+744+196884q+\dots
%\end{equation}
%is the $j$-invariant. The automorphism group of the VOA is, famously, the Monster %group \cite{Frenkel:1988xz}.
%

Following \cite{Dixon:1988qd} the FLM Moonshine module can be extended by considering the Hilbert space  
\begin{equation}\label{eq:BB-HilbertSpace}
    \mathscr{H}_{BB}:=\mathscr{H}(\Lambda_{\mathrm{L}})\oplus\mathscr{H}_T(\Lambda_{\mathrm{L}}).
\end{equation} 
%
%
%\\\\Let us now see why $\mathscr{H}_{BB}$ is a spin lift of the Monster VOA %$\widetilde{\mathscr{H}}(\Lambda_{\mathrm{L}})$ (see Theorem \ref{thm:latvoaext} (2)). %The Monster VOA has an obvious $\mathbb{Z}_2$ automorphism $\iota$ acting as the %identity on $\mathscr{H}^+(\Lambda_{\mathrm{L}})$ and as multiplication by $-1$ on $\mathscr{H}^+_T(\Lambda_{\mathrm{L}})$. 
%
The spin lift of the Monster VOA with respect to the quantum symmetry $\iota$ of the Monster theory gives $\mathscr{H}_{BB}$ \cite{Lin:2019hks}
and we will now review why this is so. 
\footnote{Incidentally, there is another nontrivial $\mathbb{Z}_2$ automorphism $\sigma$ of the Monster module called the \textit{triality involution} in  \cite{Frenkel:1988xz,Dolan:1994st}.  The involution $\sigma$ mixes the untwisted and twisted sectors. It is also shown in \cite{Lin:2019hks} that the 
spin lift of the Monster VOA with respect to $\sigma$ gives the tensor product of the Baby Monster VOA \cite{hoehn2007selbstduale} and a Majorana-Weyl CFT.}
The $\check{\sigma}_d=+1$ NS and R sector of the spin lift corresponding to $\iota$ is simply given by 
eq.  \eqref{eq:hilspSL}:
\begin{equation}\label{eq:NSR+BB}
\mathscr{H}^+_{\text{NS}}=\mathscr{H}^+(\Lambda_{\mathrm{L}}),\quad \mathscr{H}^+_{\text{R}}=\mathscr{H}^+_{T}(\Lambda_{\mathrm{L}}).
\end{equation}
%The role of $\text{Arf}(S^1_\mathrm{NS})$ and $\text{Arf}(S^1_\mathrm{R})$ is just to make $\mathscr{H}^{+}(\Lambda_{\mathrm{L}})$ an even and $\mathscr{H}^{+}_T(\Lambda_{\mathrm{L}}) $ an odd subspace respectively and define a super-vector space structure on $\mathscr{H}^{+}(\Lambda_{\mathrm{L}})\oplus\mathscr{H}^{+}_T(\Lambda_{\mathrm{L}})$.
The twisted sector corresponding to the gauging of $\iota$ can be determined using \eqref{eq:defecthilmodtrans}. The right hand side of \eqref{eq:defecthilmodtrans} is the McKay Thompson series 
\begin{equation}
\begin{split}
    Z^\iota_{\widetilde{\mathscr{H}}(\Lambda_{\mathrm{L}})}(\tau)&=\text{Tr}_{\widetilde{\mathscr{H}}(\Lambda_{\mathrm{L}})}\iota q^{L_0-1}=q^{-1}\text{Tr}_{\mathscr{H}^+(\Lambda_{\mathrm{L}})} q^{L_0}-q^{-1}\text{Tr}_{\mathscr{H}^+_T(\Lambda_{\mathrm{L}})} q^{L_0^T}\\&=\frac{1}{2}\left[\frac{\Theta_{\Lambda_{\text{L}}}(\tau)}{\eta(\tau)^{24}}+q^{-1}\prod_{n=1}^{\infty}(1+q^n)^{-24}\right]-\frac{2^{12}q^{\frac{1}{2}}}{2}\left[\prod_{n=1}^{\infty}(1-q^{n-\frac{1}{2}})^{-24}-\prod_{n=1}^{\infty}(1+q^{n-\frac{1}{2}})^{-24}\right]\\&=\frac{1}{2}\frac{\Theta_{\Lambda_{\text{L}}}(\tau)}{\eta(\tau)^{24}}+\frac{1}{2\eta(\tau)^{24}}\left[(\vartheta_3\vartheta_4)^{12}-(\vartheta_2\vartheta_3)^{12}+(\vartheta_2\vartheta_4)^{12}\right]\\&=\frac{1}{2}(J(\tau)+24)-\frac{1}{2}(J(\tau)-24)+\frac{(\vartheta_3\vartheta_4)^{12}}{\eta(\tau)^{24}}\\&=\frac{\eta(\tau)^{24}}{\eta(2 \tau)^{24}}+24\\&=\frac{1}{q}+276 q-2048 q^2+11202 q^3+\mathcal{O}\left(q^4\right)
\end{split}    
\end{equation}
where we used the infinite product expansions for the Jacobi theta functions:
\begin{equation}
\begin{split}  &\vartheta_2(q)=\frac{2 \eta^2(2\tau)}{\eta(\tau)}=2 q^{\frac{1}{8}}\prod_{n=1}^{\infty}\left(1-q^{ n}\right)\left(1+ q^{n} \right)^2 \\ & \vartheta_3(q)=\frac{\eta^5(\tau)}{\eta^2\left(\frac{1}{2} \tau\right) \eta^2(2 \tau)}=\prod_{n=1}^{\infty}\left(1-q^{n}\right)\left(1+ q^{n-\frac{1}{2}} \right)^2 \\ & \vartheta_4(q)=\frac{\eta^2\left(\frac{1}{2} \tau\right)}{\eta(\tau)}=\prod_{n=1}^{\infty}\left(1-q^{n}\right)\left(1- q^{n-\frac{1}{2}} \right)^2.
\end{split}    
\end{equation}
along with numerous identities such as $\vartheta_2 \vartheta_3 \vartheta_4 = 2\eta^3$ and $\vartheta_3^4 = \vartheta_2^4 + \vartheta_4^4$. 
Now we can take the  modular transform  \cite{Lin:2019hks}:
\begin{equation}
\begin{aligned}
Z^\iota_{\widetilde{\mathscr{H}}(\Lambda_{\mathrm{L}})}(-1/\tau)& =2^{12} \frac{\eta(\tau)^{24}}{\eta(\tau / 2)^{24}}+24 \\& =24+4096 q^{1 / 2}+98304 q+1228800 q^{3 / 2}+10747904 q^2+\mathcal{O}\left(q^{5 / 2}\right)\\&=\operatorname{Tr}_{\mathscr{H}^{\text{eq}}(P_{\bar{1}})}q^{L_0-1}\\&:=Z_{\mathscr{H}^{\text{eq}}(P_{\bar{1}})}(\tau),
\end{aligned}    
\end{equation}
%
%where $\bar{1}$ is the nontrivial conjugacy class of $\mathbb{Z}_2$.
%
On the other hand, we have
\begin{equation}
\begin{split}
\operatorname{Tr}_{\mathscr{H}^{-}(\Lambda_{\mathrm{L}})}&q^{L_0-1}+\operatorname{Tr}_{\mathscr{H}^{-}_T(\Lambda_{\mathrm{L}})}q^{L_0^T-1}\\&=\frac{1}{2}\left[\frac{\Theta_{\Lambda_{\text{L}}}(\tau)}{\eta(\tau)^{24}}-q^{-1}\prod_{n=1}^{\infty}(1+q^n)^{-24}\right]+\frac{2^{12}q^{\frac{1}{2}}}{2}\left[\prod_{n=1}^{\infty}(1-q^{n-\frac{1}{2}})^{-24}+\prod_{n=1}^{\infty}(1+q^{n-\frac{1}{2}})^{-24}\right]\\&=\frac{1}{2}\frac{\Theta_{\Lambda_{\text{L}}}(\tau)}{\eta(\tau)^{24}}+\frac{1}{2\eta(\tau)^{24}}\left[(\vartheta_2\vartheta_3)^{12}+(\vartheta_2\vartheta_4)^{12}-(\vartheta_3\vartheta_4)^{12}\right]\\&=\frac{\eta(\tau)^{24}}{\eta(2 \tau)^{24}}+24-\frac{1}{2\eta(\tau)^{24}}\left[2(\vartheta_2\vartheta_3)^{12}-2(\vartheta_3\vartheta_4)^{12}\right]\\&=2^{12} \frac{\eta(\tau)^{24}}{\eta(\tau / 2)^{24}}+24\\&=24+4096 q^{1 / 2}+98304 q+1228800 q^{3 / 2}+10747904 q^2+\mathcal{O}\left(q^{5 / 2}\right). 
\end{split}
\end{equation}
Thus we identify the nontrivial equivariant Hilbert space of the Monster module as 
\begin{equation}
\mathscr{H}^{\text{eq}}(P_{\bar{1}})=\mathscr{H}^{-}(\Lambda_{\mathrm{L}})\oplus\mathscr{H}^{-}_T(\Lambda_{\mathrm{L}}).    
\end{equation}
Next by eq.  \eqref{eq:P11fromP10} we have 
\begin{equation}
Z^{\iota}_{\mathscr{H}^{\text{eq}}(P_{\bar{1}})}(\tau):=\operatorname{Tr}_{\mathscr{H}^{\text{eq}}(P_{\bar{1}})}\iota q^{L_0-1}=Z_{\mathscr{H}^{\text{eq}}(P_{\bar{1}})}(\tau+1).  
\end{equation}
Note that we must have 
\begin{equation}
Z^\iota_{\mathscr{H}^{\text{eq}}(P_{\bar{1}})}(-1/\tau) =Z^\iota_{\mathscr{H}^{\text{eq}}(P_{\bar{1}})}(\tau)   
\end{equation}
for consistency. It is true and can be checked using the fact that the McKay Thompson series $Z^\iota_{\widetilde{\mathscr{H}}(\Lambda_{\mathrm{L}})}(\tau)$ is a Hauptmodul for $\Gamma_0(2).$ Indeed 
\begin{equation}
\begin{split}
Z^\iota_{\mathscr{H}^{\text{eq}}(P_{\bar{1}})}(-1/\tau) &=Z_{\mathscr{H}^{\text{eq}}(P_{\bar{1}})}\left(-\frac{1}{\tau}+1\right)= Z^\iota_{\widetilde{\mathscr{H}}(\Lambda_{\mathrm{L}})}\left(-\frac{\tau}{\tau-1}\right)=Z^\iota_{\widetilde{\mathscr{H}}(\Lambda_{\mathrm{L}})}\left(-\frac{\tau}{\tau-1}+1\right)\\&=Z^\iota_{\widetilde{\mathscr{H}}(\Lambda_{\mathrm{L}})}\left(-\frac{1}{\tau-1}\right)=Z_{\mathscr{H}^{\text{eq}}(P_{\bar{1}})}(\tau-1)=Z_{\mathscr{H}^{\text{eq}}(P_{\bar{1}})}(\tau+1)=Z^\iota_{\mathscr{H}^{\text{eq}}(P_{\bar{1}})}(\tau)   \end{split} 
\end{equation}
where we used the fact that $Z^\iota_{\widetilde{\mathscr{H}}(\Lambda_{\mathrm{L}})}(\tau)$ is invariant under $\tau\to\tau+1$ and $Z_{\mathscr{H}^{\text{eq}}(P_{\bar{1}})}(\tau)$ is invariant under $\tau\to\tau+2$.
From eq.  \eqref{eq:trNS-R-} we see that 
\begin{equation}
\begin{split}
\text{Tr}_{\mathscr{H}_{\text{NS}}^{-}}q^{L_0-\frac{c}{24}}&=\frac{1}{2}\left[Z_{\mathscr{H}^{\text{eq}}(P_{\bar{1}})}(\tau)-Z^\iota_{\widetilde{\mathscr{H}}(\Lambda_{\mathrm{L}})}(\tau)\right] \\&=\frac{2^{12}}{2}\left[ \frac{\eta(\tau)^{24}}{\eta(\tau / 2)^{24}}-\frac{\eta(\tau)^{24}}{\eta\left(\frac{\tau+1} {2}\right)^{24}}\right]\\&=\frac{2^{12}q^{\frac{1}{2}}}{2}\left[\prod_{n=1}^{\infty}(1-q^{n-\frac{1}{2}})^{-24}+\prod_{n=1}^{\infty}(1+q^{n-\frac{1}{2}})^{-24}\right]\\&=\text{Tr}_{\mathscr{H}_{T}^{-}(\Lambda_{\mathrm{L}})}q^{L_0-\frac{c}{24}},
\end{split}   
\end{equation}
where we used the infinite product expansion of the Dedekind eta function:
\begin{equation}
    \eta(\tau)=q^{\frac{1}{24}}\prod_{n=1}^\infty(1-q^n)^{24}.
\end{equation}
%where we used the fact that $\iota=1$ on $\mathscr{H}^{-}(\Lambda_{\mathrm{L}})$ and $-1$ on $\mathscr{H}^{-}_T(\Lambda_{\mathrm{L}})$ and the fact that if the spin structure on the spatial slice $S^1$ is NS then the spin structure on the total space $\widetilde{S}^1$ of the bundle is R and vice versa. Again the role of $\text{Arf}(\widetilde{S}^1_\mathrm{NS})$ and $\text{Arf}(\widetilde{S}^1_\mathrm{R})$ is just to define a super-vector space structure on $\mathscr{H}^{-}(\Lambda_{\mathrm{L}})\oplus\mathscr{H}^{-}_T(\Lambda_{\mathrm{L}})$. 
Thus we identify
\begin{equation}\label{eq:NSR-BB}
\mathscr{H}^-_{\text{NS}}=\mathscr{H}^-_T(\Lambda_{\mathrm{L}}).    
\end{equation}
This also confirms that $\iota$ acts as $-1$ on $\mathscr{H}^-_T(\Lambda_{\mathrm{L}})$. Similarly one can check that 
\begin{equation}
\mathscr{H}^-_{\text{R}}=\mathscr{H}^-(\Lambda_{\mathrm{L}})    
\end{equation}
which also implies that $\iota$ acts as $+1$ on $\mathscr{H}^-(\Lambda_{\mathrm{L}}) $.
This completes the demonstration (at least at the level of state spaces) that the spin lift of the Monster CFT with respect to $\iota$ is the Beauty and the Beast SCFT 
\eqref{eq:BB-HilbertSpace}.

To summarize, we have
\begin{equation}
\mathscr{H}_{\mathrm{NS}}= \mathscr{H}^+(\Lambda_{\mathrm{L}})\oplus \mathscr{H}_T^-(\Lambda_{\mathrm{L}}),\quad \mathscr{H}_{\mathrm{R}}=\mathscr{H}_T^+(\Lambda_{\mathrm{L}})\oplus \mathscr{H}^-(\Lambda_{\mathrm{L}})    
\end{equation}
in agreement with \eqref{eq:NSR+BB} and \eqref{eq:NSR-BB}. 
The superscripts on the RHS refer to the $\theta$ eigenvalue. 
These formulae can be refined to:
\begin{equation}
\mathscr{H}_{\mathrm{NS}}^+\cong \mathscr{H}^+(\Lambda_{\mathrm{L}}),~~\mathscr{H}_{\mathrm{NS}}^-\cong \mathscr{H}_T^-(\Lambda_{\mathrm{L}}),\quad  \mathscr{H}_{\mathrm{R}}^+\cong \mathscr{H}_T^+(\Lambda_{\mathrm{L}}),~~\mathscr{H}_{\mathrm{R}}^-\cong \mathscr{H}^-(\Lambda_{\mathrm{L}}),   
\end{equation}
where the superscripts on the LHS refer to the $(-1)_q^F$ eigenvalue. Here $(-1)_q^F$ is the emergent automorphism $\check{\sigma}_d$ (see discussion around \eqref{eq:FactonSL}).

As an aside note that a corollary of the above discussion is that the $\iota$-symmetry of the Monster theory is the quantum $\mathbb{Z}_2$ symmetry of the $\theta$-gauged Leech torus theory, and if we gauge the $\iota$-symmetry of the Monster theory we recover the Leech torus theory, as expected on general grounds. Similarly, if we apply the GSO projection to the BB theory we obtain the Monster theory.

Let us now begin the search for a SVOA structure within $\mathscr{H}_{BB}$. 
We stress that this space is not itself a SVOA.

As a preliminary, note that the vertex operators on $\mathscr{H}_{BB}$ can be described using the block form 
relative to the direct sum: 
\begin{equation}
\widetilde{V}(\psi,z)=\begin{pmatrix}
V(\psi,z)&0\\0&V_T(\psi,z)
\end{pmatrix},\quad \psi\in\mathscr{H}(\Lambda_{\mathrm{L}}),
\label{eq:veropuntbb}
\end{equation}
where $V(\psi,z)$ and $V_T(\psi,z)$ for homogeneous states are defined in eq.  \eqref{eq:veropuntwisted} and eq.  \eqref{eq:twistverop} respectively and the extension to general states is  defined by linearity. Vertex operators for twisted states take the form: 
\begin{equation}
\widetilde{V}(\chi,z)=\begin{pmatrix}
0&\overline{W}(\chi,z)\\W(\chi,z)&0
\end{pmatrix},\quad \chi\in\mathscr{H}_T(\Lambda_{\mathrm{L}}).
\label{eq:veroptbb}
\end{equation}
Here $W$ and $\overline{W}$ are the intertwiners associated to the representation $\mathscr{H}_T(\Lambda_{\mathrm{L}})$ of the VOA $\mathscr{H}^+(\Lambda_{\mathrm{L}})$ (see Theorem \ref{thm:latvoaext}) defined by the relations in eq.  \eqref{eq:defW} and \eqref{eq:defWbar}. 
The stress tensor is defined as 
\begin{equation}
    T(z):=\widetilde{V}(\psi_L,z)=\begin{pmatrix}V(\psi_L,z)&0\\0&V_T(\psi_L,z)\end{pmatrix}.
\label{eq:}
\end{equation}

\begin{prop}
The OPE of a 
vertex operator corresponding to a state in $\mathscr{H}_T^-(\Lambda_{\mathrm{L}})$ with a vertex operator corresponding to $\psi\in\mathscr{H}^-(\Lambda_{\mathrm{L}})$ or $\chi\in\mathscr{H}_T^+(\Lambda_{\mathrm{L}})$ has square root singularity and the OPE with vertex operator corresponding to $\psi\in\mathscr{H}^+(\Lambda_{\mathrm{L}})$ or $\chi\in\mathscr{H}_T^-(\Lambda_{\mathrm{L}})$ has integer power singularity. 
\label{prop:sqrtsingope}
\begin{proof}
Consider the vertex operator $\widetilde{V}(\chi,z),~\chi\in\mathscr{H}_T^-(\Lambda_{\mathrm{L}})$. We first analyse its OPE with an untwisted vertex operator $\widetilde{V}(\psi,z)$ for $\psi\in\mathscr{H}^-(\Lambda_{\mathrm{L}}).$ Note that any vector in $\mathscr{H}^-(\Lambda_{\mathrm{L}})$ is a linear combination of states of the form: 
\begin{equation}
   \Psi_{\lambda}=\left(\prod_{a=1}^M a^{j_a}_{-m_a}\right)\left(\frac{|\lambda\rangle\pm|-\lambda\rangle}{\sqrt{2}}\right)
\end{equation}
where we choose $+$ sign when $M$ is odd and $-$ sign when $M$ is even. It is easy to see that the Laurent expansion of $V_T(\Psi_{\lambda},z)$ contains $\mathbb{Z}+\frac{1}{2}$ powers of $z$. Indeed when $M$ is odd, the derivatives of the fields $R^{j_a}$ in the expression \eqref{eq:twistverop} gives $\mathbb{Z}+\frac{1}{2}$ powers of $z$ while the exponential factor in \eqref{eq:twistverop} combines to give $\cosh (i\lambda\cdot R)$ which still contributes integer powers of $z$ and hence the final vertex operator contains $\mathbb{Z}+\frac{1}{2}$ powers of $z$. Similarly when $M$ is even the derivatives of the fields $R^{j_a}$ contains integer powers of $z$ while the exponentials combine to give $\sinh (i\lambda\cdot R)$ which contributes $\mathbb{Z}+\frac{1}{2}$ powers of $z$ and hence the final vertex operator contains $\mathbb{Z}+\frac{1}{2}$ powers of $z$. Thus we can expand as 
\begin{equation}
    V_T(\Psi_{\lambda},z)=\sum_{n\in \mathbb{Z}+\frac{1}{2}}\left(V_T(\Psi_{\lambda})\right)_{n}z^{-n-h_{\Psi_{\lambda}}}.
\end{equation}
Now the duality of this vertex operator with $\widetilde{V}(\chi,z)$ implies 
\begin{equation}
\begin{split}
\widetilde{V}(\Psi_{\lambda},z)\widetilde{V}(\chi,w)&=\sum_{n\in\mathbb{Z}+\frac{1}{2}}(z-w)^{n-h_{\Psi_{\lambda}}}\widetilde{V}(\chi_n,w)\\&=  \sum_{n\in\mathbb{Z}}(z-w)^{n-h_{\Psi_{\lambda}}-h_{\chi}}\widetilde{V}(\widetilde{\chi}_n,w)  
\end{split}
\label{eq:opeh-ht-}
\end{equation}
where 
\begin{equation}
    \widetilde{\chi}_n=\left(V_T(\Psi_{\lambda})\right)_{h_{\chi}-n}\chi
    \label{eq:chitilden}
\end{equation}
Noting that $h_{\chi}\in\mathbb{Z}+\frac{1}{2}$ and $h_{\Psi_{\lambda}}\in\mathbb{Z}$, we see that the OPE has square root singularity. Similarly if $\widetilde{V}(\chi',z),~\chi'\in\mathscr{H}^+_T(\Lambda_{\mathrm{L}})$ is a twisted vertex operator then $\chi'$ has the form \eqref{eq:chitwistcft} with $M$ even. The duality now gives 
\begin{equation}
 \widetilde{V}(\chi',z)\widetilde{V}(\chi,w)=\widetilde{V}(\overline{W}(\chi',z-w)\chi,w),   
\end{equation}
where we used the fact that 
\begin{equation}
    \widetilde{V}(\chi',z-w)\chi=\overline{W}(\chi',z-w)\chi.
\end{equation}
The Laurent expansion of $\overline{W}(\chi',z)$ is of the form 
\begin{equation}
  \overline{W}(\chi',z)=\sum_n\overline{W}_n(\chi')z^{-n-h_{\chi'}}  
\end{equation}
Thus we get 
\begin{equation}
    \widetilde{V}(\chi',z)\widetilde{V}(\chi,w)=\sum_n(z-w)^{n-h_{\chi}-h_{\chi'}}\widetilde{V}(\widetilde{\chi}(n),w)
\end{equation}
where 
\begin{equation}
    \widetilde{\chi}(n)=\overline{W}_{n-h_{\chi}}(\chi')\chi.
\end{equation}
Now using \eqref{eq:momcomlm}, we have
\begin{equation}
L_0\overline{W}_n(\chi)-\overline{W}_n(\chi)L_0^T=-n\overline{W}_n(\chi)
\end{equation}
This implies that  $\overline{W}_n(\chi')\chi=0$ for $n\in\mathbb{Z}$ since the untwisted sector has integer conformal weights. Thus we see that 
\begin{equation}
    \widetilde{V}(\chi',z)\widetilde{V}(\chi,w)=\sum_{n\in\mathbb{Z}}(z-w)^{n-h_{\chi}-h_{\chi'}}\widetilde{V}(\widetilde{\chi}(n),w)
\end{equation}
Again since $h_{\chi'}\in\mathbb{Z}$ and $h_{\chi}\in\mathbb{Z}+\frac{1}{2}$, we see that the OPE has square-root singularity. The second part of the statement follows by similar arguments. 
\end{proof}
\end{prop}

We will now identify a subspace of $\mathscr{H}_{BB}$ which forms an SVOA. To get an SVOA, we anticipate that we need to include the $\theta=-1$ subspace of the twisted sector. But then from Proposition \ref{prop:sqrtsingope}, we see that we can neither have the $\theta=+1$ subspace of the twisted sector nor the $\theta=-1$ subspace of the untwisted sector. Thus we are led to the following theorem which was first proved in \cite{huang_1996} in the standard formalism of VOAs.
\begin{thm}\label{thm:BBSVOA}
The subspace
\begin{equation}\label{eq:Hscdef}
\mathscr{H}_{SC}:=\mathscr{H}_{\mathrm{NS}}=\mathscr{H}^+(\Lambda_{\mathrm{L}})\oplus\mathscr{H}^-_T(\Lambda_{\mathrm{L}})   
\end{equation} 
of the ``Beauty and the Beast'' SCFT possesses the structure of an SVOA.
\begin{proof}
We first put a super-vector space structure on $\mathscr{H}_{SC}$ by declaring $\mathscr{H}^+(\Lambda_{\mathrm{L}})$ to the even space and $\mathscr{H}^-_T(\Lambda_{\mathrm{L}})$ to be the odd space. We will call them respectively the untwisted (bosonic) and twisted (fermionic) sector. The Fock space of $\mathscr{H}_{SC}$ is 
\begin{equation}
\mathscr{F}_{SC}=\mathscr{F}^+\oplus \mathscr{F}_T^-
\end{equation}
where $\mathscr{F}$ and $\mathscr{F}_T$ is the untwisted and twisted Fock spaces respectively and the superscript $\pm 1$ indicates the $\theta=\pm 1$ subspaces of $\mathscr{F},\mathscr{F}_T$. The conformal and vacuum vector is as in the BB module $\mathscr{H}_{BB}$ and the vertex operators are also defined analogous to the BB module. 
We now need to prove the locality relation  
\begin{equation}
    \widetilde{V}(\psi,z)\widetilde{V}(\phi,w)=(-1)^{|\psi||\phi|}\widetilde{V}(\phi,w)\widetilde{V}(\psi,z),\quad \psi,\phi\in\mathscr{H}_{SC}.
\label{eq:locHsc}
\end{equation}
For bosonic states $\psi,\phi$, the locality relation eq.  \eqref{eq:locHsc} is equivalent to the locality relations
\begin{equation}
    V(\psi,z)V(\phi,w)=V(\phi,w)V(\psi,z),\quad V_T(\psi,z)V_T(\phi,w)=V_T(\phi,w)V_T(\psi,z).
\label{eq:}
\end{equation}
The first relation is proved in \cite[Section 5, 6]{Dolan:1989vr} and the second follows from \cite[Proposition 3.2 (1)]{Dolan:1994st}. This implies the duality of the untwisted vertex operators
\begin{equation}
   \widetilde{V}(\psi,z)\widetilde{V}(\phi,w)=\widetilde{V}(\widetilde{V}(\psi,z-w)\phi,w).
   \label{eq:dualityuntwist}
\end{equation}
The locality of an untwisted and a twisted vertex operator is 
\begin{equation}
    \widetilde{V}(\psi,z)\widetilde{V}(\chi,w)=\widetilde{V}(\chi,w)\widetilde{V}(\psi,z)
\label{eq:}
\end{equation}
This follows from the locality relations 
\begin{equation}
    V(\psi,z)\overline{W}(\chi,w)=\overline{W}(\chi,w)V_T(\psi,z),\quad V_T(\psi,z)W(\chi,w)=W(\chi,w)V(\psi,z).
\label{eq:}
\end{equation}
The first relation is proved in \cite[Section 6]{Dolan:1989vr} while the second follows from Proposition \ref{prop:reploc} and Theorem \ref{thm:latvoaext} (1). 
This implies the duality of the twisted-untwisted vertex operators
\begin{equation}
   \widetilde{V}(\psi,z)\widetilde{V}(\chi,w)=\widetilde{V}(\widetilde{V}(\psi,z-w)\chi,w).
   \label{eq:dualityuntwisttwist}
\end{equation}
%The locality relation for bosonic twisted vertex operators with fermionic twisted vertex operators is proved in \cite{Dolan:1989vr}. 
The locality relation for twisted vertex operators with 
$\chi,\chi' \in \mathscr{H}_T^-(\Lambda_{\text{L}})$ is
\begin{equation}
    \widetilde{V}(\chi,z)\widetilde{V}(\chi',w)=-\widetilde{V}(\chi',w)\widetilde{V}(\chi,z).
\label{eq:}
\end{equation}
For this, we need to prove the locality relations 
\begin{equation}
\begin{split}
\overline{W}(\chi,z)W(\chi',w)=-\overline{W}(\chi',w)W(\chi,z),\\W(\chi,z)\overline{W}(\chi',w)=-W(\chi',w)\overline{W}(\chi,z),
\end{split}
\label{eq:locWbarWWWbar}
\end{equation}
for $\chi,\chi'\in\mathscr{H}_T^-(\Lambda_{\mathrm{L}}).$ 
These relations, especially the second, are highly nontrivial. It is  proven  in Appendix \ref{App:TwistedLocality} that the locality relations are satisfied if the conjugation operation on $\mathscr{H}_T^+(\Lambda_{\mathrm{L}})$ is extended to $\mathscr{H}_T(\Lambda_{\mathrm{L}})$ as follows:
\begin{equation}
\overline{\chi}=e^{-i\pi L_0^T}\theta M(\chi),
\label{eq:modconjdef}
\end{equation}
where $\chi$  is a general homogeneous state as in \eqref{eq:chitwistcft} in $\mathscr{H}_T^-(\Lambda_{\mathrm{L}})$ and the conjugation is extended to all other states by antilinearity. It is easy to check that this modification still satisfies $\overline{\overline{\chi}}=\chi$ using the fact that $\overline{z\chi}=z^*\overline{\chi}$ for a complex number number $z$. All other axioms in Definition \ref{defn:scft} follows from the results of \cite{Dolan:1989vr,Dolan:1994st}.
\end{proof}
\end{thm}
\begin{remark}\label{rem:abelian_int}
We include here a remark from one of the referees. Similar remarks were explained to us by S. Carnahan. 
Let us give a broader context on Theorem \ref{thm:BBSVOA}. The first proof of Theorem \ref{thm:BBSVOA} in the standard formalism of super vertex operator algebra appeared in  \cite{huang_1996}. Alternatively, the existence of the SVOA structure on $\mathscr{H}_{SC}$ follows in general from Theorem 3.9 in \cite{Creutzig_2019}.
This requires the simple current property of the odd module, which was worked out by Dong in \cite{dong1994representations}. This construction of (super) vertex operator algebras can be generalized vastly using the general theory of abelian intertwining algebras. Given a nice (called strongly rational in VOA literature) holomorphic vertex operator algebra $V$ and an automorphism $g$ of the VOA one can consider the fixed-point vertex operator algebra $V^g$. Then one can show that $V^g$ again satisfies the niceness assumptions \cite{dong_mason_1998,miyamoto_2015,carnahan_miyamoto_2016,mcrae_2020}. Then one can determine the representation category of $V^g$, which is a twisted Drinfeld double \cite{svenmoller,van_ekeren_moeller_scheithauer_2015} (see also \cite{MIYAMOTO200480,dong_ren_xu_2001}), essentially proving a conjecture in \cite{dijkgraaf_pasquier_roche_1991}. In particular, all irreducible $V^g$-modules are simple currents, which translates into the statement that on the direct sum of irreducible $V^{g}$-modules one gets the structure of an abelian intertwining algebra in the sense of \cite{dong_lepowsky_1993}. What this entails is that the extensions of $V^g$ to larger vertex operator (super)algebras (with the same conformal vector) can be fully described in terms of isotropic subgroups of discriminant forms (again, this is described in \cite{svenmoller,van_ekeren_moeller_scheithauer_2015} and essentially a special case of \cite{huang_kirillov_lepowsky_2003}).
Specializing to order 2 and assuming that $g$ is anomaly-free, the structure of the representation category (and hence of the extensions) of $V^g$ is described by the abelian group $\mathbb{Z}_2 \times \mathbb{Z}_2$ with quadratic form $0,0,0,1 / 2$. 
The situation is now very simple: $V^g$ has three non-trivial extensions, all of which are holomorphic. One is the original vertex operator algebra, one is the orbifolded vertex operator algebra and the third is a vertex operator superalgebra. When we apply this general result to the Monster VOA, then the two holomorphic vertex operator algebras are the Leech lattice vertex operator algebra $\mathscr{H}(\Lambda_{\mathrm{L}})$ and the moonshine module $\widetilde{\mathscr{H}}(\Lambda_{\mathrm{L}})$ (usually denoted by $V^{\natural}$ in literature), and the holomorphic vertex operator superalgebra is the SVOA of Theorem \ref{thm:BBSVOA}. 
\end{remark}

In the next section, we will finally  construct a superconformal vector inside $\mathscr{H}_{SC}$ making it into an $\mathcal{N}=1$ SCVOA. 
\subsection{Constructing The Supercurrent}\label{subsec:ConstructTF}

Our supercurrent will be a vertex operator associated to a twisted sector 
ground state   $\mathcal{S}(\Lambda_{\mathrm{L}})$, which is 
the $2^{12}$-dimensional representation of the Heisenberg group \footnote{Also known in the finite group theory literature as an extraspecial group and denoted by $2_+^{1+24}$.} $\Gamma(\Lambda_{\mathrm{L}})$ defined analogous to \eqref{eq:gammamatalg}.
%
%Let $\{\chi_a\}_{a=1}^{2^{12}}$ be an orthonormal basis of $\mathcal{S}(\Lambda_{\mathrm{L}})$ and %$\chi=\sum_as_a\chi_a$ for $s_a\in\mathbb{C}$ be a general twisted ground state.
%
For any $\chi\in\mathcal{S}(\Lambda_{\mathrm{L}})$, the corresponding vertex operator (see Theorem \ref{thm:extensionofcft} for notations) has the block diagonal form relative to the direct sum decomposition $\mathscr{H}_{\text{BB}} = \mathscr{H}(\Lambda_{\text{L}}) \oplus  \mathscr{H}_{T}(\Lambda_{\text{L}})$: 
\begin{equation}
\widetilde{V}(\chi,z)=\begin{pmatrix}
0&\overline{W}(\chi,z)\\W(\chi,z)&0
\end{pmatrix},\quad \chi\in\mathcal{S}(\Lambda_{\mathrm{L}})
\end{equation}
are well-defined and are present in the extended theory $\mathscr{H}_{BB}$.
We Laurent expand the vertex operators $\widetilde{V}(\chi,z):$
\begin{equation}
\widetilde{V}(\chi,z)=\sum_{n}\widetilde{V}_{n}(
\chi)z^{-n-\frac{3}{2}},
\end{equation}
where 
\begin{equation}
\widetilde{V}_{n}(\chi)=\begin{pmatrix}
0&\overline{W}_{n}(\chi)\\W_n(\chi)&0
\end{pmatrix}
\end{equation}

and 
\begin{equation}
\overline{W}(\chi,z)=\sum_{n}\overline{W}_{n}(\chi)z^{-n-\frac{3}{2}},\quad W(\chi,z)=\sum_{n}W_{n}(\chi)z^{-n-\frac{3}{2}}
\label{eq:WWbarmodes}
\end{equation}
\begin{prop}
The twisted vertex operators $\widetilde{V}(\chi,z)$ are conformal primary operators of weight $3/2$. 
\begin{proof}
First note that the states $\chi$ are of conformal weight $3/2$. We now compute the OPE of $\widetilde{V}(\chi,z)$ with the energy momentum tensor $\widetilde{V}(\psi_L,z)$. From the locality of all vertex operators and intertwiners and using the general form \eqref{eq:genope} of the OPE of vertex operators, we get 
\begin{equation}
\widetilde{V}(\psi_L,z)\widetilde{V}(\chi,w)=\sum_{\substack{n\in\mathbb{Z}+\frac{1}{2}\\n\geq 0}}(z-w)^{n-2-3/2}\widetilde{V}(\chi(n),w)
\end{equation}
  
where 
\begin{equation}
\chi(n)=L^T_{3/2-n}\chi.
\end{equation}

We easily see that 
\begin{equation}
\widetilde{V}(\psi_L,z)\widetilde{V}(\chi,w)=\frac{3/2}{(z-w)^{2}}\widetilde{V}(\chi,w)+\frac{1}{(z-w)}\partial_w\widetilde{V}(\chi,w)+\text{reg}.
\end{equation}

\end{proof}
\end{prop}
For any $\chi_a,\chi_b\in\mathcal{S}(\Lambda_{\mathrm{L}})$, we now compute the OPE of $\widetilde{V}(\chi_a,w)$ and $\widetilde{V}(\chi_b,w)$. We have 
\begin{equation}\label{eq:opexaxb}
\widetilde{V}(\chi_a,z)\widetilde{V}(\chi_b,w)=\sum_{\substack{n\geq 0}}(z-w)^{n-3}\widetilde{V}(\widetilde{\chi}_{ab}(n),w)
\end{equation}
where we  can identify 
\begin{equation}
    \widetilde{\chi}_{ab}(n)= \overline{W}_{3/2-n}(\chi_a)\chi_b
\end{equation}
because the formalism identifies 
%\begin{equation}
%|\widetilde{\chi}_{ab}(n)\rangle=\widetilde{V}_{3/2-n}(\chi_a)|\chi_b\rangle.
%\end{equation}
%Note that 
%\begin{equation}
%\begin{pmatrix} 
%|\widetilde{\chi}_{ab}(n)\rangle\\
%0 \\ 
%\end{pmatrix} = 
 %    \begin{pmatrix}
 %   0&\overline{W}_{3/2-n}(\chi_a)\\W_{3/2-n}(\chi_a)&0
 %   \end{pmatrix}
%\end{equation}
%Thus 
%
\begin{equation}
\begin{pmatrix} 
\widetilde{\chi}_{ab}(n)\\
0 \\ 
\end{pmatrix} = 
\widetilde{V}_{3/2-n}(\chi_a) \begin{pmatrix} 0 \\ \chi_b\\ \end{pmatrix} = 
    \begin{pmatrix}
    0&\overline{W}_{3/2-n}(\chi_a)\\W_{3/2-n}(\chi_a)&0
    \end{pmatrix}\begin{pmatrix}0\\\chi_b\end{pmatrix} . 
\label{eq:}
\end{equation}
Note that using \eqref{eq:momcomlm}, we have
\begin{equation}
L_0\overline{W}_n(\chi_a)-\overline{W}_n(\chi_a)L_0^T=-n\overline{W}_n(\chi_a)
\end{equation}
and hence: 
\begin{equation}
    L_0 \left( \overline{W}_n(\chi_a)\chi_b\right) + n \left( \overline{W}_n(\chi_a)\chi_b\right) = 3/2 \left( \overline{W}_n(\chi_a)\chi_b\right)
\end{equation}
which immediately implies that $\overline{W}_n(\chi_a)\chi_b=0$ for $n\in\mathbb{Z}$ since the untwisted sector has integer conformal weights. In particular, the OPE of any two twist fields $\widetilde{V}(\chi_a,z),\widetilde{V}(\chi_b,w)$ is local. We now restrict to the case $a=b$.
\begin{lemma}
The states $\overline{W}_{3/2}(\chi)\chi$ is $\mathfrak{su}(1,1)$-invariant and the states $\overline{W}_{1/2}(\chi)\chi$ and $\overline{W}_{-1/2}(\chi)\chi$ are of conformal weights $1$ and $2$ respectively. 
% \begin{proof}
% Using the general relation \eqref{eq:momcomlm}, we see that 
% \begin{equation}
% L_0\overline{W}_n(\chi)-\overline{W}_n(\chi)L_0^T=-n\overline{W}_n(\chi).
% \end{equation}
% This immediately implies that 
% \begin{equation}
% L_0(\overline{W}_{3/2}(\chi)\chi)=0.
% \end{equation}

% Thus $\overline{W}_{3/2}(\chi)\chi$ is of conformal weight 0 and hence $\mathfrak{su}(1,1)$-invariant.
% The other statement follows by using the above operator relation for $n=1/2,-1/2.$
% \end{proof}
\end{lemma}
\begin{prop}
We have that 
\begin{equation}\label{eq:barW-1/2chi}
\begin{split}
\overline{W}_{3/2}(\chi)\chi&=\alpha(\chi)|0\rangle,\\\overline{W}_{1/2}(\chi)\chi&=0,\\ \overline{W}_{-1/2}(\chi)\chi&=\frac{\alpha(\chi)}{8}\psi_L+\sum_{\lambda^2=4}\kappa_{\lambda}(\chi)|\lambda\rangle
\end{split}
\end{equation}
% \begin{equation}\label{eq:barW-1/2chi}
% \begin{split}
% \overline{W}_{3/2}(\chi)\chi&=\alpha(\chi)|0\rangle,\\\overline{W}_{1/2}(\chi)\chi&=2\sum_{i}\kappa_ia^{i}_{-1}|0\rangle,\\ \overline{W}_{-1/2}(\chi)\chi&=\frac{\alpha(\chi)}{8}\psi_L+\sum_{i}\kappa_ia^{i}_{-2}|0\rangle+\sum_{\lambda^2=4}\kappa_{\lambda}(\chi)|\lambda\rangle
% \end{split}
% \end{equation}
% for some coefficients $\kappa_{i}$,
where\footnote{The sign of $\kappa_{\lambda}(\chi)$   depends on the choice of lift of $\lambda$ to the Heisenberg group. Given such a lift, $\kappa_{\lambda}(\chi)$ depends only on the class $[\lambda]\in\Lambda_{\mathrm{L}}/2\Lambda_{\mathrm{L}}$. } 
\begin{equation}\label{eq:alp_kappa_chi}
\alpha(\chi)=\langle\overline{\chi}|\chi\rangle,\quad \kappa_{\lambda}(\chi)=16\langle\chi|e_{\lambda}|\overline{\chi}\rangle^*.     
\end{equation}
and
\begin{proof}
%\textcolor{red}{Clarification: I think we have to choose a section of the Heisenberg extension in \eqref{eq:centextheis} to define the twisted vertex operators in \eqref{eq:veroptbb}. So $\kappa_{\lambda}$ indeed depends on the choice of the lift.}
By uniqueness of an $\mathfrak{su}(1,1)$-invariant state we have
\begin{equation}
    \overline{W}_{3/2}(\chi)\chi=\alpha(\chi)|0\rangle.
\label{eq:}
\end{equation}
Taking the inner product with $\langle 0|$ and using the fact that $\overline{W}_{3/2}(\chi)^{\dagger}=W_{-3/2}(\overline{\chi})$ (see Proposition \ref{prop:propcft} (3)), we get
\begin{equation}
\alpha(\chi)=\langle\overline{\chi}|\chi\rangle.
\label{eq:alphachi}
\end{equation}
Next, we prove the third relation of \eqref{eq:barW-1/2chi}. Since $\overline{W}_{-1/2}(\chi)\chi$ has conformal dimension 2, we can expand it in a complete basis of states of conformal dimension 2 as: 
\begin{equation}\label{eq:Wbar1/2xexp}
\overline{W}_{-1/2}(\chi)\chi=\sum_{i}\kappa_{ii} a_{-1}^ia_{-1}^i|0\rangle+\sum_{i<j}\kappa_{ij}a^{i}_{-1}a^{j}_{-1}|0\rangle+\sum_j\kappa_ja^j_{-2}|0\rangle+\sum_{\lambda^2=4}\kappa_{\lambda}(\chi)|\lambda\rangle
\end{equation}
Under the $\mathbb{Z}_2$-orbifold action (see \cite[eq. (3.11), eq. (3.27)]{Dolan:1989vr})
\begin{equation}
    \theta \widetilde{V}(\chi,z)\widetilde{V}(\chi,w)\theta^{-1}=\theta \widetilde{V}(\chi,z)\theta^{-1}\theta\widetilde{V}(\chi,w)\theta^{-1}=\widetilde{V}(\theta\chi,z)\widetilde{V}(\theta\chi,w)=\widetilde{V}(\chi,z)\widetilde{V}(\chi,w)~.
\label{eq:}
\end{equation}
Thus the LHS of \eqref{eq:opexaxb} is even. Thus the invariance of RHS of \eqref{eq:opexaxb} under $\theta$ implies that $W_n(\chi)\chi$ must be invariant under $\theta$ for all $n$. Observing that $a_n^j$ is anti-invariant under $\theta$ we conclude that $\kappa_i=0$.  
Next using \eqref{eq:momcomlm} for $m=2$ and $n=-1/2$, we get 
\begin{equation}
L_2\overline{W}_{-1/2}(\chi)-\overline{W}_{-1/2}(\chi)L_{2}^T=\frac{3}{2}\overline{W}_{3/2}(\chi).
\end{equation}
Applying it on $\chi$, we get 
\begin{equation}
L_2\overline{W}_{-1/2}(\chi)\chi=\alpha(\chi)\frac{3}{2}|0\rangle
\end{equation}

Using the above expansion of $\overline{W}_{-1/2}(\chi)\chi$, we get
\begin{equation}\label{eq:barW-1/2chisum}
\sum_i\kappa_{ii} L_2a^i_{-1}a^{i}_{-1}|0\rangle+\sum_{i<j}\kappa_{ij}L_2a^{i}_{-1}a^{j}_{-1}|0\rangle+\sum_{\lambda^2=4}\kappa_{\lambda}(\chi)L_2|\lambda\rangle=\alpha(\chi)\frac{3}{2}|0\rangle.
\end{equation}
Using the fact that 
\begin{equation}
L_{2}a^{i}_{-1}a^{j}_{-1}|0\rangle=\delta^{ij}|0\rangle,\quad L_2|\lambda\rangle=0, 
\end{equation}
we get
\begin{equation}
\sum_{i}\kappa_{ii}=\frac{3}{2}\alpha(\chi).
\label{eq:gammasum}
\end{equation}
We now find the explicit form of and $\kappa_{ii}$ and $\kappa_{ij}$. Taking inner product of $\langle 0|a^{m}_{1}a^{n}_{1}$ with $\overline{W}_{-1/2}(\chi)\chi$ and using the expression \eqref{eq:Wbar1/2xexp} for $\overline{W}_{-1/2}(\chi)\chi$, we get 
\begin{equation}
\begin{split}
\langle 0|a^{m}_{1}a^{n}_{1}\overline{W}_{-1/2}(\chi)|\chi\rangle&=\kappa_{nm};\quad n\neq m\\&=2\kappa_{mm};\quad n=m.
\end{split}
\end{equation}
where we used the fact that 
\begin{equation}
\langle 0|a^{m}_{1}a^{n}_{1}a_{-1}^ia_{-1}^j|0\rangle=\delta^{ni}\delta^{mj}+\delta^{mi}\delta^{nj},\quad \langle 0|a^{m}_{1}a^{n}_{1}a_{-2}^j|0\rangle=0,\quad a^i_1|\lambda\rangle =0.
\end{equation}
Note that  
\begin{equation}
\langle 0|a^{m}_{1}a^{n}_{1}\overline{W}_{-1/2}(\chi)|\chi\rangle=\langle \chi|W_{1/2}(\overline{\chi})a_{-1}^ma_{-1}^n|0\rangle^{*}.
\end{equation}
Using eq.  \eqref{eq:defW}, we get 
\begin{equation}
W(\overline{\chi},z)a_{-1}^ma_{-1}^n|0\rangle=e^{zL_{-1}^T}V_T(a_{-1}^ma_{-1}^n|0\rangle,-z)\overline{\chi} .
\end{equation}
To compute $V_T(a_{-1}^ma_{-1}^n|0\rangle,z)$, we need to calculate $e^{A(-z)}a_{-1}^ma_{-1}^n|0\rangle$ (see eq.  \eqref{eq:twistverop}), we have (see \cite[Page 92]{Dolan:1994st})
\begin{equation}
e^{A(-z)}a_{-1}^ma_{-1}^n|0\rangle=a_{-1}^ma_{-1}^n|0\rangle+\frac{1}{8}z^{-2}\delta^{nm}|0\rangle.
\end{equation}
This gives 
\begin{equation}
\begin{split}
W(\overline{\chi},z)a_{-1}^ma_{-1}^n|0\rangle&=-e^{zL_{-1}^T}:\left(\frac{dR^m(-z)}{dz}\frac{dR^n(-z)}{dz}+\frac{1}{8}z^{-2}\delta^{nm}\right):\overline{\chi}\\&=e^{zL_{-1}^T}\left(\sum_{r,s\in\mathbb{Z}+\frac{1}{2}}(-1)^{-r-s}:c_r^mc_s^n:z^{-r-s-2}+\frac{1}{8}\delta^{mn}z^{-2}\right)\overline{\chi}\\&=\left(\sum_{\substack{r,s\in\mathbb{Z}+\frac{1}{2}\\t\in\mathbb{N}_0}}(-1)^{-r-s}\frac{(L_{-1}^T)^t:c_r^mc_s^n:}{t!}z^{t-r-s-2}+\frac{\delta^{mn}}{8}\sum_{t\in\mathbb{N}_0}\frac{(L_{-1}^T)^t}{t!}z^{t-2}\right)\overline{\chi}.
\end{split}
\end{equation}
Thus we have
\begin{equation}
\begin{split}
\left(\sum_{r\in\mathbb{Z}+\frac{1}{2}}W_{r}(\overline{\chi})z^{-r-3/2}\right)a_{-1}^ma_{-1}^n|0\rangle=\left(\sum_{\substack{r,s\in\mathbb{Z}+\frac{1}{2}\\t\in\mathbb{N}_0}}(-1)^{-r-s}\frac{(L_{-1}^T)^t:c_r^mc_s^n:}{t!}z^{t-r-s-2}\right.\\+\left.\frac{\delta^{mn}}{8}\sum_{t\in\mathbb{N}_0}\frac{(L_{-1}^T)^t}{t!}z^{t-2}\right)\overline{\chi},
\end{split}
\end{equation}
which on comparing the powers of $z^{-2}$ gives 
\begin{equation}
\begin{split}
W_{1/2}(\overline{\chi})a_{-1}^ma_{-1}^n|0\rangle=&\left(\sum_{\substack{r\in\mathbb{Z}+\frac{1}{2}\\t\in\mathbb{N}_0}}\frac{(-L_{-1}^T)^t:c_r^mc_{t-r}^n:}{t!}+\frac{\delta^{mn}}{8}\right)\overline{\chi}.\\&=\frac{\delta^{mn}}{8}\overline{\chi}
\end{split}
\end{equation} 
This implies that 
\begin{equation}
\kappa_{mn}=0,\quad \kappa_{ii}=\frac{1}{16}\alpha(\chi)\quad\text{for all $i$}.
\end{equation}
Note that this result is consistent with eq.  \eqref{eq:gammasum}.
To compute $\kappa_{\lambda}(\chi)$, note that 
\begin{equation}
    \kappa_{\lambda}(\chi)=\langle\lambda|\overline{W}_{-\frac{1}{2}}(\chi)|\chi\rangle.
\label{eq:}
\end{equation}
Next using \eqref{eq:defW} we get 
\begin{equation}
W(\overline{\chi},z)|\lambda\rangle=e^{zL_{-1}^T}V_T(|\lambda\rangle,z)\overline{\chi}.   
\end{equation}
Since 
\begin{equation}
    e^{A(-z)}|\lambda\rangle=(-4z)^{-\frac{\lambda^2}{2}}|\lambda\rangle,
\end{equation}
we get
\begin{equation}
    W(\overline{\chi},z)|\lambda\rangle=e^{zL_{-1}^T}(-4z)^{-\frac{\lambda^2}{2}}:\exp(i\lambda\cdot R(-z)):e_{\lambda}\overline{\chi}
\label{eq:W1/2chibarlamb}
\end{equation}
Thus using the fact that $\lambda^2=4$ and comparing the coefficient of $z^{-2}$ on both sides of \eqref{eq:W1/2chibarlamb} we obtain 
\begin{equation}
    W_{\frac{1}{2}}(\overline{\chi})|\lambda\rangle=\text{coefficient of $z^{-2}$ in }\frac{16}{z^2}e^{zL_{-1}^T}:\exp(i\lambda\cdot R(-z)):e_{\lambda}\overline{\chi}.
\label{eq:}
\end{equation}
Because of the factor $16/z^2$, the only contribution to $z^{-2}$ on the RHS comes from the product of terms with $z^{\ell}$ in $e^{zL_{-1}^T}$ and $z^{-\ell}$ in $:\exp(i\lambda\cdot R(-z)):$ with $\ell\geq 0$. Note that $\ell$ must be an integer since $e^{zL_{-1}^T}$ only has integer powers of $z$ in its expansion. The quantity 
$e^{zL_{-1}^T}$ contributes $(L_{-1}^T)^\ell z^{\ell}$ but $:\exp(i\lambda\cdot R(-z)):$ may have nontrivial contributions. Write 
\begin{equation}
    :\exp(i\lambda\cdot R(-z)):=1+\sum_{n=1}^{\infty}\frac{1}{n!}:\left(i\sum_{k\in\mathbb{Z}+\frac{1}{2}}\frac{\lambda\cdot c_k}{k}(-z)^{-k}\right)^n:.
\label{eq:explambR-zexp}
\end{equation}
When $\ell=0$, only the constant term 1 in \eqref{eq:explambR-zexp} contributes. When $\ell>0$, to get a term with $z^{-\ell}$ from the sum, we have to consider $n>0$ terms in the sum. This involves the product of at least two oscillators $c_k$ since $\ell$ is an integer and the summation index $k$ runs over strictly half-integers. But then the normal ordering annihilates $e_{\lambda}\overline{\chi}$. So the only nontrivial contribution from the sum is 1 when $\ell=0$.
When $\ell>0$, there is no nontrivial contribution.
So we finally have 
\begin{equation}
    W_{\frac{1}{2}}(\overline{\chi})|\lambda\rangle=16e_{\lambda}\overline{\chi}.
\label{eq:}
\end{equation}
Thus 
\begin{equation}
\kappa_{\lambda}(\chi)=\langle\lambda|\overline{W}_{-\frac{1}{2}}(\chi)|\chi\rangle=\langle\chi|W_{\frac{1}{2}}(\overline{\chi})|\lambda\rangle^*=16\langle\chi|e_{\lambda}|\overline{\chi}\rangle^*.
\label{eq:kappalambda}
\end{equation}
%But note that since $\gamma_{\lambda}=\gamma_{-\lambda}$, we have $\kappa_{\lambda}=\kappa_{-\lambda}.$ 
%
Thus, we have 
\begin{equation}
\overline{W}_{-1/2}(\chi)\chi=\frac{\alpha(\chi)}{8}\psi_L+\sum_{\lambda^2=4}\kappa_{\lambda}|\lambda\rangle.
\end{equation}

It remains to prove the second equation in \eqref{eq:barW-1/2chi}.   By \eqref{eq:momcomlm} we have 
\begin{equation}
    L_{1}\overline{W}_{-1/2}(\chi)-\overline{W}_{-1/2}(\chi)L^T_1=\overline{W}_{1/2}(\chi).
\end{equation}
Applying this on $\chi$ and using the fact that $L_1^T\chi=0$, we obtain 
\begin{equation}
\begin{split}
\overline{W}_{1/2}(\chi)\chi&=L_{1}\overline{W}_{-1/2}(\chi)\chi\\&=\frac{\alpha(\chi)}{8}L_{1}\psi_L+\sum_{\lambda^2=4}\kappa_{\lambda}L_{1}|\lambda\rangle.
\end{split}
\end{equation}
Using the commutator $[a^i_m,L_n]=ma^i_{m+n}$, which is easy to check using the expression \eqref{eq:virgenuntwisted} for Virasoro generators, we see that $L_1\psi_L=0$. Moreover, $L_1|\lambda\rangle=0$. Thus we obtain 
\begin{equation}
\overline{W}_{1/2}(\chi)\chi=0.    
\end{equation}
This concludes the proof of \eqref{eq:barW-1/2chi} above. 
\end{proof}
\end{prop}
%

% It remains to determine the coefficients $\kappa_i$. In fact, we can argue that they are zero as follows. 
We use the OPE \eqref{eq:opexaxb} and the expression \eqref{eq:barW-1/2chi} to get: 
\begin{equation}
\widetilde{V}(\chi,z)\widetilde{V}(\chi,w)=\frac{\alpha(\chi)}{(z-w)^{3}}+\frac{\alpha(\chi)/8}{(z-w)}T_B(w)+\sum_{\lambda^2=4}\frac{\kappa_{\lambda}(\chi)}{z-w}\begin{pmatrix}e^{i\lambda\cdot X(w)}\sigma_{\lambda}&0\\0&e^{i\lambda\cdot R(w)}e_{\lambda}\end{pmatrix}.
\label{eq:opesimplified}
\end{equation}
Our goal is to provide a superconformal current. This will be achieved if the third term above 
vanishes. Therefore we now search for a twisted ground state $\chi$ such that $\alpha(\chi)=1$ and $\kappa_{\lambda}(\chi)=0$ for every $\lambda\in\Lambda_{\mathrm{L}}$ with $\lambda^2=4$. Choose an orthonormal basis of $\mathcal{S}(\Lambda_{\mathrm{L}})$. Then we have
\begin{equation}
\alpha(\chi)^*=\langle\chi|\overline{\chi}\rangle=-is^{\dagger}Ms^*,
\label{eq:alphachi} 
\end{equation}
where $s$ is the column vector representing $\chi$ in the chosen basis and and $M$ is as in  \eqref{eq:mdef}. We have also used  \eqref{eq:modconjdef} here. As discussed in Remark \ref{rem:Mi}, we will take the representation matrices $\gamma_{\lambda}$ to be real and $M=\mathds{1}$. 
\begin{comment}
Before we proceed, we find an explicit form of the matrix $M$ which will play a crucial role in the next result.  As explained in Appendix C, below eq. (C.10) of \cite{Dolan:1989vr}, we can take $M$ to be the charge conjugation matrix in the case of the Leech lattice. To find the charge conjugation matrix in 24 dimensions, we use the results of \cite{gammamat}. In 24 dimensions, the gamma matrices can be represented as \cite[eq. 4]{gammamat}
\begin{equation}
\begin{aligned}
\gamma^{1} &=\sigma_{1} \otimes I \otimes I \otimes \ldots \\
\gamma^{2} &=\sigma_{2} \otimes I \otimes I \otimes \ldots \\
\gamma^{3} &=\sigma_{3} \otimes \sigma_{1} \otimes I \otimes \ldots \\
\gamma^{4} &=\sigma_{3} \otimes \sigma_{2} \otimes I \otimes \ldots \\
\gamma^{5} &=\sigma_{3} \otimes \sigma_{3} \otimes \sigma_{1} \otimes \ldots \\
\ldots & \ldots \\
\gamma^{24} &=\sigma_{3} \otimes \sigma_{3} \otimes \ldots \otimes \sigma_{2}
\end{aligned}
\label{eq:gammamat}
\end{equation}
Put 
\begin{equation}
M=\sigma_{1} \otimes \sigma_{2} \otimes \sigma_{1} \otimes \ldots\otimes\sigma_2.
\end{equation}
Then by \cite[eq. 8,13,15,20,21, Table 1]{gammamat}, we have 
\begin{equation}
    M^{-1}\gamma_{\mu}M=\gamma_{\mu}^*,\quad M^{-1}=M^T=M.
\label{eq:}
\end{equation}
\end{comment}
\\\\
We next construct an explicit matrix representation of the Heisenberg group $\Gamma(\Lambda_{\mathrm{L}})$, defined analogous to \eqref{eq:gammamatalg}. Let $\{\gamma_i\}_{i=1}^{24}$ be  $(2^{12}\times 2^{12})$ real gamma matrices in 24 dimensions satisfying $\gamma_i^2=\mathds{1}$ and $\gamma_i\gamma_j=-\gamma_j\gamma_i$ for $i\neq j$.
Let $\Gamma$ denote the group generated by the $24$ Dirac gamma matrices. It has order $2 \cdot 2^{24}$ and is a Heisenberg extension of $\mathbb{Z}_2^{24}$. Introduce the notation:
\begin{equation}
    \gamma_w:=\gamma_1^{w^1}\cdots\gamma_{24}^{w^{24}}\quad\text{for } w=(w^1,\dots,w^{24})\in\mathbb{F}_2^{24}.
\label{eq:}
\end{equation}
Then clearly 
\begin{equation}
    \Gamma :=\{\pm\gamma_w : w\in\mathbb{F}_{2}^{24}\}.
\label{eq:}
\end{equation}
As is well-known, there is a unique irreducible representation of the Heisenberg group, 
up to isomorphism. Since $\Gamma$ is a group of matrices it has a canonical representation, we take $\mathcal{S}(\Lambda_{\mathrm{L}})$ to be this canonical representation.  
See \cite[Appendix C]{Dolan:1989vr} for one choice of $\gamma_i$. The group 
$\Gamma(\Lambda_{\mathrm{L}})$ is also a Heisenberg extension of $\mathbb{Z}_2^{24}$ and hence is 
isomorphic to $\Gamma$. To construct an explicit isomorphism, pick a $\mathbb{Z}_2$-basis $\{[\alpha_i]\}_{i=1}^{24}$ for $\Lambda_{\mathrm{L}}/2\Lambda_{\mathrm{L}}$ using Leech lattice vectors $\alpha_i$ with the property that 
\begin{equation}\label{eq:basisalphidotprod}
    \alpha_i\cdot\alpha_j\equiv (1-\delta_{ij})\bmod 4.
\end{equation} 
An explicit demonstration of the existence of such a  basis is given in 
\eqref{eq:specrepheisdiracgammamat} et. seq. Because of the property \eqref{eq:basisalphidotprod},
\begin{equation}
    \rho(e_{\alpha_i})=\gamma_{i},\quad\rho(-e_{\alpha_i})=-\gamma_{i} 
    \label{eq:basisgammamatrep}
\end{equation}
defines a group isomorphism $\rho: \Gamma(\Lambda_{\mathrm{L}}) \longrightarrow \Gamma$. Here $\{e_{\alpha_i},-e_{\alpha_i}\}$ is the fiber above $[\alpha_i]$.  
%, $\rho$ generates a group homomorphism. 
%other words, there exists an integral basis $\{ \alpha_i\}$ for the Leech lattice and a basis $\{w_i\}$ for $\mathbb{F}_2^{24}$ so that $\gamma_{w_i} \gamma_{w_j} = (-1)^{\alpha_i \cdot \alpha_j} \gamma_{w_j} \gamma_{w_i}$ and $\gamma_{w_i}^2 = +1 $ for all $w_i$. 
%
For any $\lambda\in\Lambda_{\mathrm{L}}$ with $[\lambda]=\sum_{i=1}^{24}w^i[\alpha_i]$ with $w^i\in\mathbb{Z}_2$, we get 
\begin{equation}
e_{\lambda}=v([\lambda])\prod_{i=1}^{24}e_{\alpha_i}^{w^i}
\end{equation}
where $v(\lambda)\in\{\pm 1\}$ is determined using the value of the cocycle $\varepsilon$ on the basis vectors. 
We can then easily check that 
\begin{equation}
\rho(e_{\lambda}) 
=v([\lambda])\prod_{i=1}^{24}\gamma_{i}^{w^i} =v([\lambda])\gamma_w
\label{eq:heisrepdirac}
\end{equation}
where $w=(w^1\dots w^{24})\in\mathbb{F}_2^{24}$.  As mentioned in Remark \ref{rem:Udaggerrep} we will often write 
\begin{equation}
    \gamma_\lambda:=\rho(e_\lambda)~.
\end{equation}
\begin{comment}
************************

I THINK THE FOLLOWING SENTENCES ADD NOTHING

Using this, we can construct an explicit isomorphism between $\mathbb{Z}$-modules $\Lambda_{\mathrm{L}}/2\Lambda_{\mathrm{L}}$ and $\mathbb{F}_2^{24}$ where the module structure on $\mathbb{F}_2^{24}$ is defined as 
\begin{equation}
    r\cdot w =\bar{r}w\quad \alpha\in\mathbb{Z},\quad w\in\mathbb{F}_2^{24},
\label{eq:}
\end{equation}
where $\bar{r}=r \bmod 2$. 

*************************************
\end{comment}
\begin{remark}
A choice of basis $\{[\alpha_i]\}_{i=1}^{24}$ for $\Lambda_{\mathrm{L}}/2\Lambda_{\mathrm{L}}$ with the property \eqref{eq:basisalphidotprod} is equivalent to a choice of representation of the Heisenberg group $\Gamma(\Lambda_{{\rm L}})$, the representation being specified by the group isomorphism \eqref{eq:basisgammamatrep}. The choice of representation gives the action of $e_\lambda$ on $\mathcal{S}(\Lambda_{{\rm L}})$. A difference choice of such a basis and hence a distinct isomorphic representation of $\Gamma(\Lambda_{{\rm L}})$ gives another isomorphic realization of the Beauty and the Beast SCFT. The explicit isomorphism is identical to the one in Remark \ref{rem:deponsectioniso}. Throughout the paper, we choose one such basis and fix it.      
\end{remark}
The choice of basis $\{[\alpha_i]\}_{i=1}^{24}$ of $\Lambda_{\mathrm{L}}/2\Lambda_{\mathrm{L}}$ with the property \eqref{eq:basisalphidotprod} gives us a map
\begin{equation}
    \begin{split}
&f:\Lambda_{\mathrm{L}}\longrightarrow \mathbb{F}_2^{24}\\&\pm\lambda\mapsto \left(w^{1}\dots w^{24}\right)
\end{split}
\label{eq:isof}
\end{equation}
where 
\begin{equation}
    [\lambda]=\sum_{i=1}^{24}w^i[\alpha_i],\quad w^i\in\mathbb{F}_2.
\label{eq:w^i_def}
\end{equation}
This map is well defined, surjective and an $\mathbb{F}_2$-module homomorphism. It is easy to see that $\text{Ker}(f)=2\Lambda_{\mathrm{L}}$. We denote the induced isomorphism of $\mathbb{F}_2$-modules $\Lambda_{\mathrm{L}}/2\Lambda_{\mathrm{L}}$ and $\mathbb{F}_2^{24}$ by the same symbol $f$. 
\begin{comment}
\begin{lemma}
We have that $\Lambda_{\mathrm{L}}/2\Lambda_{\mathrm{L}}\cong \mathbb{F}_2^{24}$ as $\mathbb{Z}$-modules, where $\mathbb{F}_2$ is a field with 2 elements denoted by $\{0,1\}$.
\begin{proof}
We will explicitly construct the isomorphism.
Using this construction, we can define a module homomorphism between $\Lambda_{\mathrm{L}}$ and $\mathbb{F}_2^{24}$ as modules over $\mathbb{Z}$ 
\end{proof}
\end{lemma}
\end{comment}
We now have the following lemma.
\begin{lemma}
Let $\Gamma_{12}\subset\Gamma$ be an Abelian subgroup of order $2^{12}$ consisting of elements of the form $\xi_w \gamma_w$, with $ w\in\mathbb{F}_2^{24}$ and $\xi_w\in \{ \pm 1 \}$  such that $-\mathds{1}\not\in\Gamma_{12}$.  
Then the image $\mathcal{C}_{12}$ of $\Gamma_{12}$ under the map $\Gamma\to \mathbb{F}_2^{24}$ given by $\pm\gamma_w\mapsto w$ is a $[24,12]$ binary linear code.
\label{lemma:subcodegamma12}
\begin{proof}
Since $\Gamma_{12}$ is a subgroup, $\mathds{1}\in \Gamma_{12}$ and hence $(00\dots 0)\in\mathcal{C}_{12}$. Moreover since $\gamma_{w}\gamma_{w'}=\pm\gamma_{w+w'}$ we see that $\mathcal{C}_{12}$ is a linear subspace of $\mathbb{F}_2^{24}$. 
Since $-\mathds{1}\not\in\Gamma_{12}$, we cannot have both $\gamma_w$ and $-\gamma_w$  in $\Gamma_{12}$. It is thus clear that $|\mathcal{C}_{12}|=2^{12}$. The statement now follows. 
\end{proof}
\end{lemma}
\begin{lemma}\label{lemma:cobbtri}
Let $\Gamma_{12}$ and $\mathcal{C}_{12}$ be as in Lemma \ref{lemma:subcodegamma12} with the additional property that $\gamma_w^2=\mathds{1}$ for all $\gamma_w\in\Gamma_{12}$. Then the cocycle $\varepsilon$ is trivialisable on $\mathcal{C}_{12}$ by a $U(1)$ coboundary $b:\mathbb{F}_{2}^{24}\longrightarrow U(1)$. That is 
\begin{equation}
    \varepsilon(w_1,w_2)=\frac{b(w_1+w_2)}{b(w_1)b(w_2)}.
\label{eq:}
\end{equation}
\begin{proof}
The isomorphism class of $U(1)$ central extensions is determined by the commutator function \cite[Proposition A.1]{Freed:2006ya} 
\begin{equation}
    \kappa(w_1,w_2)=\gamma_{w_1}\gamma_{w_2}\gamma_{w_1}^{-1}\gamma_{w_2}^{-1}.
\label{eq:}
\end{equation}
Since elements of $\Gamma_{12}$ square to identity, $\gamma_{w}^{-1}=\gamma_{w}$ for all $w\in\mathcal{C}_{12}$ and since $\Gamma_{12}$ is abelian, it is clear that $\kappa(w_1,w_2)=1$ for all $w_1,w_2\in\mathcal{C}_{12}$.
\end{proof}
\end{lemma}
\begin{lemma}
Let $b:\mathbb{F}_{2}^{24}\longrightarrow U(1)$ be the coboundary from Lemma \ref{lemma:cobbtri}. Define $\widetilde{\gamma}_{w}=b(w)\gamma_w$ for $w\in\mathbb{F}_2^{24}$. Then the operator 
\begin{equation}
    P_{12}:=\frac{1}{2^{12}}\sum_{w\in\mathcal{C}_{12}}\widetilde{\gamma}_w
\label{eq:}
\end{equation}
is a rank one projection operator acting on $\mathcal{S}(\Lambda_{\mathrm{L}})$. 
\begin{proof}
Observe that 
\begin{equation}
    \widetilde{\gamma}_{w_1}\widetilde{\gamma}_{w_2}=b(w_1)b(w_2)\varepsilon(w_1,w_2)\gamma_{w_1+w_2}=b(w_1+w_2)\gamma_{w_1+w_2}=\widetilde{\gamma}_{w_1+w_2}.
\label{eq:}
\end{equation}
Thus 
\begin{equation}
    P_{12}^{2}=\frac{1}{2^{24}}\sum_{w\in\mathcal{C}_{12}}\sum_{w'\in\mathcal{C}_{12}}\widetilde{\gamma}_{w+w'}=\frac{1}{2^{12}}\sum_{w\in\mathcal{C}_{12}}P_{12}=P_{12}.
\label{eq:}
\end{equation}
Moreover using the identity
\begin{equation}
    \begin{gathered}
\operatorname{Tr}\left(\gamma_{\mu_{1}} \ldots \gamma_{\mu_{n}}\right)=\delta_{\mu_{1} \mu_{2}} \operatorname{Tr}\left(\gamma_{\mu_{3}} \ldots \gamma_{\mu_{n}}\right)-\delta_{\mu_{1} \mu_{3}} \operatorname{Tr}\left(\gamma_{\mu_{2}} \gamma_{\mu_{4}} \ldots \gamma_{\mu_{n}}\right)+\ldots \\
+\delta_{\mu_{1} \mu_{n}} \operatorname{Tr}\left(\gamma_{\mu_{2}} \ldots \gamma_{\mu_{n-1}}\right)~,
\end{gathered}
\label{eq:}
\end{equation}
%where $g_{\mu\nu}=\delta_{\mu\nu}$ is the Euclidean metric, 
it is clear that 
\begin{equation}
    \text{Tr}(\gamma_w)=0,\quad \forall~~w\in\mathbb{F}_{2}^{24}\backslash\{0\}.
\label{eq:}
\end{equation}
Thus 
\begin{equation}
    \text{Tr}P_{12}=1.
\label{eq:}
\end{equation}
Since the only eigenvalues of $P_{12}$ are 0 and 1, it is clear that $P_{12}$ is rank one.
\end{proof}
\end{lemma}
\begin{defn}\label{def:supconflat}
A sublattice $\Lambda_{\text{SC}}\subset\Lambda_{\mathrm{L}}$ is called a \textit{superconformal sublattice} if the following properties are satisfied:  
\begin{enumerate}
    \item For every $\lambda_1,\lambda_2,\lambda\in\Lambda_{\text{SC}}$, $\langle\lambda_1,\lambda_2\rangle\equiv 0\bmod~2$ and $\langle\lambda,\lambda\rangle\equiv 0\bmod ~4$. 
    \item $\Lambda_{\text{SC}}$ does not contain any vector of norm-squared 4.
    \item $2\Lambda_{\mathrm{L}}\subset \Lambda_{\text{SC}}$ with index $2^{12}$. 
\end{enumerate} 
Hypothesis (1) and (3) implies that $\widetilde{\Gamma}(\Lambda_{\text{SC}}):=\{e_{\lambda}:[\lambda]\in\widetilde{\Lambda}_{\text{SC}}\}$ is an order $2^{12}$ Abelian subgroup of $\Gamma(\Lambda_{\mathrm{L}})$, where $\widetilde{\Lambda}_{\text{SC}}:=\Lambda_{\text{SC}}/2\Lambda_{\mathrm{L}}$. Hypothesis (2), as we will see below, is crucial for the construction of the supercurrent. 
\end{defn}
We can extend the map $f$, defined in \eqref{eq:isof} to a homomorphism $f:\Gamma(\Lambda_{\mathrm{L}})\longrightarrow\mathbb{F}_2^{24}$ as follows:
\begin{equation}
    \begin{split}
&f:\Gamma(\Lambda_{\mathrm{L}})\longrightarrow \mathbb{F}_2^{24}\\&\pm e_\lambda\mapsto \left(w^{1}\dots w^{24}\right)
\end{split}
\label{eq:isof_Gamma}
\end{equation}
where $w^i$ is defined as in \eqref{eq:w^i_def}. We then claim that for any superconformal sublattice $\Lambda_{\text{SC}}$ 
\begin{equation}
\mathcal{C}_{\mathrm{SC}}:=f(\widetilde{\Gamma}(\Lambda_{\mathrm{SC}}))
\end{equation}
is a $[24,12]$ binary code.  
It is clear that the group $\rho(\widetilde{\Gamma}(\Lambda_{\text{SC}}))$  satisfies all hypotheses of Lemma \ref{lemma:subcodegamma12}, where $\rho:\Gamma(\Lambda_{\text{L}})\longrightarrow\Gamma$ is the group isomorphism described in  \eqref{eq:basisgammamatrep}. Indeed using Property (1) of Definition \ref{def:supconflat} 
\begin{equation}
    e_{\lambda}e_{\mu}=(-1)^{\lambda\cdot\mu}e_{\mu}e_{\lambda}=e_{\mu}e_{\lambda}\quad\forall~ [\lambda],[\mu]\in \Lambda_{\text{SC}}/2\Lambda_{\mathrm{L}}. 
\end{equation}
Thus $\widetilde{\Gamma}(\Lambda_{\text{SC}})$ is abelian. By the same property, we have 
\begin{equation}
    e_{\lambda}^2=(-1)^{\frac{\lambda^2}{2}}\mathds{1}=\mathds{1},\quad\forall~[\lambda]\in \Lambda_{\text{SC}}/2\Lambda_{\mathrm{L}}.
\end{equation}
Finally $-\mathds{1}\not\in\rho(\widetilde{\Gamma}(\Lambda_{\text{SC}}))$ since otherwise if $ e_\lambda\in\widetilde{\Gamma}(\Lambda_{\text{SC}})$ then so does $ -e_\lambda\in\widetilde{\Gamma}(\Lambda_{\text{SC}})$ and $\left|\overline{\widetilde{\Gamma}(\Lambda_{\text{SC}})}\right|=|\Lambda_{\text{SC}}/2\Lambda_{\mathrm{L}}|=2^{11}$ contradicting the index of $2\Lambda_{\mathrm{L}}$ in $\Lambda_{\text{SC}}$. 
Hence $\mathcal{C}_{\mathrm{SC}}:=f(\widetilde{\Gamma}(\Lambda_{\mathrm{SC}}))$ is a $[24,12]$ binary code. 
\begin{thm}\label{thm:supcongenspin}
Let $\Lambda_{\mathrm{SC}}\subset\Lambda_{\mathrm{L}}$ be a superconformal sublattice and let $\mathcal{C}_{\mathrm{SC}}=f(\widetilde{\Gamma}(\Lambda_{\mathrm{SC}}))$ be the $[24,12]$ binary code with corresponding projection operator $P_{\mathrm{SC}}$. Then there exists $\chi_{\mathrm{SC}}\in\mathrm{Im}P_{\mathrm{SC}}$ %$\chi:=P_{\mathrm{SC}}\cdot\chi_0\in\mathrm{Im}P_{\mathrm{SC}}$, for some $\chi_0\in\mathcal{S}(\Lambda_{\mathrm{L}})$, 
such that $\alpha(\chi_{\mathrm{SC}})=1$ and $\kappa_{\lambda}(\chi_{\mathrm{SC}})=0$ for all $\lambda\in\Lambda_{\mathrm{L}}$ with norm squared 4, where $\alpha(\chi),\kappa_\lambda(\chi)$ is as in \eqref{eq:alp_kappa_chi}.
\end{thm}
\begin{proof}
Choose a normalised vector $\chi_1$ (i.e., scaled such that its norm is 1) in $\mathrm{Im}P_{\mathrm{SC}}$ and extend it to a basis $\{\chi_1,\chi_2,\dots,\chi_{2^{12}}\}$ of $\mathcal{S}(\Lambda_{\mathrm{L}})$. We have 
\begin{equation}
    P_{\mathrm{SC}}=|\chi_1\rangle\langle\chi_1|,
    \label{eq:projectorchi}
\end{equation}
where $|\chi\rangle\langle\chi|$ denotes the outer product of vectors.
By \eqref{eq:alphachi}  $\alpha(\chi_1)^*=-i$. 
%
%Now we have
%\begin{equation}
%    \begin{split}
%    \alpha(\chi_1)^*&=-i\langle\chi_1|\chi_1\rangle\\&=-%i\text{Tr}|\chi_1\rangle\langle \chi_1|\\&=-i,
%\end{split}
%\label{eq:}
%\end{equation}
%
%where we used the fact that Tr$P_{\mathrm{SC}}=1.$ 
%
So we scale $\chi_1$ by $e^{i\pi/4}$: 
\begin{equation}
    \chi_{\mathrm{SC}}:=e^{-i\pi/4}\chi_1.
\label{eq:}
\end{equation}
then we have 
\begin{equation}
\alpha(\chi_{\mathrm{SC}})^*=-i\langle e^{-i\pi/4}\chi_1|e^{i\pi/4}\chi_1\rangle=\langle\chi_1|\chi_1\rangle =1 =\alpha(\chi_{\mathrm{SC}}).  
\end{equation}
We now prove that $\kappa_{\lambda}(\chi_{\mathrm{SC}})=0.$ Since $\Lambda_{\text{SC}}$ does not contain vectors of norm squared 4, we see that if $\gamma_{\lambda}=\gamma_w$ with $w\in\mathbb{F}_2^{24}$ and $\lambda'\in[\lambda]$ with $\lambda'^2=4$, then $w\not\in \mathcal{C}_{\text{SC}}$. Thus if $\gamma_{\lambda}=\gamma_{w'}$ with $w'\not\in\mathcal{C}_{\text{SC}}$,
we have 
\begin{equation}
    \begin{split}
\kappa_\lambda(\chi_{\mathrm{SC}})=16\langle\chi_{\mathrm{SC}}|e_{\lambda}|\overline{\chi}_{\mathrm{SC}}\rangle &=-16i\langle e^{-i\pi/4}\chi_1|e_{\lambda}|e^{i\pi/4}\chi_1\rangle\\&=16\langle \chi_1|e_{\lambda}|\chi_1\rangle\\&=\frac{16}{2^{12}}\text{Tr}\sum_{w\in\mathcal{C}_{\text{SC}}}b(w)\gamma_{\lambda}\gamma_w\\&=\frac{16}{2^{12}}\sum_{w\in\mathcal{C}_{\text{SC}}}b(w)\text{Tr}\gamma_{w'}\gamma_{w},
\end{split}
\label{eq:}
\end{equation}
where we used the property of outer product:
\begin{equation}
    \langle\Psi|\Phi\rangle=\text{Tr}(|\Psi\rangle\langle\Phi|)~.
\end{equation}
Now since $w+w'\not\in\mathcal{C}_{\text{SC}}$ for any $w\in\mathcal{C}_{\text{SC}}$ and $w'\in\mathbb{F}_2^{24}\backslash\mathcal{C}_{\text{SC}}$, we have that $w+w'\neq 0$ and hence $\text{Tr}\gamma_{w'}\gamma_{w}=0$ since $w\in\mathcal{C}_{\text{SC}}$ and $\kappa_{\lambda}(\chi_{\mathrm{SC}})=0$. The proof is complete.
\end{proof}
Thus the existence of a superconformal vector reduces to the existence of a superconformal sublattice. We now prove that such a superconformal sublattice exists completing the proof of superconformal symmetry in the (spin lift of the) Moonshine module.
\begin{thm}
There exists a superconformal sublattice of $\Lambda_{\mathrm{L}}$.
\begin{proof}
It is proved in \cite[Lemma 4.3]{dong1996associative} that for any Niemeier lattice $L$, the lattice $\sqrt{2}L$, the lattice isomorphic to $L$ with all norms doubled, can be isometrically embedded\footnote{for an explicit embedding, see \cite{10.1006/eujc.1999.0360}.} in $\Lambda_{\mathrm{L}}$. We denote this embedded sublattice by  $\sqrt{2}L^\iota\subseteq\Lambda_{\mathrm{L}}$ where $\iota:\sqrt{2}L\longrightarrow\Lambda_{\mathrm{L}}$ is the emebedding map. Moreover, it is also true that $2\Lambda_{\mathrm{L}}\subset\sqrt{2}L^\iota$ and $\sqrt{2}L^\iota/2\Lambda_{\mathrm{L}}$ is a Lagrangian subspace  of $\Lambda_{\mathrm{L}}/2\Lambda_{\mathrm{L}}$. In particular, if we choose the Niemeier lattice $L$ to be the Leech lattice itself, then we have the following properties
\begin{equation}
    \begin{split}
         &(i)\quad \lambda^2\in 4\mathbb{N},\quad \lambda\in  \sqrt{2}\Lambda^\iota_{\mathrm{L}}\backslash\{0\}\\&(ii)\quad \lambda_1\cdot\lambda_2\in 2\mathbb{Z},\quad \lambda_i\in \sqrt{2}\Lambda^\iota_{\mathrm{L}}\\&(iii)\quad \lambda^2\neq 4\quad \forall~\lambda\in \sqrt{2}\Lambda^\iota_{\mathrm{L}}.
    \end{split}
    \label{eq:prop2E8}
\end{equation} 
The last property follows from the fact that Leech lattice has minimal norm 4. 
Since $\sqrt{2}\Lambda^\iota_{\mathrm{L}}/2\Lambda_{\mathrm{L}}$ is a Lagrangian subspace of $\Lambda_{\mathrm{L}}/2\Lambda_{\mathrm{L}}$, it is clear that the subgroup of 
$\Gamma(\Lambda_{\mathrm{L}})$: 
\footnote{Recall that we can choose a cocycle that is trivial on this subset}
\begin{equation}
    \widetilde{\Gamma}(\sqrt{2}\Lambda^\iota_{\mathrm{L}}):=\{e_{\lambda}:[\lambda]\in \sqrt{2}\Lambda^\iota_{\mathrm{L}}/2\Lambda_{\mathrm{L}}\}
\label{eq:heisliftsublatt}
\end{equation}
has order $2^{12}$. By \eqref{eq:heisliftsublatt} $\sqrt{2}\Lambda^\iota_{\mathrm{L}}$ is a superconformal sublattice. 
\end{proof}
\end{thm}
\begin{thm}
Let $\chi_{\mathrm{SC}}$ be the twisted ground state of Theorem \ref{thm:supcongenspin} for a choice of a superconformal sublattice $\Lambda_{\mathrm{SC}}$. Then $T_F(z):=2V(\chi_{\mathrm{SC}},z)$ is a superconformal current.
\begin{proof}
From \eqref{eq:opesimplified}, the OPE is 
\begin{equation}
    \widetilde{V}(\chi_{\mathrm{SC}},z)\widetilde{V}(\chi_{\mathrm{SC}},w)\sim\frac{1}{(z-w)^{3}}+\frac{1/8}{(z-w)}T_B(w).
\label{eq:}
\end{equation}
Any potential supercurrent for $\mathscr{H}_{SC}$ must have central charge  $\widehat{c}=16$. This is obvious since the  Monster VOA has central charge $c=24$. 
It is now clear that $T_F$ has the required OPE of a supercurrent with central charge $\widehat{c}=16$.
\end{proof} 
\end{thm}

%
%The spinor $\chi_{\text{SC}}$ constructed above now makes $\mathscr{H}_{SC}$ defined %in \eqref{eq:Hscdef} into an $\mathcal{N}=1$ superconformal vertex operator algebra %with central charge $\widehat{c}=16$. 
%Indeed $\mathscr{H}_{SC}$ is a super Hilbert space with $\mathscr{H}_{\bar{0}}:=\mathscr{H}_{\mathbb{Z}}$ and $\mathscr{H}_{\bar{1}}:=\mathscr{H}_{\mathbb{Z}+\frac{1}{2}}$ where the $\mathbb{Z}$ and $\mathbb{Z}+\frac{1}{2}$ subscript indicate integral and strictly half-integral conformal weight states. The conformal element is $\psi_L$ of \eqref{eq:confelem} and the superconformal element is given by the spinor described in Theorem \ref{thm:supcongenspin}.
\subsection{Classification Of $\mathcal{N}=1$ Structures Obtained From Codes}\label{sec:N>1susy}

Our construction of the supercurrent crucially depends on a choice of superconformal sublattice $\Lambda_{\text{SC}}$ of the Leech lattice.  
%
% As we have described, an embedding of $\sqrt{2}\Lambda_{\mathrm{L}}$ in 
% the Leech lattice is a superconformal sublattice. 
%
If there exists one such choice of superconformal sublattice $\Lambda_{\text{SC}}$, then for any $R\in\text{Co}_0$, where $\text{Co}_0$ is the Conway group, $R\cdot(\Lambda_{\text{SC}})$ is again a superconformal sublattice. 
\begin{defn}
Two supercurrents constructed from superconformal lattices are said to be \emph{equivalent} if the superconformal lattices are related by an automorphism of the Leech lattice or the codes associated to the two superconformal lattices are isomorphic.     
\end{defn}
This terminology is motivated by the fact that the $\mathcal{N}=1$ structures obtained from $\Lambda_{\text{SC}}$ and $R\cdot(\Lambda_{\text{SC}})$ (i.e., the SCVOA $\mathscr{H}_{SC}$ with superconformal states obtained from $\Lambda_{\text{SC}}$ and $R\cdot(\Lambda_{\text{SC}})$) are isomorphic (see Definition \ref{def:scvoa_hom}). Similarly $\mathcal{N}=1$ structures obtained from two isomorphic codes is isomorphic. We now explicitly demonstrate the isomorphism. 
\begin{prop}\label{prop:equiv_sup}
Let $\chi_{\mathrm{SC}},\chi'_{\mathrm{SC}}$ be two supercurrents obtained from either two isomorphic codes or from two Conway equivalent (i.e., related by a Conway rotation) superconformal sublattices. Then there exists vertex operators $V'(\psi,z)$ such that the \emph{SCVOA}s $(\mathscr{H}_{SC},V,|0\rangle,\psi_L,\chi_{\mathrm{SC}})$ and $(\mathscr{H}_{SC},V',|0\rangle,\psi_L,\chi'_{\mathrm{SC}})$ are isomorphic.    
\end{prop}
\begin{proof}
Let us first suppose that 
\begin{equation}
|\chi_{\text{SC}}\rangle=P|\chi_0\rangle,\quad   |\chi'_{\text{SC}}\rangle=P'|\chi_0\rangle~,\quad |\chi_0\rangle\in\mathcal{S}(\Lambda_{\mathrm{L}}),~\langle\chi_0|\chi_0\rangle=1~,  
\end{equation}
where 
\begin{equation}
    P:=\frac{1}{2^{12}}\sum_{w\in\mathcal{C}}\gamma_w~,\quad P':=\frac{1}{2^{12}}\sum_{w\in\mathcal{C}'}\gamma_w~,
\end{equation}
and $\mathcal{C},\mathcal{C}'$ are isomorphic $[24,12]$ codes such that $P,P'$ are rank 1 projectors. Let $S:\mathbb{F}_2^{24}\longrightarrow\mathbb{F}_2^{24}$ be a permutation such that $\mathcal{C}'=\mathcal{C}S$. Note that 
\begin{equation}
    \Gamma':=\{\pm\gamma_{wS}:w\in\mathbb{F}_2^{24}\}~,
\end{equation}
is defines another unitary irreducible representation of the Heisenberg extension of $\mathbb{Z}_2^{24}$ acting on $\mathcal{S}(\Lambda_{\mathrm{L}})$. By uniqueness of the representation, there exists a unitary matrix $\widehat{S}$ such that\footnote{Note that there is no $\pm 1$ factor because $S$ preserves the cocycle $\varepsilon(w,w')=\varepsilon(Sw,Sw')$.} 
\begin{equation}\label{eq:int_code_iso}
    \widehat{S}\gamma_w\widehat{S}^\dagger=\gamma_{wS}~,\quad w\in\mathbb{F}_2^{24}~.
\end{equation}
Now define the map $u_S:\mathscr{H}_{BB}\longrightarrow \mathscr{H}_{BB}$ by 
\begin{equation}
\begin{split}
u_{S} a_n^j u_{S}^{-1}&=a_n^j, \quad\quad  u_{S} c_k^j u_{S}^{-1}=c_k^j \\ u_{S}|\lambda\rangle &=|\lambda\rangle, \quad u_{S} |\chi\rangle=\widehat{S} |\chi\rangle~.
\end{split}
\end{equation}
Define $V'$ as in \eqref{eq:veropuntbb}, \eqref{eq:veroptbb} with $V'_T$ defined as in \eqref{eq:twistverop} with $e_\lambda$ replaced by $e_{(f^{-1}\circ S\circ f)([\lambda])}$. Then 
one can check that $u_S$ satisfies $u_SV(\psi,z)u_S^{-1}=V'(u_S\psi,z)$ for all $\psi\in \mathscr{H}_{BB}$ and preserves the vacuum and conformal vector. We now show that 
\begin{equation}
    u_S|\chi_{\text{SC}}\rangle=|\chi'_{\text{SC}}\rangle~.
\end{equation}
It is clear from \eqref{eq:int_code_iso} that $\widehat{S}P=P'\widehat{S}$. Thus we have 
\begin{equation}
 u_S|\chi_{\text{SC}}\rangle=\widehat{S}P|\chi_0\rangle=P'\widehat{S}|\chi_0\rangle=|\chi'_{\text{SC}}\rangle~.   
\end{equation}
Next suppose that $\chi_{\text{SC}}$ and $\chi'_{\text{SC}}=:\chi^R_{\text{SC}}$ be the superconformal vectors constructed from superconformal sublattices $\Lambda_{\text{SC}}$ and $R\cdot(\Lambda_{\text{SC}})$ respectively for some $R\in {\rm Co}_0$. Notice that $\{\pm\gamma_{R\lambda}:\lambda\in\Lambda_{{\rm L}}\}$ is another unitary representation of the Heisenberg group $\Gamma(\Lambda_{{\rm L}})$. Thus by uniqueness of unitary representation of the Heisenberg group, there exists a unitary matrix $S$ such that 
\begin{equation}
    S\gamma_\lambda S^\dagger=v_R(\lambda)\gamma_{R\lambda}~,
\end{equation}
for some $v_R(\lambda)=\pm 1$ which satisfies 
\begin{equation}\label{eq:vRlambepsilon}
    \frac{v_R(\lambda)v_R(\lambda')}{v_R(\lambda+\lambda')}\varepsilon(\lambda,\lambda')=\varepsilon(R\lambda,R\lambda')~.
\end{equation}
Define the map $u_R:\mathscr{H}_{BB}\longrightarrow \mathscr{H}_{BB}$ by 
\begin{equation}
\begin{split}
u_{R} a_n^\ell u_{R}^{-1}&=R^\ell_{~j} a_n^j, \quad\quad  u_{R} c_k^\ell u_{R}^{-1}=R^\ell_{~j} c_k^j \\ u_{R}|\lambda\rangle &=v_{R}(\lambda)|R \lambda\rangle, \quad u_{R} |\chi\rangle=S |\chi\rangle~.
\end{split}
\end{equation}
One can check that $u_R$ satisfies $u_RV(\psi,z)u_R^{-1}=V(u_R\psi,z)$ for all $\psi\in \mathscr{H}_{BB}$ and preserves the vacuum and conformal vector. We now show that 
\begin{equation}
    u_R|\chi_{\text{SC}}\rangle=|\chi^R_{\text{SC}}\rangle~.
\end{equation}
Observe that if $b(\lambda)=\pm 1$ trivializes the cocycle on $\widetilde{\Gamma}(\Lambda_{{\rm SC}})$ then $b(\lambda)v_R(\lambda)$ trivializes the cocycle on $\widetilde{\Gamma}(R\cdot (\Lambda_{{\rm SC}}))$. Indeed, using \eqref{eq:vRlambepsilon} we have 
\begin{equation}
\frac{b(\lambda)b(\lambda')}{b(\lambda+\lambda')}\frac{v_R(\lambda)v_R(\lambda')}{v_R(\lambda+\lambda')}=\varepsilon(\lambda,\lambda') \frac{v_R(\lambda)v_R(\lambda')}{v_R(\lambda+\lambda')}= \varepsilon(R\lambda,R\lambda')~.  
\end{equation}
Let $P_{{\rm SC}}$ and $P_{{\rm SC}}^R$ be the projection operators corresponding to the superconformal sublattices $\Lambda_{{\rm SC}}$ and $R\cdot(\Lambda_{{\rm SC}})$ respectively. Then we have 
\begin{equation}
\begin{split}
    SP_{{\rm SC}}S^\dagger=\frac{1}{2^{12}}\sum_{\lambda\in\Lambda_{{\rm SC}}/2\Lambda_{{\rm L}}}b(\lambda)S\gamma_\lambda S^\dagger &=\frac{1}{2^{12}}\sum_{\lambda\in\Lambda_{{\rm SC}}/2\Lambda_{{\rm L}}}b(\lambda)v_R(\lambda)\gamma_{R\lambda}=P_{{\rm SC}}^R~.
\end{split}
\end{equation}
Since $|\chi_{\text{SC}}\rangle\in {\rm Im}P_{{\rm SC}}$ we have 
\begin{equation}
    S|\chi_{\text{SC}}\rangle=SP_{{\rm SC}}|\chi_{\text{SC}}\rangle=P^R_{{\rm SC}}S|\chi_{\text{SC}}\rangle=|\chi^R_{\text{SC}}\rangle~,
\end{equation}
where we used the fact that $P^R_{{\rm SC}}$ has rank 1. 
\end{proof}

In the previous section we used the existence of an embedding of $\sqrt{2}\Lambda_{{\rm L}}$ into $\Lambda_{{\rm L}}$.  
In fact, there can be different embeddings not related by the action of the Conway group. 
In our language, these give inequivalent supercurrents as long as the resulting codes are also non-isomorphic. 
We will now show that there can be at most 9 inequivalent supercurrents constructed using embeddings $\iota: \sqrt{2}\Lambda_{{\rm L}} \hookrightarrow \Lambda_{{\rm L}}$. Moreover we construct two such inequivalent supercurrents explicitly. Since there are several supercurrents in the Beauty and the Beast SCFT, a natural question to ask is whether supersymmetry can be enhanced to $\mathcal{N}>1$. At the end of this section we show that there cannot be $\mathcal{N}>1$ supersymmetry in the Beauty and the Beast SCFT. 
%
%Although one may and have more than one supercurrent, it is not possible to have %$\mathcal{N}>1$ supersymmetry in the SCVOA $\mathscr{H}_{SC}$ since there are no %dimension 1 primary vertex operators in $\mathscr{H}_{SC}$ and hence there is no R-%symmetry current. Nevertheless one may ask whether the ``beauty and the beast'' module %has $\mathcal{N}>1$ supersymmetry. At the end of the section, we argue that one cannot %extend the supersymmetry beyond $\mathcal{N}=1$ even in the ``beauty and the beast'' %module. 
\\\\ 
Suppose that we can find another sublattice $\Lambda'$ of the Leech lattice with the properties listed in \eqref{eq:prop2E8}, then we can construct another supercurrent using the code $f(\Lambda'/2\Lambda_{\mathrm{L}})$ and the associated projection operator. Our construction is based on the isometric embedding of $\sqrt{2}\Lambda_{\mathrm{L}}$ into $\Lambda_{\mathrm{L}}$ which provides such a sublattice. Now there are many inequivalent embeddings %\footnote{by inequivalent, we mean not related by any automorphism of the Leech lattice. We will consider equivalence classes of embeddings under this equivalence.} 
of $\sqrt{2}\Lambda_{\mathrm{L}}$ into $\Lambda_{\mathrm{L}}$. We now want to classify all the supercurrents that can be constructed using these embeddings. As before, let us denote the image $\iota(\sqrt{2}\Lambda_{\mathrm{L}})$ of an embedding $\iota:\sqrt{2}\Lambda_{\mathrm{L}}\longrightarrow\Lambda_{\mathrm{L}}$ by $\sqrt{2}\Lambda_{\mathrm{L}}^{\iota}$.  We need the following proposition. 
\begin{prop}
Let $\iota:\sqrt{2}\Lambda_{\mathrm{L}}\longrightarrow\Lambda_{\mathrm{L}}$ be an embedding and $\{[\alpha_i]\}_{i=1}^{24}$ be the basis of $\Lambda_{\mathrm{L}}/2\Lambda_{\mathrm{L}}$ chosen in \eqref{eq:basisalphidotprod}. Suppose further that 
\begin{equation}\label{eq:sumalphiselfdualcodecond}
    \sum_{i=1}^{24}[\alpha_i]\in \sqrt{2}\Lambda_{\mathrm{L}}^{\iota}/2\Lambda_{\mathrm{L}}~.
\end{equation}
% \begin{enumerate}
%     \item $\alpha_i^2\equiv 0 \bmod~4$ for any representative $\alpha_i\in [\alpha_i]$;
%     %\item $\alpha_i\cdot\alpha_j\equiv 1\bmod~2,\quad i\neq j$;
%     \item $\sum_{i=1}^{24}[\alpha_i]\in \sqrt{2}\Lambda_{\mathrm{L}}^{\iota}/2\Lambda_{\mathrm{L}}$. 
% \end{enumerate}
Then the code 
\begin{equation}
    \mathcal{C}_{12}^{\iota}:=\{w\in\mathbb{F}_2^{24}: \gamma_w=\gamma_{\lambda}=\rho(e_\lambda)~\text{for some}~e_{\lambda}\in\widetilde{\Gamma}(\sqrt{2}\Lambda_{\mathrm{L}}^{\iota})\}
    \label{eq:ciota12}
\end{equation}
where $\widetilde{\Gamma}(\sqrt{2}\Lambda_{\mathrm{L}}^{\iota})$ is defined analogous to \eqref{eq:heisliftsublatt}, is a self-dual, doubly even code. 
\label{prop:sddecode}
\begin{proof}
%
%Let $\lambda_i,~i=1,2,\dots,24$ be vectors in $\Lambda_{\mathrm{L}}$ such that under %the isomorphism %$f:\Lambda_{\mathrm{L}}/2\Lambda_{\mathrm{L}}\longrightarrow\mathbb{F}_2^{24}$ of %\eqref{eq:isof} maps to the standard basis of $\mathbb{F}_{2}^{24}$. That is 
%\begin{equation}
%    \gamma_{[\lambda_i]}=\gamma_i,\quad i=1,2,\dots,24
%\end{equation}
%where $\gamma_i$ are the standard Dirac gamma matrices. It is clear that $\%%{[\lambda_i]\}$ forms a basis for $\Lambda_{\mathrm{L}}/2\Lambda_{\mathrm{L}}$. Now we %write this basis in terms of the basis $\{[\alpha_i]\}$ with the assumed properties. %Suppose 
%\begin{equation}
%    [\lambda_i]=\sum_{j=1}^{24}M_{ij}[\alpha_j],\quad M_{ij}\in\mathbb{Z}_2.
%\end{equation}
%Now since $\{[\alpha_i]\}$ is also a basis of %$\Lambda_{\mathrm{L}}/2\Lambda_{\mathrm{L}}$, $M_{ij}$ must be an invertible matrix. %In particular, the sum of entries in each row of $M_{ij}$ must be $1\bmod 2$  %otherwise $(1,1,\dots,1)^T$ would be an eigenvector of $M_{ij}$ with eigenvalue zero %and hence $M_{ij}$ is not invertible. This implies that 

Let $\{\alpha_i\}_{i=1}^{24}$ be a set of Leech vectors  representing the basis cosets $\{[\alpha_i]\}_{i=1}^{24}$ in the hypothesis. Then we have 
\begin{equation}
    \alpha_i \cdot \alpha_j = (1-\delta_{ij})\bmod 4~.
\end{equation} 
Then condition \eqref{eq:sumalphiselfdualcodecond} on an embedding implies that 
%\begin{equation}
%   \sum_{i=1}^{24}[\lambda_i]=\sum_{i=1}^{24}%[\alpha_i]\in\sqrt{2}\Lambda_{\mathrm{L}}^{\iota}.
%\end{equation}
%But then by construction 
%
\begin{equation}
    f\left(\sum_{i=1}^{24}[\alpha_i]\right) =(11\dots 1)\in\mathbb{F}_2^{24}.
\end{equation}
This implies that $(11\dots 1)\in\mathcal{C}_{12}^{\iota}$. Next note that for any  $w,w'\in\mathbb{F}_2^{24}$ \cite{Moore:2022}, 
\begin{equation}
    \gamma_w\gamma_{w'}=(-1)^{\text{wt}(w)\text{wt}(w')+w\cdot w'}\gamma_{w'}\gamma_{w}.
    \label{eq:commdiracgamma}
\end{equation}
Now since $\widetilde{\Gamma}(\sqrt{2}\Lambda_{\mathrm{L}}^{\iota})$ is abelian, using \eqref{eq:commdiracgamma} with $w'=(11\dots 1)$, we conclude that wt$(w)$ is even for every $w\in\mathcal{C}_{12}^{\iota}$. 
Since $\gamma_w^2=\mathds{1}$ for every $w\in\mathcal{C}_{12}^{\iota}$ and 
\begin{equation}
    \gamma_w^2=(-1)^{\frac{\text{wt}(w)(\text{wt}(w)-1)}{4}}\mathds{1},
\end{equation}
we must have wt$(w)=4k$ or $4k+1$ for $k$ a non-negative integer. But since wt$(w)$ is even for every $w\in\mathcal{C}_{12}^{\iota}$, we see that the code $\mathcal{C}_{12}^{\iota}$ is doubly even. The fact that it is self-dual follows by the fact that $(11\dots 1)\in\mathcal{C}_{12}^{\iota}.$ 
\end{proof}
\end{prop}

We now show that for any  
embedding $\iota:\sqrt{2}\Lambda_{\mathrm{L}}\longrightarrow\Lambda_{\mathrm{L}}$, 
there is an equivalent embedding $R\circ\iota$, for some suitable $R\in\text{Co}_0$,
such that the hypotheses in Proposition \ref{prop:sddecode} are satisfied for 
$R\circ\iota$.
 
We start by constructing 24 Leech lattice vectors $\{\alpha_i\}_{i=1}^{24}$ such that $\alpha_i\cdot\alpha_j=3\delta_{ij}+1$. (Reference \cite{Basis:leech} contains some useful information for this construction.)  Recall that the Mathieu group $M_{24}$ acts on $\Lambda_{\mathrm{L}}$ by permuting the coordinates. In particular, $M_{23}$ acts by permuting the first 23 coordinates. Now start with the norm-squared 4 vector \cite{conway:1999} 
\begin{equation}
\begin{split}
\alpha_1=\frac{1}{\sqrt{8}}\left(\begin{array}{l}-1,  1,  1,  1,  1,  -1,  1,  -1,  1,  1,  -1,  -1,  1,  1,  -1,  -1,  1,  -1,  1,  -1,  -1,  -1,  -1,  -3\end{array}\right).    
\end{split} 
\label{eq:basis1}
\end{equation}
Obtain 22 more vectors $\alpha_2,\dots,\alpha_{23}$ by acting on this vector by the order 23 permutation $(12\dots 23)\in M_{23}\hookrightarrow S_{23}\hookrightarrow S_{24}$, where 
$S_{23}\hookrightarrow S_{24}$ is the embedding by the action on the first $23$ coordinates.  
 Notice that in these 23 vectors, the $+1$s and $-3$ is supported on dodecads which intersect at 6 places pairwise. This will be very important.  Choose the last vector to be another norm squared 4 vector \cite{conway:1999} 
\begin{equation}
 \begin{split}
&\alpha_{24}=\frac{1}{\sqrt{8}}  \left(\begin{array}{l}1 , 1 , 1 , 1 , 1 , 1 , 1 , 1 , 1 , 1 , 1 , 1 , 1 , 1 , 1 , 1 , 1 , 1 , 1 , 1 , 1 , 1 , 1 , -3\end{array}\right).
\end{split} 
\label{eq:basis24}
\end{equation}
It is easily checked that these 24 vectors $\alpha_i$ have mutual inner product 1 and hence are inequivalent modulo $2\Lambda_{\mathrm{L}}$. This is because if $[\alpha_i]=[\alpha_j]$ then $\alpha_i=\alpha_j+2\lambda$ for some $\lambda\in\Lambda_{\mathrm{L}}$ and $\alpha_i\cdot\alpha_j\equiv 0\bmod~2$ contradicting the fact that $\alpha_i\cdot\alpha_j=1$. Thus we conclude that $\{[\alpha_i]\}$ with $\alpha_i$ constructed above is a basis for $\Lambda_{\mathrm{L}}/2\Lambda_{\mathrm{L}}$. 
\footnote{We remark that the $\alpha_i$ do \underline{not} form a basis for the Leech lattice. Rather, they span a sublattice of index $3^{13}$.}
Alternatively, the Gram matrix for the 24 classes $[\alpha_i]$ has 1 on off diagonal and 0 on diagonal and hence is invertible implying that $[\alpha_i]$ is a basis for $\Lambda_{\mathrm{L}}/2\Lambda_{\mathrm{L}}$. Next, note that the norm squared of each of the 24 vectors $\alpha_i$ is 4 and hence any representative of $[\alpha_i]$ will also have norm-squared divisible by 4. Finally it is easily checked that
\begin{equation}
    \sum_{i=1}^{24}[\alpha_i]=\left[\frac{1}{\sqrt{8}}(0,0,\dots,-72)\right]=\left[\frac{1}{\sqrt{8}}(0,0,\dots,-8)\right].
\end{equation}
Note that $\frac{1}{\sqrt{8}}(0,0,\dots,-8)$ is a norm-squared 8 vector.

Now consider an arbitrary  
embedding $\iota:\sqrt{2}\Lambda_{\mathrm{L}}\longrightarrow\Lambda_{\mathrm{L}}$. 
Of course $\sqrt{2}\Lambda_{\mathrm{L}}^{\iota}$ contains norm squared 8 vectors. 
Choose one. 
Since the Conway group Co$_0$ is transitive on norm-squared 8 vectors \cite[Theorem 29, Chapter 10]{conway:1999}   we can postcompose the embedding map $\iota$ by an automorphism of $\Lambda_{\mathrm{L}}$ to obtain an equivalent embedding containing the vector $\frac{1}{\sqrt{8}}(0~~0\dots-8)$. This means that $[\frac{1}{\sqrt{8}}(0~~0\dots-8)]\in\widetilde{\Gamma}(\sqrt{2}\Lambda_{\mathrm{L}}^{R\circ \iota})$ and we have satisfied the key condition \eqref{eq:sumalphiselfdualcodecond} of Proposition \ref{prop:sddecode}. 

This proves the following proposition.

\begin{prop}\label{prop:selddeequivemb}
Every equivalence class of embeddings\footnote{Two embeddings $\iota_1,\iota_2$ are equivalent if there exists $R\in {\rm Co}_0$ such that $\iota_2=R\circ\iota_1$.} contains a representative $\iota:\sqrt{2}\Lambda_{\mathrm{L}}\longrightarrow\Lambda_{\mathrm{L}}$ such that $\mathcal{C}^\iota_{12}$ is a self-dual doubly even $[24,12]$ code.
\end{prop}
\begin{thm}
Embeddings of $\sqrt{2}\Lambda_{\mathrm{L}}$ into $\Lambda_{\mathrm{L}}$ give rise to at most 9 inequivalent supercurrents in the ``Beauty and the Beast'' SCFT $\mathscr{H}_{BB}$.
\begin{proof}
By Proposition \ref{prop:selddeequivemb}, the code generating the supercurrents obtained from equivalence classes of embeddings is self-dual and doubly even.  There are only 9 such codes up to equivalence \cite{PLESS1975313}. Moreover, by Porposition \ref{prop:equiv_sup}, equivalent codes give rise to equivalent supercurrents. Hence the supercurrents constructed from such embeddings can only be one of the 9 constructed from the 9 selfdual, doubly even binary linear codes. 
\end{proof}
\end{thm}
\begin{comment}
Infact we can prove that
\begin{thm}
For any embedding $\iota:\sqrt{2}\Lambda_{\mathrm{L}}\longrightarrow\Lambda_{\mathrm{L}}$, the code $\mathcal{C}^{\iota}_{12}$ does not contain any codeword of weight 4. In particular $\mathcal{C}^{\iota}_{12}$ is the Golay code $\mathcal{G}_{24}$.
\begin{proof}
By \cite[Proposition 4.4]{dongmasonnortan}, the 2A-involutions of the Monster group in the maximal elemenentary abelian group $\Gamma(\sqrt{2}\Lambda_{\mathrm{L}}^{\iota})<\Gamma(\Lambda_{\mathrm{L}})$ corresponds to the root lattice of the Leech lattice $\Lambda_{\mathrm{L}}$ which is empty. This means that $\Gamma(\sqrt{2}\Lambda_{\mathrm{L}}^{\iota})$ does not contain any 2A-involutions of the Monster. Now it is known that any 2A-involution $\gamma_{[\lambda]}\in\Gamma(\Lambda_{\mathrm{L}})$ projects to a Leech vector $\lambda$ with $\lambda^2=4$. Since any class $[\lambda]\in\Lambda_{\mathrm{L}}/2\Lambda_{\mathrm{L}}$ has a representative with $\lambda^2=4,6$ or 8, any  involution $\gamma_{[\lambda]}\in\Gamma(\Lambda_{\mathrm{L}})$ must project to a Leech vector $\lambda$ with $\lambda^2=4$ or $\lambda^2=8$ \footnote{recall that $\gamma_{\lambda}^2=(-1)^{\frac{\lambda^2}{2}}$, so that $\gamma_{\lambda}$ is not an involution for $\lambda^2=6$.}, $\Gamma(\sqrt{2}\Lambda_{\mathrm{L}}^{\iota})$ only contains involutions $\gamma_{[\lambda]}$ which project to   
\end{proof}
\end{thm}
\end{comment}
We now explore some explicit examples of the codes in above theorem. We pick the following representation of the Heisenberg group: for the basis $[\alpha_i]$ of $\Lambda_{\mathrm{L}}/2\Lambda_{\mathrm{L}}$ constructed in \eqref{eq:basis1} and \eqref{eq:basis24}, assign 
\begin{equation}
    \rho(e_{\alpha_i})=\gamma_{\alpha_i}=\gamma_i,\quad i=1,\dots,24
    \label{eq:specrepheisdiracgammamat}
\end{equation}
where $\gamma_i$ are the standard Dirac gamma matrices. This is well defined because $\alpha_i\cdot\alpha_j=3\delta_{ij}+1$ and hence is compatible with $\{\gamma_i,\gamma_j\}=2\delta_{ij}\mathds{1}$. For any other $\lambda$, we first write 
\begin{equation}
    [\lambda]=\sum_{i=1}^{24}x_i[\alpha_i],\quad x_i\in\{0,1\}
\end{equation}
and then we have 
\begin{equation}
    \gamma_{\lambda}=\pm\gamma_1^{x_1}\cdots\gamma_{24}^{x_{24}}
\end{equation}
where the choice of sign is fixed by the cocycle $\varepsilon$ of the Heisenberg extension. 
\begin{prop}
With the above representation of the Heisenberg group $\Gamma(\Lambda_{\mathrm{L}})$, if $\emph{wt}(w)=4$ with $\gamma_w=\gamma_{\lambda}$ for $\lambda\in\sqrt{2}\Lambda_{\mathrm{L}}^{\iota}$ for any embedding $\iota:\sqrt{2}\Lambda_{\mathrm{L}}\longrightarrow\Lambda_{\mathrm{L}}$ then there is a representative $\lambda$ of $[\lambda]$ with norm-squared 12 and of the shape $\frac{1}{\sqrt{8}}(\pm 2^{8},\pm 4^{4},0^{12}),\frac{1}{\sqrt{8}}(\pm 2^{12},\pm 4^{3},0^{9})$ or $\frac{1}{\sqrt{8}}(\pm 2^{16},\pm 4^{2},0^{6})$.
\label{prop:wt4codewords}
\begin{proof}
Observe that for $\text{wt}(w)=4$, if $\gamma_{\lambda}=\gamma_w$ then there is a $\lambda\in[\lambda]$ such that 
\begin{equation}
    \lambda=\alpha_{i_1}+\alpha_{i_2}+\alpha_{i_3}+\alpha_{i_4},
\end{equation}
for some $i_1,\dots,i_4\in\{1,\dots,24\}.$ 
\begin{comment}
Next observe that from the explicit expression for $\alpha_i$, it is clear that $\alpha_{i_1}+\alpha_{i_2}$ has the form $8^{-1/2}(2^5,-2^6,0^{12},-6)$ since any two distinct $\alpha_i$ is supported on dodecads which intersect in 6 places. This also means that $\lambda$ has the form $8^{-1/2}(\pm 2^{k_1},\pm 4^{k_2},0^{12},-12)$. 
\end{comment}
Notice that since the dodecads on which the $-1$s of $\alpha_i,i=1,\dots,23$ are supported intersect at 6 places pairwise, for any $1\leq i_1,i_2\leq 23$, $\alpha_{i_1}+\alpha_{i_2}$ has the shape $8^{-1/2}(2^{5},-2^{6},0^{12},-6)$ with the nonzero coordinates supported on a dodecad. Moreover, for any $\alpha_i+\alpha_{24}$ for $i=1,\dots 23$ has the shape $8^{-1/2}(2^{11},0^{12},-6)$. This moreover implies that $\lambda$ has the shape  $8^{-1/2}(\pm 2^{k_1},\pm 4^{k_2},0^{23-k_1-k_2},-12)$. Next, using the fact that $\alpha_i\cdot\alpha_j=3\delta_{ij}+1$, we see that $\lambda^2=28$. Thus the norm squared of $8^{-1/2}(\pm 2^{k_1},\pm 4^{k_2},0^{12},-12)$ being 28 implies that 
\begin{equation}
    k_1+4k_2=20.
\end{equation}
The integral solutions to this equation are $(k_1,k_2)=(20,0),(16,1),(12,2),(8,3),(4,4),(0,5)$.  
Moreover, the only possible choices of $k_1$ are $8,12,16$ in which case $k_2=3,2,1$. Other choices can be ruled out using the fact that $\pm 4$ appears only where the dodecads in $\alpha_{i_1}+\alpha_{i_2}$ and $\alpha_{i_3}+\alpha_{i_4}$ intersect and the fact that dodecads can only intersect in 0,4,6, and 8 places. Adding $8^{-1/2}(0^{23},16)$ to $\lambda$ does not change the class, so we have another representative of $[\lambda]$ of the shape $8^{-1/2}(\pm 2^{k_1},\pm 4^{k_2+1},0^{23-k_1-k_2})$. Thus $[\lambda]$ can have representative of the shape $8^{-1/2}(\pm 2^{8},\pm 4^{4},0^{12}), 8^{-1/2}(\pm 2^{12},\pm 4^{3},0^{9})$ or $8^{-1/2}(\pm 2^{16},\pm 4^{2},0^{6})$.  
\end{proof}
\end{prop}
Using the construction of \cite{10.1006/eujc.1999.0360}, we will now explicitly construct two embeddings of $\sqrt{2}\Lambda_{\mathrm{L}}$ into $\Lambda_{\mathrm{L}}$ which will give rise to inequivalent codes and hence inequivalent supercurrents. \\\\
To construct an embedding, we start with a $\mathbb{Z}_4$-code $\mathscr{C}$ which is generated by codewords $g\in\mathbb{Z}_4^{24}$ such that the following holds: for each generator $g$ of $\mathscr{C}$, there is a vector $\widetilde{g}\in\mathbb{Z}^{24}$ such that $\frac{\widetilde{g}}{\sqrt{2}}\in\Lambda_{\mathrm{L}}$ and $g\equiv \widetilde{g}(\bmod~4)$. In simple terms $\frac{g}{\sqrt{2}}\in\Lambda_{\mathrm{L}}$ if $\frac{g}{\sqrt{2}}$ is considered as an element of $\mathbb{R}^{24}$ under the map $\mathbb{Z}_4\to\mathbb{R}$ which sends $0\to 0,1\to 1,2\to 2$ and $3\to 3$. Pseudo Golay codes and the Harada-Kitazume codes are examples of such $\mathbb{Z}_4$-codes, described in Appendix \ref{app:Z4codes}.
A typical element of the associated lattice $A(\mathscr{C})$ (see Theorem \ref{thm:latticez4code}) is of the form 
\begin{equation}
    \frac{g}{2}+\frac{4x}{2},\quad x\in\mathbb{Z}^{24}.
\end{equation}
This implies that 
\begin{equation}\label{eq:sq2lamintoleech}
    \sqrt{2}\left(\frac{g}{2}+\frac{4x}{2}\right)=\frac{g}{\sqrt{2}}+\frac{8x}{2\sqrt{2}}\in\Lambda_{\mathrm{L}}
\end{equation}
since $\frac{g}{\sqrt{2}}\in\Lambda_{\mathrm{L}}$ and $\frac{8x}{2\sqrt{2}}\in\Lambda_{\mathrm{L}}$ for any $x\in\mathbb{Z}^{24}$. 
Hence $\sqrt{2}A(\mathscr{C})\subseteq\Lambda_{\mathrm{L}}$ which gives an explicit embedding if we choose $\mathscr{C}$ such that $A(\mathscr{C})\cong\Lambda_{\mathrm{L}}$. To explicitly construct the embedding, recall that $\sqrt{2}\Lambda_{\mathrm{L}}=(\Lambda_{\mathrm{L}},2\langle\cdot,\cdot\rangle)$ where $\langle\cdot,\cdot\rangle$ is the Euclidean inner product in $\mathbb{R}^{24}$. This means that $\sqrt{2}\Lambda_{\mathrm{L}}\cong \sqrt{2}A(\mathscr{C})$ isometrically. The inclusion map $i:\sqrt{2}A(\mathscr{C})\longrightarrow\Lambda_{\mathrm{L}}$ then provides the embedding $\iota_{\mathscr{C}}:\sqrt{2}\Lambda_{\mathrm{L}}\longrightarrow\Lambda_{\mathrm{L}}$. 
\subsubsection{Duncan's Spinor Defines A Supercurrent In The BB Module}\label{subsec:duncspinor}
Recall that pseudo Golay codes $\mathscr{C}$ are $\mathbb{Z}_4$ codes such that Res$(\mathscr{C})=\mathcal{G}$. From \cite[Theorem 9]{RAINS1999215}, we see that every pseudo Golay code is Type II with minimum Euclidean distance $d_E(\mathscr{C})=16$. Thus by Theorem \ref{thm:latticez4code} and the uniqueness of Leech lattice, we see that $A(\mathscr{C})\cong\Lambda_{\mathrm{L}}$. We now show that the generators $g$ of the code satisfy $\frac{g}{\sqrt{2}}\in\Lambda_{\mathrm{L}}$. Indeed since $g ~(\bmod~2)\in\mathcal{G}$, the positions of $1$ and $3$ in $g$ must be an octad in the Golay code and there are even number of 2s in $g$. So the shape of $\frac{g}{\sqrt{2}}$ is
\begin{equation}
\begin{split}
    \frac{g}{\sqrt{2}}=\frac{2g}{2\sqrt{2}}&=\frac{1}{\sqrt{8}}(2^{k},6^{8-k},4^{2\ell},0^{16-2\ell})\\&=\frac{1}{\sqrt{8}}(\underbrace{2^{k},-2^{8-k}}_{\text{octad}},0^{16})+\frac{1}{\sqrt{8}}(8^{8-k},4^{2\ell},0^{16+k-2\ell}).
\end{split}
\end{equation}
Moreover $k$ is even for all the 13 pseudo Golay codes \cite{RAINS1999215}. Clearly from the list of vectors in the Leech lattice given in \cite[Table 4.13, Chapter 4]{conway:1999}, we see that $\frac{g}{\sqrt{2}}\in\Lambda_{\mathrm{L}}$.
\begin{prop}
Let $\iota_i:\sqrt{2}A(\mathscr{C}_i)\hookrightarrow\Lambda_{\mathrm{L}},~i=1,\dots,13$ be the embedding of $\sqrt{2}\Lambda_{\mathrm{L}}$ into $\Lambda_{\mathrm{L}}$ obtained from the 13 pseudo Golay codes $\mathscr{C}_i$. Then  
\begin{equation}
    \mathcal{C}_{12}^{\iota_i}\cong\mathcal{G},\quad i=1,\dots,13
\end{equation}
where $\mathcal{C}_{12}^{\iota_i}$ is defined as in \eqref{eq:ciota12}.
\label{prop:pseudogol}
\begin{proof}
By Proposition \ref{prop:sddecode}, the code $\mathcal{C}_{12}^{\iota_i}$ is a self-dual, doubly even binary $[24,12]$ code. It suffices to prove that for every $w\in\mathcal{C}_{12}^{\iota_i}$, wt$(w)\neq 4$. We prove this computationally using Mathematica. In view of Proposition \ref{prop:wt4codewords}, it suffices to prove that the ${24\choose 4}=10626$ vectors of the shape $\frac{1}{\sqrt{8}}(\pm 2^{8},\pm 4^{4},0^{12}),\frac{1}{\sqrt{8}}(\pm 2^{12},\pm 4^{3},0^{9})$ and $\frac{1}{\sqrt{8}}(\pm 2^{16},\pm 4^{2},0^{6})$ does not belong to $\sqrt{2}A(\mathscr{C}_i)$. This is equivalent to proving that the 10626 codewords of the form \footnote{note that $-1,-2$ are considered mod 4 in the code.} $(\pm 1^{8},\pm 2^{4},0^{12}),(\pm 1^{12},\pm 2^{3},0^{9})$ and $(\pm 1^{16},\pm 2^{2},0^{6})$ does not belong to $\mathscr{C}_i$. This was explicitly checked using Mathematica for each of the 13 pseudo Golay codes and we indeed found that this is true.
\end{proof}
\end{prop}
\begin{remark}\label{rem:HKcodesupcurgen}
Recall that the Harada-Kitazume code $\mathscr{C}_{HK}$ is another $\mathbb{Z}_4$ code satisfying all properties required to construct an embedding $\sqrt{2}A(\mathscr{C}_{HK})\subseteq\Lambda_{\mathrm{L}}$. Using similar Mathematica calculations as in Proposition \ref{prop:pseudogol}, we found that $\mathcal{C}_{12}^{HK}$ defined by 
\begin{equation}
    \mathcal{C}_{12}^{HK}:=\{w\in\mathbb{F}_2^{24}: \gamma_w=\gamma_{\lambda}=\rho(e_\lambda)~\text{for some}~e_{\lambda}\in\widetilde{\Gamma}(\sqrt{2}A(\mathscr{C}_{HK}))\}
    \label{eq:ciota12}
\end{equation}
is isomorphic to the Golay code $\mathcal{G}$.
\end{remark}
The above results imply that one of the supercurrent is generated by the projection operator constructed out of the Golay code as is the case for \textit{Conway supermoonshine} elucidated in \cite{Harvey:2020jvu}. It was called the Duncan spinor in \cite{Harvey:2020jvu} as it was first constructed by John Duncan in \cite{https://doi.org/10.48550/arxiv.math/0502267,Duncan:2014eha}. Thus one of the supercurrents in the ``Beauty and the Beast'' SCFT is based on the Duncan spinor.
\subsubsection{Supercurrent From Quaternary Golay Code}\label{subsec:quatgolaysupcur}
Using the quaternary Golay code, we now construct a supercurrent which is inequivalent to the Duncan spinor, see \eqref{eq:quatgolayexseq} and Table \ref{table:q4} for the definition of quaternary Golay code.  
We first show that the quaternary Golay code $\widehat{\mathcal{Q}}_4$ gives rise to an embedding of $\sqrt{2}\Lambda_{\mathrm{L}}$ into $\Lambda_{\mathrm{L}}$. Depending on $n_{\pm 1}\in\{0,8,12,16,24\}$, it is clear that any element of the code $\widehat{\mathcal{Q}}_4$ is of the shape 
\begin{equation}
    g=(1^{k},3^{n_{\pm}(g)-k},0^{24-n_{\pm}(g)})+(2^{2\ell},0^{24-2\ell}),
\end{equation}
where $k$ is even since otherwise the Euclidean norm of $g$ would be odd and hence it cannot form an even lattice. This implies that 
\begin{equation}
\begin{split}
    \frac{g}{\sqrt{2}}=\frac{2g}{\sqrt{8}}&=\frac{1}{\sqrt{8}}(2^{k},6^{n_{\pm}(g)-k},0^{24-n_{\pm}(g)})+\frac{1}{\sqrt{8}}(4^{2\ell},0^{24-2\ell})\\&=\frac{1}{\sqrt{8}}(\underbrace{2^{k},-2^{n_{\pm}(g)-k}}_{\text{Golay codeword}},0^{24-n_{\pm}(g)})+\frac{1}{\sqrt{8}}(8^{n_{\pm}(g)-k},0^{24+k-n_{\pm}(g)})+\frac{1}{\sqrt{8}}(2^{2\ell},0^{24-2\ell}).
\end{split}
\end{equation}
It is now clear that $\frac{g}{\sqrt{2}}\in\Lambda_{\mathrm{L}}$. Finally, by \cite[Theorem 4.3]{Bonnecaze97quaternaryquadratic} we see that $A(\widehat{\mathcal{Q}}_4)=\Lambda_{\mathrm{L}}$. Thus using \eqref{eq:sq2lamintoleech} this code gives us an embedding of $\sqrt{2}\Lambda_{\mathrm{L}}$ into $\Lambda_{\mathrm{L}}$. 

\begin{thm}\label{thm:quatgolcodsupcur}
The code $\mathcal{C}_{12}^{\widehat{\mathcal{Q}}_4}$ defined as in \eqref{eq:ciota12} is not isomorphic to the Golay code.    
\end{thm}
\begin{proof}
It suffices to show that the code $\mathcal{C}_{12}^{\widehat{\mathcal{Q}}_4}$ contains a weight 4 codeword and hence cannot be the Golay code. 
Indeed from Lemma \ref{rem:wt4cwQ4}, we see for any dodecad $w\in\mathcal{G}$, the code $\widehat{\mathcal{Q}}_4$ contains the codeword $(\pm 1^{12},2^{3},0^{9})$ for any choice of position of the three 2s with $\pm 1$ supported on $w$. From the proof of Proposition \ref{prop:wt4codewords}, we see that there exists a vector of the shape $8^{-1/2}(\pm 2^{12},4^{3},0^{9})$ which map to a weight 4 codeword under the homomorphism $f$. But then this means that $\sqrt{2}A(\widehat{\mathcal{Q}}_4)$ contains Leech vectors which under the homomorphism $f$ of \eqref{eq:isof} map to weight 4 codewords. Thus the code $\mathcal{C}_{12}^{\widehat{\mathcal{Q}}_4}$ cannot be the Golay code. 
\end{proof}
Thus we see that the code $\mathcal{C}_{12}^{\widehat{\mathcal{Q}}_4}$ thus gives rise to a new supercurrent different from the Duncan spinor. 
\\\\
We have now shown that there are at least two inequivalent supercurrents in the Beauty and the Beast module. Let us now calculate the OPE of the supercurrents with dimension 1 primaries in the theory since they are required for $\mathcal{N}>1$ supersymmetry. There are 24 linearly independent dimension 1 fields: $\widetilde{V}(a_{-1}^{i}|0\rangle,z)$. In terms of the fields $X^i,R^i$ it is given by  
\begin{equation}
   \widetilde{V}(a_{-1}^{i}|0\rangle,z)=\begin{pmatrix}
   :\partial_zX^i(z):&0\\0&:\partial_zR^i(z):
   \end{pmatrix} 
\end{equation}
Let $\widetilde{V}(\chi_s,z)$ be a supercurrent. By \eqref{eq:opeh-ht-} and \eqref{eq:chitilden} we have the OPE
\begin{equation}
   \widetilde{V}(a_{-1}^{i}|0\rangle,z) \widetilde{V}(\chi_s,w)=\frac{1}{(z-w)^{\frac{1}{2}}}\widetilde{V}(c_{-1/2}^i\chi_s,w)+\mathcal{O}((z-w)^{\frac{1}{2}}),
\end{equation}
where we used the fact that $c_k^i\chi=0$ for $k>0$ and $\chi\in\mathcal{S}(\Lambda_{\mathrm{L}})$. Thus the dimension 1 primaries have a square root singularity with the supercurrents. Therefore there is no mutually local dimension one $U(1)$ R-symmetry current and therefore we cannot have extended supersymmetry in the BB module. It might be possible to define an $\mathcal{N}=2$ supersymmetry algebra by passing to suitable covering spaces. In any case, there is a novel symmetry structure here which might be worth understanding more deeply.\\\\We also note that similar calculations as in Section \ref{subsec:BB} give us the OPE of two vertex operators $\widetilde{V}(\chi_a,z)$ and $\widetilde{V}(\chi_b,z)$ corresponding to twisted ground states $\chi_a,\chi_b\in\mathcal{S}$ to be  
\begin{equation}
\widetilde{V}(\chi_a,z)\widetilde{V}(\chi_b,w)\sim =\frac{\alpha^{ab}}{(z-w)^{3}}+\frac{\alpha^{ab}/8}{(z-w)}T_B(w)+\sum_{\lambda^2=4}\frac{\kappa_{\lambda}^{ab}}{z-w}\begin{pmatrix}e^{i\lambda\cdot X(w)}\sigma_{\lambda}&0\\0&e^{i\lambda\cdot R(w)}e_{\lambda}\end{pmatrix}     
\end{equation}
where 
\begin{equation}
    \alpha^{ab}:=\langle \overline{\chi}_a|\chi_b\rangle,\quad \kappa_{\lambda}^{ab}=16i\langle \chi_a|e_\lambda|\overline{\chi}_b\rangle.
\end{equation}
Thus any two supercurrents are local with respect to each other. For the two supercurrents constructed above, it may be possible to calculate the coefficients $\alpha^{ab}$ and $\kappa_{\lambda}^{ab}$ explicitly. But we will leave this calculation for future. This suggests that there is a novel chiral algebra in the BB module.
\begin{appendix}

\section{Codes And Lattices}\label{sec:codeslat}
Linear codes and their associated lattices play a central role in our discussion of the superconformal symmetry in the Monster module. We collect the definitions and some preliminary results about codes and lattices in this section.
\subsection{Binary Linear Codes}
\begin{defn}
Let $\mathbb{F}_2$ be the field with two elements $0,1$.
\begin{enumerate}
    \item A \textit{binary linear code} $\mathcal{C}$ of length $n$ is a vector subspace of $\mathbb{F}_2^n$. If the dimension of $\mathcal{C}$ is $k$ then it is called an $[n,k]$ binary code. The elements of the code are called \textit{codewords}.
    \item Two binary linear codes $\mathcal{C}_1$ and $\mathcal{C}_2$ are said to be \textit{equivalent} if there exists a permutation matrix $P$ such that $\mathcal{C}_1P=\mathcal{C}_2$. 
    \item A $k\times n$ matrix with rows being the basis vectors of $\mathcal{C}$ is called a generator matrix of $\mathcal{C}$. 
    \item The inner product of codewords $w,w'\in\mathcal{C}$ is defined as 
    \begin{equation}
        w\cdot w'=\sum_{i=1}^nw_i\cdot w'_i ~(\bmod~2).
    \end{equation}
    where $w=(w_1,\dots,w_n),w'=(w'_1,\dots,w'_n)$.
The code \textit{dual} to $\mathcal{C}$ is denotes by $\mathcal{C}^{\perp}$ and is defined as 
\begin{equation}
    \mathcal{C}^{\perp}:=\{w\in\mathbb{F}_2^n:w\cdot w'=0~\forall~w'\in\mathcal{C}\}.
\end{equation}
The code is called \textit{self-orthogonal} if $\mathcal{C}^{\perp}\subseteq\mathcal{C}$ and \textit{self-dual} if $\mathcal{C}^{\perp}=\mathcal{C}$.
\item The \textit{(Hamming) weight} of a codeword $w$ is defined as the number of nonzero components in the codeword and is denoted by wt($w$). The code is called \textit{even} and \textit{doubly even} if the weight of every codeword is divisible by 2 and 4 respectively.   
\item The \textit{Hamming distance} $d_H(w,w')$ of two codewords $w,w'$ is defined as the number of places in which the codewords differ. More precisely, $d_H(w,w')=\text{wt}(w-w').$ The \textit{(Hamming) distance} of the code is defined 
\begin{equation}
    d_H(\mathcal{C}):=\min_{\substack{w,w'\in\mathcal{C}\\w\neq w'}}d_H(w,w').
\end{equation}
With the notation $d_H:=d_H(\mathcal{C})$, we call $\mathcal{C}$ an $[n,k,d_H]$ code.
\item The \textit{automorphism group} Aut$(\mathcal{C})$ of the code $\mathcal{C}$ is a subset of the permutation group $S_n$ which fixes the code as a set.   
\end{enumerate}
\end{defn}
By standard row echelon form algorithm, it is easy to see that for any $[n,k]$ code $\mathcal{C}$, there is an equivalent code with generator matrix $G$ of the form 
\begin{equation}
    G=[\mathds{1}_k,A]
\end{equation}
where $\mathds{1}_k$ is the $k\times k$ identity matrix and $A$ is some $k\times(n-k)$ matrix. The dual code $\mathcal{C}^{\perp}$ is an $[n,n-k]$ code with generator matrix $[-A^{\text{tr}},\mathds{1}_{n-k}]$. In particular, if $\mathcal{C}$ is self-dual then $n-k=k$ and $n=2k$. We will be interested in $[24,12]$ codes. The following theorem gives the classification of self-dual, doubly even $[24,12]$ codes.
\begin{thm}\emph{\cite{PLESS1975313}}
There are exactly 9 inequivalent self-dual, doubly even $[24,12]$ codes. Up to equivalence, there is a unique  $[24,12,8]$ code called the \emph{Golay code} denoted by $\mathcal{G}$. 
\end{thm}
\subsubsection{Golay Code}
The Golay code is of particular interest to us due to its application in the construction of Leech lattice. It is determined by a generator matrix of the form 
\begin{equation}
     \left(\begin{array}{llllll}
1000&0000&0000&1010&0011&1011\\
0100&0000&0000&1101&0001&1101\\
0010&0000&0000&0110&1000&1111\\
0001&0000&0000&1011&0100&0111\\
0000&1000&0000&1101&1010&0011\\
0000&0100&0000&1110&1101&0001\\
0000&0010&0000&0111&0110&1001\\
0000&0001&0000&0011&1011&0101\\
0000&0000&1000&0001&1101&1011\\
0000&0000&0100&1000&1110&1101\\
0000&0000&0010&0100&0111&0111\\
0000&0000&0001&1111&1111&1110 
    \end{array}\right)
\end{equation}
The following results can be proved using the standard constructions of the Golay code. All codewords have weights of 0, 8, 12, 16, or 24. Codewords of weight 8 are called \textit{octads} and codewords of weight 12 are called \textit{dodecads}.
\begin{prop}\emph{\cite{conway:1999}} The following is true about the Golay code:
\begin{enumerate}
    \item There are 759 octads and 2576 dodecads in the Golay code.
    \item Two octads intersect in 0, 2, or 4 coordinates and an octad and a dodecad intersect at 2, 4, or 6 coordinates. 
\end{enumerate}
\end{prop}
 This proposition also implies that there are 759 weight 16 codewords obtained by adding $\underline{1}=(11\dots 1)\in\mathcal{G}$ to the octads totalling to $2+(2\times 759)+2576=2^{12}$ codewords. 
\subsection{$\mathbb{Z}_4$-Codes}\label{app:Z4codes}
We now recall the definition and some properties of $\mathbb{Z}_4$-codes (see \cite{CARYHUFFMAN2005451} for details). 
\begin{defn}
Let $\mathbb{Z}_4$ be the ring with elements 0,1,2,3 with modulo 4 addition and multiplication as the operations. 
\begin{enumerate}
    \item A $\mathbb{Z}_4$-code $\mathscr{C}$ of length $n$ is a submodule of $\mathbb{Z}_4^n$. The elements of the code are called codewords. 
    \item A matrix whose rows are a minimal spanning set for the code is called a \textit{generator matrix}. 
    \item The \textit{Euclidean weights} of the elements 0, 1, 2 and $3=-1$ of $\mathbb{Z}_4$ are ordinary integers 0, 1, 4 and 1, respectively,
and the Euclidean weight of a codeword is the sum of the Euclidean weights of its
components. The minimum Euclidean weight $d_E$ of $\mathscr{C}$ is the smallest Euclidean weight among all non-zero codewords of $\mathscr{C}$. 
\item The (Hamming) weight wt$(x)$ of a codeword $x\in\mathscr{C}$ is the number of non-zero coordinates of $x$. The Hamming distance of two codewords is $d_H(x,y)=\text{wt}(x-y)$ and the Hamming distance $d_H(\mathscr{C})$ of the code is defined as before as
\begin{equation}
    d_H(\mathscr{C})=\min_{\substack{x,y\in\mathscr{C}\\x\neq y}}d_H(x,y).
\end{equation} 
\item An invertible monomial matrix $M$ over the ring $\mathbb{Z}_4$ is an invertible matrix of the form $M = P D$ where $P$ is a permutation matrix and $D$ is a diagonal matrix
with $\pm 1$ on the diagonal. Two codes $\mathscr{C}_1$ and, $\mathscr{C}_2$ are \textit{equivalent} if there exists an
invertible monomial matrix $M$ such that $\mathscr{C}_1M=\mathscr{C}_2$ where $\mathscr{C}_1M=\{xM:x\in\mathscr{C}_1\}.$
\item Self-orthogonality and self-duality of the code is defined as before with the same inner product now calculated modulo 4. Self-dual
codes with the property that all Euclidean weights are divisible by 8 are called Type II.
\item The \textit{residue code} Res$(\mathscr{C})$ and the \textit{torsion code} Tor$(\mathscr{C})$ are binary linear codes defined by the exact sequence 
\begin{equation}
    0\longrightarrow \mathrm{Tor}(\mathscr{C})\stackrel{f}{\longrightarrow}\mathscr{C}\stackrel{g}{\longrightarrow}\mathrm{Res}(\mathscr{C})\longrightarrow 0
\end{equation}
obtained from the componentwise exact sequence
\begin{equation}
    0\longrightarrow \mathbb{Z}_2\stackrel{f}{\longrightarrow}\mathbb{Z}_4\stackrel{g}{\longrightarrow}\mathbb{Z}_2\longrightarrow 0
\end{equation}
where $f:\mathbb{Z}_2\longrightarrow\mathbb{Z}_4$ maps $0\to 0,1\to 2$ and $g:\mathbb{Z}_4\longrightarrow\mathbb{Z}_2$ maps $0,1,2,3\in\mathbb{Z}_4$ to $0,1,0,1\in\mathbb{Z}_2$ respectively.
Explicitly, the residue and torsion codes are given by 
\begin{equation}
    \text{Res}(\mathscr{C})=\{x~ (\bmod~2):x\in\mathscr{C}\}.
\end{equation}
and
\begin{equation}
    \text{Tor}(\mathscr{C})=\{x~(\bmod ~2): x\in\mathbb{Z}_4^{n}:2x\in\mathscr{C}\}.
\end{equation}
\end{enumerate}
%The symmetrized weight enumerator of $\mathcal{C}$ is given by
%\begin{equation}
 %   \mathrm{swe}_{\mathcal{C}}(a,b,c)=\sum_{c\in\mathcal{C}}a^{n_0(c)}b^{n_1(c)+n_3(c)}c^{n_2(c)}
%\end{equation}
%where $n_i(c)$ are the number of components of $c\in\mathbb{C}$ that are equal to $i$. Self duality is defined with respect to the usual Euclidean inner product.
\end{defn} 
We have the following theorem: 
\begin{thm}\emph{\cite{568705}}
Let $\mathscr{C}$ be a Type II $\mathbb{Z}_4$-code of length
$n$ and minimum Euclidean weight $d_E(\mathscr{C})$. Then 
\begin{equation}
    d_E(\mathscr{C})\leq
    8\lfloor n/24\rfloor + 8. 
\end{equation}
\label{thm:extremalcd}
\end{thm}
\begin{defn}
If equality holds in Theorem \ref{thm:extremalcd}, then $\mathscr{C}$ is called an \textit{extremal code}. Such extremal codes will play a role in our explicit constructions of superconformal currents. 
\end{defn}
There is no known classification of extremal $\mathbb{Z}_4$ codes. In section \ref{sec:N>1susy},  we 
will relate the construction of superconformal currents to certain codes. Three of the codes that will be relevant are described below: \\\\
\textbf{Pseudo Golay codes:}
One set of $\mathbb{Z}_4$ codes which concern us here are those whose residue code is the Golay code. Following \cite{RAINS1999215}, we call such codes \textit{pseudo Golay codes}. There are exactly 13 pseudo Golay codes \cite{RAINS1999215}. The generator matrix for two of these pseudo Golay code is given below. The rest of the 11 generator matrices can be found in \cite{RAINS1999215}.
\[
\begin{pmatrix}
1000&0000&0000&1300&2111&0123 \\
0100&0000&0000&1210&2300&3111 \\
0010&0000&0000&3311&2330&1200 \\
0001&0000&0000&0331&1233&0120 \\
0000&1000&0000&0033&1123&3012 \\
0000&0100&0000&2023&3312&1123 \\
0000&0010&0000&1012&1231&1223 \\
0000&0001&0000&1311&0023&0233 \\
0000&0000&1000&1301&3302&2130 \\
0000&0000&0100&0130&1330&2213 \\
0000&0000&0010&1223&2033&3332\\
0000&0000&0001&2102&3003&1111
\end{pmatrix}\quad \begin{pmatrix}
1000 & 0000 & 0000 & 2202 & 2313 & 3311 \\
0100 & 0000 & 0000 & 2233 & 3301 & 1001 \\
0010 & 0000 & 0000 & 1010 & 3203 & 3312 \\
0001 & 0000 & 0000 & 0101 & 3220 & 3331  \\
0000 & 1000 & 0000 & 3223 & 3013 & 2121  \\
0000 & 0100 & 0000 & 1113 & 3210 & 1020  \\
0000 & 0010 & 0000 & 1302 & 3130 & 2310  \\
0000 & 0001 & 0000 & 2112 & 3231 & 0211  \\
0000 & 0000 & 1000 & 3121 & 2323 & 1320  \\
0000 & 0000 & 0100 & 1231 & 0113 & 2212  \\
0000 & 0000 & 0010 & 3332 & 0102 & 1013 \\
0000 & 0000 & 0001 & 2311 & 0132 & 2121
\end{pmatrix}
\]
\noindent\textbf{Quaternary Golay code:}
Another important example of a $\mathbb{Z}_4$ code is a code denoted $\widehat{\mathcal{Q}}_4$ in \cite{Bonnecaze97quaternaryquadratic}. It has the property that Res$(\widehat{\mathcal{Q}}_4)\cong \text{Tor}(\widehat{\mathcal{Q}}_4) \cong \mathcal{G}$. Thus it fits in an exact sequence 
\begin{equation}\label{eq:quatgolayexseq}
    0\longrightarrow \mathcal{G}\stackrel{f}{\longrightarrow}\widehat{\mathcal{Q}}_4\stackrel{g}{\longrightarrow}\mathcal{G}\longrightarrow 0.
\end{equation}
Let $n_0(x),n_{\pm 1}(x),n_2(x)$ denote the number of entries in a   codeword $x\in \widehat{\mathcal{Q}}_4$ with $0, \pm 1, 2 \in \mathbb{Z}/4\mathbb{Z}$, respectively.   
%
%Here $-1$ should be understood as $3$. 
%
Since  Res$(\widehat{\mathcal{Q}}_4)=\mathcal{G}$, $n_{+ 1}(x)+n_{- 1}(x)=0,8,12,16,24$ for every $x\in\widehat{\mathcal{Q}}_4$. 
The number of codewords of each type is listed in Table \ref{table:q4}.
\begin{remark}
The codewords in $\widehat{\mathcal{Q}}_4$ with $n_{\pm 1}=12$ have $\pm 1$s supported on a dodecad in the Golay code. The reason is that the mod 2 reduction of $\widehat{\mathcal{Q}}_4$ is the Golay code and $0,2$ reduce to 0 mod 2.    
\end{remark}
In Theorem \ref{thm:quatgolcodsupcur} of Subsection \ref{sec:N>1susy} the following result will be
useful in the construction of a nontrivial supercurrent, see Subsection \ref{subsec:quatgolaysupcur} for the details of the construction.
\begin{lemma}\label{rem:wt4cwQ4}
Let $w\in\mathcal{G}$ be a dodecad and $\widehat{\mathcal{Q}}_w\subset\widehat{\mathcal{Q}}_4$ be the preimage of the mod 2 reduction map from $\widehat{\mathcal{Q}}_4$ to $\mathcal{G}$. Then for any choice of three out of the 12 remaining positions, there is a codeword $x\in\widehat{\mathcal{Q}}_w$ with 2 at those three positions and 0 at the remaining 9 positions.  
\begin{proof}
Let $ x\in\widehat{\mathcal{Q}}_w$. Then $\pm 1$s in $x$ are supported at $w$. At the other 12 positions, we can either have 0 or 2. From Table \ref{table:q4}, we see that the number of codewords with 12 $\pm 1$s, 9 0s and 3 2s is $220\cdot 2576$. Moreover, since there are exactly 2576 dodecads in the Golay code, from Table \ref{table:q4} we see that, if we fix the position of the 12 $\pm 1$s at $w$ then there are a total of 220 codewords in $\widehat{\mathcal{Q}}_4$ with 9 0s and 3 2s at the other 12 positions. Since there are exactly ${12\choose 3}=220$ choices for the positions of 2, there is a codeword in $\widehat{\mathcal{Q}}_w$ with $\pm 1$s at $w$ and every possible position of the three 2s. 
\end{proof}
\end{lemma}
%
%This is because we can place the three $2$s at ${12\choose 3}=220$ positions which are all included in Table \ref{table:q4}. 
%
\begin{landscape}
\vspace{10cm}
\begin{table}[h]
\centering
    \begin{tabular}{|c|c|c|c|c|}
    \hline
    $n_{\pm 1}=0$     & $n_{\pm 1}=8$ & $n_{\pm 1}=12$ & $n_{\pm 1}=16$ & $n_{\pm 1}=24$  \\\hline
     $(0; 0; 24)=1$  & $(2; 8; 14)=2\cdot8\cdot 759 $  & $(1; 12; 11)=2\cdot12\cdot2576$ &$(0; 16; 8)=2\cdot16\cdot 759$& $(0; 24; 0)=2^{12}$ \\ $(24; 0; 0)=1$ & $(14; 8; 2)=2\cdot8\cdot 759$ & $(11; 12; 1)=2\cdot12\cdot2576$&$(8; 16; 0)=2\cdot16\cdot 759$ & \\$(8; 0; 16)=759$ & $(4; 8; 12)=2\cdot112\cdot 759$ & $(3; 12; 9)=2\cdot220\cdot 2576$&$(2; 16; 6)=2\cdot448\cdot 759$ & Total $2^{12}$ \\  $(16; 0; 8)=759$ & $(12; 8; 4)=2\cdot112\cdot 759$ & $(9; 12; 3)=2\cdot220\cdot 2576$ &$(2; 16; 6)=2\cdot448\cdot 759$ &\\ $(12; 0; 12)=2576$ & $(6; 8; 12)=2\cdot504\cdot 759$ & $(5; 12; 7)=2\cdot792\cdot 2576$ &$(4; 16; 4)=2240\cdot 759$  & \\  & $ (12; 8; 6)=2\cdot504\cdot 759$ & $(7; 12; 5)=2\cdot792\cdot 2576$&  & \\Total $2^{12}$ & $(8; 8; 8)=1600\cdot 759$ & & Total $759\cdot 2^{12}$&\\ & &Total $2576\cdot 2^{12}$ && \\&Total $759\cdot 2^{12}$ & &&\\\hline 
    \end{tabular}
    \caption{Number of codewords of each type in $\widehat{\mathcal{Q}}_4$. The entries $(a;b;c)=N$ in each column means that there are $N$ codewords of type $(n_0;n_{\pm 1};n_2)$. The factor of 2 in the first few rows of columns 2, 3 and 4 arises from the fact that in those cases for every codeword with $n_{+1}(x)=a,n_{-1}(x)=b$ there is another codeword with $n_{+1}(x)=b,n_{-1}(x)=a$. }
    \label{table:q4}
\end{table}
\end{landscape}
\noindent\textbf{Harada-Kitazume code:}
One other $\mathbb{Z}_4$ code that we will be interested in was constructed in \cite{10.1006/eujc.1999.0360}. The generator matrix for the code is shown below.
\[
\left(\begin{array}{llllll}
2222 & 0000 & 0000 & 0000 & 0000 & 0000 \\
0022 & 2200 & 0000 & 0000 & 0000 & 0000 \\
0000 & 0022 & 2020 & 0000 & 0000 & 0000 \\
0000 & 0000 & 0202 & 2020 & 0000 & 0000 \\
0000 & 0000 & 0000 & 0202 & 2002 & 0000 \\
2020 & 2020 & 0000 & 0000 & 0000 & 0000 \\
0000 & 0220 & 2200 & 0000 & 0000 & 0000 \\
0000 & 0000 & 2002 & 2002 & 0000 & 0000 \\
0000 & 0000 & 0000 & 0022 & 2020 & 0000 \\
2000 & 2000 & 2000 & 2000 & 2000 & 2000 \\
1111 & 1111 & 2000 & 2000 & 0000 & 0000 \\
2000 & 1111 & 1111 & 0000 & 2000 & 0000 \\
0000 & 0000 & 1111 & 1111 & 2000 & 2000 \\
2000 & 0000 & 2000 & 1111 & 1111 & 0000 \\
2000 & 2000 & 0000 & 0000 & 1111 & 1111 \\
3012 & 1010 & 1001 & 1001 & 1100 & 1100 \\
3201 & 1001 & 1100 & 1100 & 1010 & 1010
\end{array}\right)
\]
It is an extremal Type II code. This code also gives rise to a supercurrent, see Section \ref{subsec:duncspinor}, specifically Remark \ref{rem:HKcodesupcurgen}.  
\subsection{Lattice Construction From Codes}\label{subsec:lattice}
Codes can be used to construct lattices. In this section we describe the construction of lattices from binary and $\mathbb{Z}_4$ codes. We begin by defining Euclidean lattices. 
\begin{defn}
Let $\mathbb{R}^n$ be the $n$-dimensional Euclidean space with inner product $x\cdot y=x_1y_1+\dots+x_ny_n$ for $x=(x_1,\dots,x_n)$ and $y=(y_1,\dots,y_n)$. 
\begin{enumerate}
\item An $n$-dimensional \textit{Euclidean lattice} $\Lambda$ is a free $\mathbb{Z}$-module spanned by $n$ vectors $e_{j}\in\mathbb{R}^n, 1 \leq j \leq n$ linearly independent in $\mathbb{R}^n$,  i.e. $$\Lambda=\left\{\sum_{j=1}^{n} n_{j} e_{j}: n_{j} \in \mathbb{Z}\right\}.$$
\item The \textit{norm} of a vector $x \in \Lambda$ is $||x||=\sqrt{x^2}=\sqrt{x\cdot x}$. %where $x\cdot x$ denotes the standard Euclidean inner product on $\mathbb{R}^n$. 
\item The \textit{dual lattice} $\Lambda^*$ is defined as $\Lambda^{*}=\{y: x \cdot y \in \mathbb{Z}~ \forall~ x \in \Lambda\}$ (which is obviously a lattice). The lattice $\Lambda$ is said to be \textit{unimodular} if $\|\Lambda\|^{2} \equiv \operatorname{det}\left(e_{i} \cdot e_{j}\right)=1 .$ 
\item The lattice $\Lambda$ is said to be \textit{integral} if $\Lambda\subseteq\Lambda^*$, i.e. $x \cdot y \in \mathbb{Z}$ for all $x, y \in \Lambda$ and \textit{self-dual} if $\Lambda=\Lambda^*$. 
\item The lattice $\Lambda$ is said to be \textit{even} if $||x||^2$ is even for all $x \in \Lambda$. The set $\Lambda(\sqrt{2})\equiv\{x\in\Lambda:||x||^2=2\}$ is called the \textit{root system}\footnote{In general, a root is a primitive vector that generates a reflection isometry of the lattice. For even lattices, the roots can have different norms. For even, unimodular lattices, the norm-2 vectors are all the roots.} of $\Lambda$.
\item Two $n$-dimensional lattices $\Lambda_1,\Lambda_2$ are said to be isomorphic if there exists a norm-preserving  $\mathbb{Z}$-module isomorphism between them. 
\item An automorphism of the lattice is a rotation $R\in\mathrm{SO}(n)$ which fixes the lattice. The group of all automorphisms (the group operation being composition) is called the automorphism group of $\Lambda$ and denoted by Aut$(\Lambda)$.  
\end{enumerate}
\end{defn}
It is easy to see that a lattice $\Lambda$ is self-dual if and only if it is integral and unimodular since $\left\|\Lambda^{*}\right\|=\|\Lambda\|^{-1}.$ Moreover it is a well known result that the dimension of an even self-dual lattice has to be a multiple of 8. Even self-dual lattices in 24
dimensions have been completely classified up to isomorphism.  There are 24 inequivalent such lattices called the \textit{Neimeier lattices} classified by their root system. There is a unique lattice with empty root system called the Leech lattice. We denote the Leech lattice by $\Lambda_{\mathrm{L}}$.
%
%One construction of the Leech lattice as an embedded lattice in $\mathbb{R}^{24}$ identifies it 
%with the set of vectors   $x:=(x_1, \dots, x_{24}) = c   (n_1, \dots , n_{24} )$
%where $c=1/\sqrt{8}$, and  $n_i \in 2\IZ + \epsilon$ with $\epsilon=0$ or $\epsilon =1$ for all $i$.
%We require moreover that $\sum_{i=1}^{24}  n_i = 4\epsilon ~\mod~8$,
%and finally we require that for each $x\in \Lambda$, each of the four sets:
%
%\begin{equation}
%\mathcal{C}_a(x) := \{ i \vert n_i = a ~\mod~ 4 \}  
%\end{equation} 
%
%for $a=0,1,2,3$ is either the empty set or is the set of nonempty entries in a codeword in the Golay %code. 
%
An explicit generator matrix for the Leech lattice is given below: the rows of the matrix below generates the Leech lattice under integer linear combinations.
\[
\frac{1}{2\sqrt{2}}\left(\begin{array}{llllllllllllllllllllllll}
8 & 0 & 0 & 0 & 0 & 0 & 0 & 0 & 0 & 0 & 0 & 0 & 0 & 0 & 0 & 0 & 0 & 0 & 0 & 0 & 0 & 0 & 0 & 0 \\
4 & 4 & 0 & 0 & 0 & 0 & 0 & 0 & 0 & 0 & 0 & 0 & 0 & 0 & 0 & 0 & 0 & 0 & 0 & 0 & 0 & 0 & 0 & 0 \\
4 & 0 & 4 & 0 & 0 & 0 & 0 & 0 & 0 & 0 & 0 & 0 & 0 & 0 & 0 & 0 & 0 & 0 & 0 & 0 & 0 & 0 & 0 & 0 \\
4 & 0 & 0 & 4 & 0 & 0 & 0 & 0 & 0 & 0 & 0 & 0 & 0 & 0 & 0 & 0 & 0 & 0 & 0 & 0 & 0 & 0 & 0 & 0 \\
4 & 0 & 0 & 0 & 4 & 0 & 0 & 0 & 0 & 0 & 0 & 0 & 0 & 0 & 0 & 0 & 0 & 0 & 0 & 0 & 0 & 0 & 0 & 0 \\
4 & 0 & 0 & 0 & 0 & 4 & 0 & 0 & 0 & 0 & 0 & 0 & 0 & 0 & 0 & 0 & 0 & 0 & 0 & 0 & 0 & 0 & 0 & 0 \\
4 & 0 & 0 & 0 & 0 & 0 & 4 & 0 & 0 & 0 & 0 & 0 & 0 & 0 & 0 & 0 & 0 & 0 & 0 & 0 & 0 & 0 & 0 & 0 \\
2 & 2 & 2 & 2 & 2 & 2 & 2 & 2 & 0 & 0 & 0 & 0 & 0 & 0 & 0 & 0 & 0 & 0 & 0 & 0 & 0 & 0 & 0 & 0 \\
4 & 0 & 0 & 0 & 0 & 0 & 0 & 0 & 4 & 0 & 0 & 0 & 0 & 0 & 0 & 0 & 0 & 0 & 0 & 0 & 0 & 0 & 0 & 0 \\
4 & 0 & 0 & 0 & 0 & 0 & 0 & 0 & 0 & 4 & 0 & 0 & 0 & 0 & 0 & 0 & 0 & 0 & 0 & 0 & 0 & 0 & 0 & 0 \\
4 & 0 & 0 & 0 & 0 & 0 & 0 & 0 & 0 & 0 & 4 & 0 & 0 & 0 & 0 & 0 & 0 & 0 & 0 & 0 & 0 & 0 & 0 & 0 \\
2 & 2 & 2 & 2 & 0 & 0 & 0 & 0 & 2 & 2 & 2 & 2 & 0 & 0 & 0 & 0 & 0 & 0 & 0 & 0 & 0 & 0 & 0 & 0 \\
4 & 0 & 0 & 0 & 0 & 0 & 0 & 0 & 0 & 0 & 0 & 0 & 4 & 0 & 0 & 0 & 0 & 0 & 0 & 0 & 0 & 0 & 0 & 0 \\
2 & 2 & 0 & 0 & 2 & 2 & 0 & 0 & 2 & 2 & 0 & 0 & 2 & 2 & 0 & 0 & 0 & 0 & 0 & 0 & 0 & 0 & 0 & 0 \\
2 & 0 & 2 & 0 & 2 & 0 & 2 & 0 & 2 & 0 & 2 & 0 & 2 & 0 & 2 & 0 & 0 & 0 & 0 & 0 & 0 & 0 & 0 & 0 \\
2 & 0 & 0 & 2 & 2 & 0 & 0 & 2 & 2 & 0 & 0 & 2 & 2 & 0 & 0 & 2 & 0 & 0 & 0 & 0 & 0 & 0 & 0 & 0 \\
4 & 0 & 0 & 0 & 0 & 0 & 0 & 0 & 0 & 0 & 0 & 0 & 0 & 0 & 0 & 0 & 4 & 0 & 0 & 0 & 0 & 0 & 0 & 0 \\
2 & 0 & 2 & 0 & 2 & 0 & 0 & 2 & 2 & 2 & 0 & 0 & 0 & 0 & 0 & 0 & 2 & 2 & 0 & 0 & 0 & 0 & 0 & 0 \\
2 & 0 & 0 & 2 & 2 & 2 & 0 & 0 & 2 & 0 & 2 & 0 & 0 & 0 & 0 & 0 & 2 & 0 & 2 & 0 & 0 & 0 & 0 & 0 \\
2 & 2 & 0 & 0 & 2 & 0 & 2 & 0 & 2 & 0 & 0 & 2 & 0 & 0 & 0 & 0 & 2 & 0 & 0 & 2 & 0 & 0 & 0 & 0 \\
0 & 2 & 2 & 2 & 2 & 0 & 0 & 0 & 2 & 0 & 0 & 0 & 2 & 0 & 0 & 0 & 2 & 0 & 0 & 0 & 2 & 0 & 0& 0 \\ 
0 & 0 & 0 & 0 & 0 & 0 & 0 & 0 & 2 & 2 & 0 & 0 & 2 & 2 & 0 & 0 & 2 & 2 & 0 & 0 & 2 & 2 & 0 & 0 \\ 0 & 0 & 0 & 0 & 0 & 0 & 0 & 0 & 2 & 0 & 2 & 0 & 2 & 0 & 2 & 0 & 2 & 0 & 2 & 0 & 2 & 0 & 2 & 0 \\ -3 & 1 & 1 & 1 & 1 & 1 & 1 & 1 & 1 & 1 & 1 & 1 & 1 & 1 & 1 & 1 & 1 & 1 & 1 & 1 & 1 & 1 & 1 & 1
\end{array}\right)
\]
We will give a more conceptual description of the Leech lattice in terms of the Golay code in the 
next section.
\\\\
\begin{comment}
\begin{table}[]
\centering
    \begin{tabular}{||c|c||}
    \hline
    $n_{\pm 1}=16$ & $n_{\pm 1}=24$  \\\hline
   $(0; 16; 8)=8\cdot 759$& $(0; 24; 0)=2^{12}$\\$(8; 16; 0)=8\cdot 759$ &\\$(2; 16; 6)=448\cdot 759$ & Total $2^{12}$\\   $(6; 16; 2)=448\cdot 759$ &\\ $(4; 16; 4)=2240\cdot 759$ &\\&\\Total $759\cdot 2^{12}$& \\\hline 
    \end{tabular}
    \caption{Number of codewords of each type in $\widehat{\mathcal{Q}}_4$.}
    \label{table:q4'}
\end{table}
\end{comment}
\noindent\textbf{Lattice construction from binary codes:}
We can construct Euclidean lattice from a code. For an $r$-dimensional binary code $\mathcal{C}$, define 
\begin{equation}
\Lambda(\mathcal{C}):=\frac{1}{\sqrt{2}}\mathcal{C}+\sqrt{2}\mathbb{Z}^r
\label{eq:latticecode}
\end{equation}
where $\mathcal{C}$ is now considered as a subset of $\mathbb{Z}^r$ via a map such as\footnote{This map is just an injective function. In particular it is not a group or ring homomorphism.}
\begin{equation}
\begin{split}    &\mathbb{F}_2\hookrightarrow\mathbb{Z}\\&0\mapsto 0,~~1\mapsto 1.
\end{split}
\end{equation}
\begin{prop}
The lattice $\Lambda(\mathcal{C})$ is integral if and only if $\mathcal{C}\subset\mathcal{C}^{\perp}$ and is even if and only if $\mathcal{C}$ is doubly-even. It is self-dual if and only if $\mathcal{C}$ is self-dual. Unimodularity of $\Lambda(\mathcal{C})$ is equivalent to $\mathrm{dim}~\mathcal{C}=\frac{1}{2}r$.  
\end{prop}
We can also construct another lattice by ``twisting'' this lattice as follows: define the subspaces $\Lambda_a(\mathcal{C}), a=0,1,2,3$ as follows: 
\begin{equation}
\begin{split}
&\Lambda_0(\mathcal{C}):=\frac{1}{\sqrt{2}}\mathcal{C}+\sqrt{2}\mathbb{Z}^r_{+},\\
&\Lambda_1(\mathcal{C}):=\frac{1}{\sqrt{2}}\mathcal{C}+\sqrt{2}\mathbb{Z}^r_{-},\\
&\Lambda_2(\mathcal{C}):=\frac{1}{\sqrt{2}}\mathcal{C}+\frac{1}{2\sqrt{2}}\underline{1}+\sqrt{2}\mathbb{Z}^r_{(-)^{r+1}},\\
&\Lambda_3(\mathcal{C}):=\frac{1}{\sqrt{2}}\mathcal{C}+\frac{1}{2\sqrt{2}}\underline{1}+\sqrt{2}\mathbb{Z}^r_{(-)^{r}},
\end{split}
\end{equation}
where 
\begin{equation}
    \begin{split}
\mathbb{Z}^r_{+}:=\{\bm{x}\in\mathbb{Z}^r:\bm{x}^2\in 2\mathbb{Z}\},\\
\mathbb{Z}^r_{-}:=\{\bm{x}\in\mathbb{Z}^r:\bm{x}^2\in 2\mathbb{Z}+1\},
\end{split}
\label{eq:}
\end{equation}
and $\underline{1}$ is a vector with all components 1.
Clearly $\Lambda(\mathcal{C})=\Lambda_0(\mathcal{C})\cup \Lambda_1(\mathcal{C})$. We define the twisted lattice by 
\begin{equation}
\widetilde{\Lambda}(\mathcal{C}):=\Lambda_0(\mathcal{C})\cup \Lambda_3(\mathcal{C}).
\label{eq:twistlatticecode}
\end{equation}
The twisted construction is particularly important for our case because of the following theorem.
\begin{thm}\emph{\cite{conway:1999}}
Let $\mathcal{G}$ be the binary Golay code. Then 
\begin{equation}
    \widetilde{\Lambda}(\mathcal{G})\cong\Lambda_{\mathrm{L}}.
\end{equation}
\end{thm}
This enables us to list all possible vectors of norm $4,6,8$ in the Leech lattice. This is listed in \cite[Table 4.13, Chapter 4]{conway:1999}. Finally one can construct lattices from $\mathbb{Z}_4$ codes as described in the theorem: 
\begin{thm}\emph{\cite[Theorem 4.1]{Bonnecaze97quaternaryquadratic}}
Let $\mathscr{C}$ be a self-dual $\mathbb{Z}_4$-code of length $n$ with
minimum Euclidean weight $d_E$. Then the lattice
\begin{equation}
    A(\mathscr{C}) =
\frac{1}{2}\left\{
\lambda \in \mathbb{Z}^n : \lambda \equiv x(\bmod~4) ~~\text{for some} ~~x \in \mathscr{C}\right\} 
\end{equation}
is an $n$-dimensional unimodular lattice with the minimum norm squared $\mathrm{min}\{4, d_E /4\}$.
Moreover, if $\mathscr{C}$ is Type II, then the lattice $A(\mathscr{C})$ is even and unimodular.
\label{thm:latticez4code}
\end{thm}

\section{Locality Of The Twisted Vertex Operators}\label{App:TwistedLocality}
\noindent We will closely follow \cite[Section 7,8]{Dolan:1989vr} to establish the locality relation \eqref{eq:locWbarWWWbar}. The following lemma is proved similar to the proof of  
\eqref{eq:wL0Vw-L0}. We record the proof for completness.
\begin{lemma}
For any untwisted state $\psi$, we have 
\begin{equation}\label{eq:appwL0Vt}
w^{L_0^T}V_T(\psi,z)w^{-L_0^T}=w^{h_{\psi}}V_T(\psi,zw).    
\end{equation}
\begin{proof}
Using \eqref{eq:defW}, we have 
\[
w^{L_0^T}
V_T(\psi,z)w^{-L_0^T}\chi=w^{-h_{\chi}}w^{L_0^T}e^{-zL_{-1}^T}W(\chi,-z)\psi,
\]
where $\chi$ is a general twisted quasi-primary state.
We now use the Baker-Campbell-Hausdorff formula \eqref{eq:bch}
to write 
\[
w^{L_0^T}e^{zL_{-1}^T}=e^{L_0^T\ln w}e^{zL_{-1}^T}.
\]
Using the fact that $[L_0^T,L_{-1}^T]=L_{-1}^T$, we see that the above theorem implies 
\[
e^{L_0^T\ln w}e^{zL_{-1}^T}=e^{\exp(\ln w)L_{-1}^T}w^{L_0^T}.
\]
Thus we get 
\begin{equation}
w^{L_0^T}e^{zL_{-1}^T}=e^{zwL_{-1}^T}w^{L_0^T}.
\label{eq:bchform}
\end{equation}
This implies
\begin{equation}
    w^{L_0^T}V_T(\psi,z)w^{-L_0^T}\chi=w^{-h_{\chi}}e^{-zwL_{-1}^T}w^{L_0^T}W(\chi,-z)\psi.
\end{equation}
Next using \cite[eq. 2.18]{Dolan:1994st} for the twisted vertex operators $\widetilde{V}(\chi,z)$, we get
\begin{equation}
w^{L_0^T}W(\chi,z)w^{-L_0}=w^{h_{\chi}}W(\chi,zw)
\end{equation}
Thus we get 
\begin{equation}
    w^{L_0^T}V_T(\psi,z)w^{-L_0^T}\chi=w^{h_{\psi}}e^{-zwL_{-1}^T}W(\chi,-zw)\psi=w^{h_{\psi}}V_T(\psi,zw)\chi.
\end{equation}
\end{proof}
\end{lemma}
\subsection{Locality Relation $\overline{W}(\chi,z)W(\chi',w)=(-1)^{|\chi||\chi'|} \overline{W}(\chi',w)W(\chi,z)$}   
For quasi-primary states $\chi,\chi'\in\mathscr{H}_T^-(\Lambda_{\mathrm{L}})$, following the steps in \cite[Section 7]{Dolan:1989vr}, we see that the required locality relation is equivalent to proving 
\begin{equation}
\left\langle\overline{\chi}'\left|V_{T}(\phi,-z)\right| \chi\right\rangle z^{2 h_{\chi}} = \left\langle\overline{\chi}\left|V_{T}\left(e^{L_{1} / z} \phi, z\right)\right| \chi'\right\rangle z^{2 h_{\chi'}},
\label{eq:locWbarW}
\end{equation} 
where $\phi$ is an arbitrary untwisted state.
Now observe that 
\[
\begin{split}
V_{T}\left(e^{L_{1} / z} \phi, z\right)^{\dagger}&=\left(V_{T}\left(e^{L_{1} / z}z^{2L_0} \phi, z\right)z^{-2h_{\phi}}\right)^{\dagger}\\&=V_{T}\left(e^{L_{1} / z}z^{2L_0} \phi, z\right)^{\dagger}(z^*)^{-2h_{\phi}}\\&=V_{T}\left(\overline{\phi}, 1/z^*\right)(z^*)^{-2h_{\phi}}
\end{split}
\]
which implies that 
\[
V_{T}\left(\overline{\phi}, 1/z^*\right)^{\dagger}z^{-2h_{\phi}}=V_{T}\left(e^{L_{1} / z} \phi, z\right).
\]
So the right hand side of \eqref{eq:locWbarW} becomes (see eq.  \eqref{eq:fabcdefcor} for notation)
\[
\begin{split}
\left\langle\overline{\chi}\left|V_{T}\left(\overline{\phi}, 1/z^*\right)^{\dagger}\right| \chi'\right\rangle z^{2 h_{\chi'}-2h_{\phi}}&=\left\langle\chi'\left|V_{T}\left(\overline{\phi}, 1/z^*\right)\right| \overline{\chi}\right\rangle^* z^{2 h_{\chi'}-2h_{\phi}}\\&=\left(\left\langle\chi'\left|(z^{-L_0})^\dagger(z^*)^{L_0}V_{T}\left(\overline{\phi}, 1/z^*\right)(z^*)^{-L_0}\right| \overline{\chi}\right\rangle(z^*)^{h_{\chi}}\right)^*z^{2 h_{\chi'}-2h_{\phi}}\\&=\left((z^*)^{-h_{\chi'}}\left\langle\chi'\left|(z^*)^{L_0}V_{T}\left(\overline{\phi}, 1/z^*\right)(z^*)^{-L_0}\right| \overline{\chi}\right\rangle(z^*)^{h_{\chi}}\right)^*z^{2 h_{\chi'}-2h_{\phi}}\\&=\left\langle\chi'\left|V_{T}\left(\overline{\phi}, 1\right)\right| \overline{\chi}\right\rangle^*z^{h_{\chi}+h_{\chi'}-h_{\phi}}\\&=(f_{\overline{\chi}'\overline{\phi}\overline{\chi}})^*z^{h_{\chi}+h_{\chi'}-h_{\phi}},
\end{split}
\]
the left hand side of \eqref{eq:locWbarW} becomes 
\begin{equation}
\begin{aligned}
\left\langle\overline{\chi}^{\prime}\left|V_T(\phi,-z)\right| \chi\right\rangle z^{2 h_\chi}&=\left\langle\overline{\chi}^{\prime}\left|(-1 / z)^{-L_0}(-1 / z)^{L_0} V_T(\phi,-z)(-1 / z)^{-L_0}(-1 / z)^{L_0}\right| \chi\right\rangle z^{2 h_\chi} \\
& =(-1 / z)^{-h_{\chi^{\prime}}}\left\langle\overline{\chi}^{\prime}\left|V_T(\phi, 1)\right| \chi\right\rangle(-1 / z)^{h_\phi}(-1 / z)^{h_\chi} z^{2 h_\chi} \\
& =(-1)^{2 h_\chi}(-z)^{h_{\chi^{\prime}}+h_\chi-h_\phi} f_{\chi^{\prime} \phi \chi} \\&=(-1)^{h_{\chi'}+h_{\chi}+h_{\phi}}z^{h_{\chi}+h_{\chi'}-h_{\phi}}f_{\chi^{\prime} \phi \chi},
\end{aligned}   
\end{equation}
where we have used $(-z)^{-k}z^{2k}=(-z)^k$ for $k\in\frac{1}{2}\mathbb{Z}$ which is allowed\footnote{Since if $z$ is in the upper half-plane then $-z$ is in the lower half-plane and vice-versa.} in the principal branch of square root. 
\begin{comment}
(-1)^{2 h_\chi}(-i\,\text{sgn}(\text{Im}(z)))^{2h_{\chi^{\prime}}+2h_\chi-2h_\phi}z^{h_{\chi^{\prime}}+h_\chi-h_\phi}     
where we used\footnote{This relation can be proved by using $(-z)^k=e^{k\ln(-z)}$.}
\begin{equation}
    (-z)^{k}=z^k(-i\,\text{sgn}(\text{Im}(z)))^{2k}
\end{equation}
where we used the fact that $2 h_\chi \in \mathbb{Z}$. Since $h_\chi+h_{\chi'}\in\mathbb{Z}$, we get 
\begin{equation}
\begin{split}
\left\langle\overline{\chi}'\left|V_{T}(\phi,-z)\right| \chi\right\rangle z^{2 h_{\chi}}=(-1)^{h_{\chi'}+h_{\chi}+h_{\phi}}z^{h_{\chi}+h_{\chi'}-h_{\phi}}f_{\chi'\phi\chi}
\end{split}    
\end{equation}
\end{comment}
So the locality relation is equivalent to proving  
\begin{equation}
(f_{\overline{\chi}'\overline{\phi}\overline{\chi}})^*=(-1)^{|\chi||\chi'|} (-1)^{h_{\chi'}+h_{\chi}+h_{\phi}}f_{\chi'\phi\chi}.
\label{eq:locWbarWlasteq}
\end{equation}
For $\chi,\chi'\in\mathscr{H}_T^+(\Lambda)$ (the sign is $+$ in this case), \eqref{eq:locWbarWlasteq} is proved in \cite{Dolan:1989vr}. We prove it for $\chi,\chi'\in\mathscr{H}_T^-(\Lambda)$. 
Following the calculations on \cite[Page 547]{Dolan:1989vr}, we have 
\[
\begin{aligned}
\langle\chi|V_{T}(\phi, 1)| &\chi'\rangle^{*} =\left\langle M(\chi)\left|V_{T}(\phi, 1)\right| M(\chi')\right\rangle \\
&=(-1)^{h_{\phi}}\left\langle M(\chi)\left|\theta V_{T}(\overline{\phi}, 1) \theta\right| M(\chi')\right\rangle \\
&=(-1)^{h_{\phi}}\left\langle M(\chi)\left|\theta e^{i\pi L_0^T}e^{-i\pi L_0^T} V_{T}(\overline{\phi}, 1)  e^{i\pi L_0^T}e^{-i\pi L_0^T}\theta\right| M(\chi')\right\rangle \\&=(-1)^{h_{\phi}}\left\langle \overline{\chi}\left|e^{-i\pi L_0^T} V_{T}(\overline{\phi}, 1)  e^{i\pi L_0^T}\right| \overline{\chi}'\right\rangle\\
&=(-1)^{h_{\phi}}e^{-i\pi h_{\chi}}e^{i\pi h_{\chi'}}\left\langle\overline{\chi}\left|V_{T}(\overline{\phi}, 1) \right| \overline{\chi}'\right\rangle \\
&=-(-1)^{h_{\chi}+h_{\phi}+h_{\chi'}}\left\langle\overline{\chi}\left|V_{T}(\overline{\phi}, 1)\right| \overline{\chi}'\right\rangle,
\end{aligned}
\]
where we used the fact that $h_{\chi}\in\mathbb{Z}+\frac{1}{2}$.
In the last step we used the definition of the conjugation operation as in \eqref{eq:modconjdef}. Thus we have established the locality relation.

\subsection{Locality Relation $W(\chi,z)\overline{W}(\chi',w)=(-1)^{|\chi||\chi'|}W(\chi',w)\overline{W}(\chi,z)$}

This is the hard one! 
We begin by proving a lemma.
\label{lemma:VTpsichi}
\begin{lemma}
We can choose a zero-momentum untwisted state $\psi$ such that $\langle \mu|V_T(\psi,w)$ and equivalently $V_T(\psi,w)|\varrho\rangle$ is an arbitrary twisted state.
\begin{proof}
We take 
$$
\psi=e^{(1 /w) L_{1}} w^{2 L_{0}} \overline{\psi}^{\prime}, \quad w=1 / \tilde{w}^{*}.
$$
By Remark \ref{rem:Udaggerrep}, we have 
$$
V_{T}(\psi, w)=V_{T}\left(\psi^{\prime}, \tilde{w}\right)^{\dagger}
$$
where we have defined
$$
\psi^{\prime}=e^{-\Delta(\tilde{w})} \tilde{\psi}
$$
where
$$
\tilde{\psi}=\left(\prod_{a=1}^{M} a_{-n_{a}}^{j_{a}}\right)|0\rangle
$$
and 
\begin{equation}
\Delta(z)=\frac{1}{2} \sum_{n, m \geq 0 \atop m+n>0}\left(\begin{array}{c}
-\frac{1}{2} \\
m
\end{array}\right)\left(\begin{array}{l}
\frac{1}{2} \\
n
\end{array}\right) \frac{(-z)^{-m-n}}{m+n} a_{m} \cdot a_{n}.
\end{equation}
We then have
$$
V_{T}(\psi, w)=V_{T}^{0}(\tilde{\psi}, \tilde{w})^{\dagger}
$$
We Laurent expand the vertex operator
\[
V_{T}^{0}(\tilde{\psi}, \tilde{w})=\sum_{n}(V_{T}^{0})_{n}(\tilde{\psi})\tilde{w}^{-n-h_{\tilde{\psi}}}
\]
We see that the mode $(V_{T}^{0})_{-h_{\tilde{\psi}}+M/2}$ creates twisted states containing $M$ oscillators with conformal weight $r/16+h_{\tilde{\psi}}-M/2$ by acting on twisted ground states. Next observe that the number of linearly independent states of conformal weight $r/16+h_{\tilde{\psi}}-M/2$ obtained by acting $M$ oscillators on a twisted ground state is the number of ways we can write $h_{\tilde{\psi}}-M/2$ as sum of $M$ positive strictly half-integers. Let this number be $p^{\mathbb{Z}+\frac{1}{2}}_M(h_{\tilde{\psi}}-M/2)$. On the other hand the number of linearly independent untwisted states with conformal weight $h_{\tilde{\psi}}$ obtained by acting $M$ oscillators on $|0\rangle$ is the number of ways we can write $h_{\tilde{\psi}}$ as sum of $M$ positive integers. Let this number be $p_{M}^{\mathbb{Z}}(h_{\tilde{\psi}})$. We want to show that 
\begin{equation}
    p^{\mathbb{Z}+\frac{1}{2}}_M(h_{\tilde{\psi}}-M/2)=p_{M}^{\mathbb{Z}}(h_{\tilde{\psi}}).
\end{equation}
Indeed if 
\begin{equation}
h_{\tilde{\psi}}=\sum_{a=1}^Mn_a,\quad n_a\in\mathbb{N}
\end{equation}
then we get a partition of $h_{\tilde{\psi}}-M/2$ into positive strictly half-integers:  $p^{\mathbb{Z}+\frac{1}{2}}_M(h_{\tilde{\psi}}-M/2)$ since 
\begin{equation}
    h_{\tilde{\psi}}-\frac{M}{2}=\sum_{a=1}^M\left(n_a-\frac{1}{2}\right).
\end{equation}
Thus 
\begin{equation}
    p^{\mathbb{Z}+\frac{1}{2}}_M(h_{\tilde{\psi}}-M/2)\geq p^{\mathbb{Z}}_M(h_{\tilde{\psi}}).
\end{equation}
Similarly 
\begin{equation}
    p^{\mathbb{Z}}_M(h_{\tilde{\psi}})\geq p^{\mathbb{Z}+\frac{1}{2}}_M(h_{\tilde{\psi}}-M/2).
\end{equation}
Thus we can get any twisted state by acting with suitable linear combinations of the modes of $V_{T}^{0}(\tilde{\psi}, \tilde{w})$.
\end{proof}
\end{lemma}

%\begin{remark}The proof of Lemma \ref{lemma:VTpsichi} also shows that we can choose a zero momentum state $|\psi\rangle$ such that $\langle\mu|V_{T}^{0}(\tilde{\psi}, \tilde{w})$ is the twisted ground state for $|\mu\rangle$ a twisted state of conformal weight 2. \end{remark}

\begin{lemma}
The relation $\langle\mu|W(\chi, z) V(\psi, w) \overline{W}(\varphi, \zeta)| \varrho\rangle=-\langle\mu|W(\varphi, \zeta) V(\psi, w) \overline{W}(\chi, z)| \varrho\rangle$ for $\mu,\chi,\varphi,\varrho$ twisted ground states is equivalent to the locality relation $W(\chi,z)\overline{W}(\varphi,w)=- W(\varphi,w)\overline{W}(\chi,z)$ for twisted ground states $\chi,\varphi$.
\label{lemma:WVbarWrel}
\begin{proof}
Using the locality relation of $V_TW$ and $V\overline{W}$, we have that 
\[
\langle\mu|W(\chi, z) V(\psi, w) \overline{W}(\varphi, \zeta)| \varrho\rangle=\langle\mu|V_T(\psi, w)W(\chi, z)  \overline{W}(\varphi, \zeta)| \varrho\rangle.
\]
But since $\langle\mu|V_T(\psi, w)$ can be taken to be a general twisted state, we see that 
\[
\begin{split}
&\langle\mu|W(\chi, z) V(\psi, w) \overline{W}(\varphi, \zeta)| \varrho\rangle=-\langle\mu|W(\varphi, \zeta) V(\psi, w) \overline{W}(\chi, z)| \varrho\rangle\\\implies &\langle\mu|V_T(\psi, w)W(\chi, z)  \overline{W}(\varphi, \zeta)| \varrho\rangle=- \langle\mu|V_T(\psi, w)W(\varphi, \zeta) \overline{W}(\chi, z)| \varrho\rangle\\\implies & W(\chi, z)  \overline{W}(\varphi, \zeta)| \varrho\rangle=- W(\varphi, \zeta) \overline{W}(\chi, z)| \varrho\rangle
\end{split}
\]
Next, for an arbitrary untwisted state $\psi$, we have from above 
\[
\begin{split}
V_T(\psi,w)W(\chi, z)  \overline{W}(\varphi, \zeta)| \varrho\rangle=- V_T(\psi,w)W(\varphi, \zeta) \overline{W}(\chi, z)| \varrho\rangle\\\implies W(\chi, z)  \overline{W}(\varphi, \zeta)V_T(\psi,w)| \varrho\rangle=- W(\varphi, \zeta) \overline{W}(\chi, z)V_T(\psi,w)| \varrho\rangle  
\end{split}
\] 
Since $V_T(\psi,w)| \varrho\rangle$ can be taken to be general twisted state, we get the required statement. Here $\chi$ and $\varphi$ are still twisted ground states.
\end{proof}
\end{lemma}
\begin{thm}
For twisted ground states $\chi$ and $\varphi$, the locality relation 
\begin{equation}
    W(\chi,z)\overline{W}(\varphi,w)=-W(\varphi,w)\overline{W}(\chi,z)
\end{equation}
holds.
\begin{proof}
In view of Lemma \ref{lemma:WVbarWrel} and Lemma \ref{lemma:VTpsichi}, we prove that 
\begin{equation}
\langle\mu|W(\chi, z) V(\psi, w) \overline{W}(\varphi, \zeta)| \varrho\rangle=-\langle\mu|W(\varphi, \zeta) V(\psi, w) \overline{W}(\chi, z)| \varrho\rangle,    
\end{equation}
for twisted ground states $\mu,\chi,\varphi,\varrho$ and a general zero momentum untwisted state $\psi$. Then by Lemma \ref{lemma:WVbarWrel}, the required locality relation follows. We will closely follow \cite[Section 9]{Dolan:1989vr} for this proof.
We begin by simplifying the left hand side. From \cite[eq. 4.2]{Dolan:1989vr} we have 
\begin{equation}\label{eq:Wactonpsilambsum}
    W(\chi,z)\psi=e^{zL_{-1}^T}\sum_{\lambda\in\Lambda_{\mathrm{L}}}
    \langle\lambda|:e^{B(z)}e^{A(z)}:|\psi\rangle e_{\lambda}\chi
\end{equation}
where $A(z)$ is as in \eqref{eq:A(z)exp} and 
\begin{equation}
    B(z):=B_+(z)+B_-(z)
\end{equation}
with 
\begin{equation}
    B_{\pm}(z):=-\sum_{\substack{n\geq 0\\s\in\left(\mathbb{Z}+\frac{1}{2}\right)_{\pm}}}{-s\choose n}\frac{(-z)^{-n-s}}{s}a_n\cdot c_s,
\end{equation}
where $\left(\mathbb{Z}+\frac{1}{2}\right)_{\pm}$ denotes positive and negative half-integers respectively. Note that $\langle\lambda|:e^{B(z)}e^{A(z)}:|\psi\rangle$ is an operator on the twisted Hilbert space, and that it is normal ordered and acting on a twisted ground state. Since $c_s$ annihilates twisted ground states for $s>0$ we see that normal ordering gives
\begin{equation}
W(\chi,z)\psi=e^{zL_{-1}^T}\sum_{\lambda\in\Lambda_{\mathrm{L}}}
    \langle\lambda|:e^{B_-(z)}e^{A(z)}:|\psi\rangle e_{\lambda}\chi.    
\end{equation}
Next by \eqref{eq:defWbar} we have 
\begin{equation}
    \overline{W}(\varphi,\zeta)|\varrho\rangle=\zeta^{-2h_{\varphi}}W(\overline{\varphi},1/\zeta^*)^\dagger|\varrho\rangle
\end{equation}
where we used the fact that $L_1^T$ annihilates the twisted ground state. Replacing $\psi$ by $V(\psi, w) \overline{W}(\varphi, \zeta)| \varrho\rangle $ in \eqref{eq:Wactonpsilambsum} we obtain
\begin{equation}\label{eq:muWVWbarrho}
\begin{split}
\langle\mu|W(\chi, z) V(\psi, w) \overline{W}(\varphi, \zeta)| \varrho\rangle &=\sum_{\lambda \in \Lambda_{\mathrm{L}}}\left\langle\mu\left|e^{zL_{-1}^T}\left\langle\lambda \left|e^{B_-(z)}e^{A(z)}V(\psi, w) \overline{W}(\varphi, \zeta)\right| \varrho\right\rangle e_{\lambda}\right| \chi\right\rangle\\&=\sum_{\lambda \in \Lambda_{\mathrm{L}}}\left\langle\lambda \left|\left\langle e^{B_-(z)^\dagger}e^{zL_{1}^T}\mu\left|e_{\lambda}\right| \chi\right\rangle e^{A(z)}V(\psi, w) W(\overline{\varphi},1/\zeta^*)^\dagger\right| \varrho\right\rangle\zeta^{-2h_{\varphi}}\\&=\sum_{\lambda \in \Lambda_{\mathrm{L}}}\left\langle \mu\left|e_{\lambda}\right| \chi\right\rangle\left\langle\lambda \left|e^{A(z)}V(\psi, w) W(\overline{\varphi},1/\zeta^*)^\dagger\right| \varrho\right\rangle\zeta^{-2h_{\varphi}},  \end{split}    
\end{equation}
where we used the fact that $B_-(z)^\dagger$ and $L_1^T$ annihilates the twisted  ground state $\mu$.
Again using \eqref{eq:Wactonpsilambsum} we have 
\begin{equation}
    \begin{split}
     W(\overline{\varphi},1/\zeta^*)\left(e^{A(z)}V(\psi, w)\right)^\dagger|\lambda\rangle&=e^{L_{-1}^T/\zeta^*}\sum_{\lambda' \in \Lambda_{\mathrm{L}}}\left\langle\lambda' \left|e^{B_-(1/\zeta^*)}e^{A(1/\zeta^*)}\left(e^{A(z)}V(\psi, w)\right)^\dagger\right| \lambda\right\rangle e_{\lambda'}|\overline{\varphi}\rangle\\&=e^{L_{-1}^T/\zeta^*}\left\langle\lambda \left|e^{B_-(1/\zeta^*)}e^{A(1/\zeta^*)}\left(e^{A(z)}V(\psi, w)\right)^\dagger\right| \lambda\right\rangle e_{\lambda}|\overline{\varphi}\rangle
    \end{split}
\end{equation}
where we used the fact that $\langle\lambda| \mathcal{O}|\lambda'\rangle=0$ if $\lambda\neq\lambda'$ for any operator $\mathcal{O}$ which does not contain $e^{q\cdot (\lambda'-\lambda)}$ (see \eqref{eq:mudorqactonlamb}). Thus we have 
\begin{equation}\label{eq:lambeAVWdag}
\left\langle\lambda \left|e^{A(z)}V(\psi, w) W(\overline{\varphi},1/\zeta^*)^\dagger\right.\right.=\left\langle\lambda \left|e^{A(z)}V(\psi, w)e^{\overline{A}(\zeta)}e^{B_-(1/\zeta^*)^\dagger}\right| \lambda\right\rangle \langle\overline{\varphi}|e^\dagger_{\lambda}e^{L_{1}^T/\zeta^*},    
\end{equation}
where 
\begin{equation}
    \overline{A}(z)=A(1/z^*)^{\dagger}.
\end{equation}
Substituting eq.  \eqref{eq:lambeAVWdag} into eq.  \eqref{eq:muWVWbarrho} and using the fact that $L_1^T$ and $B_-(1/\zeta^*)^\dagger$ annihilates twisted ground states, we get in the region $|z|>|w|>|\zeta|$
\begin{equation}
    \begin{aligned}
&\langle\mu|W(\chi, z) V(\psi, w) \overline{W}(\varphi, \zeta)| \varrho\rangle =\sum_{\lambda \in \Lambda_{\mathrm{L}}}\left\langle\mu\left|e_{\lambda}\right| \chi\right\rangle\left\langle\lambda\left|e^{A(z)} V(\psi, w) e^{\overline{A}(\zeta)}\right| \lambda\right\rangle\left\langle\overline{\varphi}\left|e_{\lambda}^{\dagger}\right| \varrho\right\rangle \zeta^{-2 h_{\varphi}} .
\end{aligned}
\label{eq:appinteqWVW}
\end{equation}
Let 
\begin{equation}
    \Delta(-z)=\frac{1}{2} \sum_{n, m \geq 0 \atop m+n>0}\left(\begin{array}{c}
-\frac{1}{2} \\
m
\end{array}\right)\left(\begin{array}{l}
\frac{1}{2} \\
n
\end{array}\right) \frac{(-z)^{-m-n}}{m+n} a_{m} \cdot a_{n}.
\end{equation}
Then we have 
\begin{equation}
\begin{split}
    &\left\langle\lambda\left|e^{A(z)} V(\psi, w) e^{\overline{A}(\zeta)}\right| \lambda\right\rangle =\left\langle\lambda\left|e^{-\frac{1}{2}p^2\ln(-4z)}e^{\Delta(-z)} V(\psi, w) e^{\overline{\Delta}(-\zeta)}e^{-\frac{1}{2}p^2\ln(-4/\zeta)}\right| \lambda\right\rangle\\&=(-4z)^{-\frac{\lambda^2}{2}}(-4/\zeta)^{-\frac{\lambda^2}{2}}\left\langle\lambda\left|e^{\Delta(-z)} V(\psi, w) e^{\overline{\Delta}(-\zeta)}\right| \lambda\right\rangle\\&=\left(-4\sqrt{z/\zeta}\right)^{-\lambda^2}\left\langle\lambda\left|e^{\Delta(-z)}\left(\sqrt{z\zeta}\right)^{L_0}\left(\sqrt{z\zeta}\right)^{-L_0}V(\psi, w)\left(\sqrt{z\zeta}\right)^{L_0}\left(\sqrt{z\zeta}\right)^{-L_0} e^{\overline{\Delta}(-\zeta)}\right| \lambda\right\rangle\\&=(-4u^{-1})^{-\lambda^2}(z\zeta)^{-\frac{1}{2}h_{\psi}}\left\langle\lambda\left|e^{\Delta(-z)}\left(\sqrt{z\zeta}\right)^{L_0}V(\psi, y)\left(\sqrt{z\zeta}\right)^{-L_0} e^{\overline{\Delta}(-\zeta)}\right| \lambda\right\rangle
\end{split}    
\end{equation}
where we defined new variables
\begin{equation}
    u=\sqrt{\frac{\zeta}{z}},\quad y=\frac{w}{\sqrt{z\zeta}}
\end{equation}
and used \eqref{eq:zVz}. Now using the translation property 
\begin{equation}
    w^{L_0}a_nw^{-L_0}=w^{-n}a_n,
\end{equation}
we see that
\begin{equation}
    \left(\sqrt{z\zeta}\right)^{-L_0}\Delta(-z)\left(\sqrt{z\zeta}\right)^{L_0}=\Delta(-u^{-1}).
\end{equation}
Similarly using $a_n^{\dagger}=a_{-n}$ we have 
\begin{equation}
     \left(\sqrt{z\zeta}\right)^{-L_0}\overline{\Delta}(-\zeta)\left(\sqrt{z\zeta}\right)^{L_0}=\Delta(-u),
\end{equation}
So we get
\begin{equation}
    \begin{split}
    &\left\langle\lambda\left|e^{A(z)}V(\psi, w) e^{\overline{A}(\zeta)}\right| \lambda\right\rangle \\&=(-4u^{-1})^{-\lambda^2}(z\zeta)^{-\frac{1}{2}h_{\psi}}\left\langle\lambda\left|\left(\sqrt{z\zeta}\right)^{L_0}e^{\Delta(-u^{-1})}V(\psi, y) e^{\overline{\Delta}(-u)}\left(\sqrt{z\zeta}\right)^{-L_0}\right| \lambda\right\rangle\\&=(-4u^{-1})^{-\lambda^2}(z\zeta)^{-\frac{1}{2}h_{\psi}}\left\langle\lambda\left|e^{\Delta(-u^{-1})}V(\psi, y) e^{\overline{\Delta}(-u)}\right| \lambda\right\rangle
\end{split} 
\label{eq:appADelta}
\end{equation}
where we used the fact that 
\begin{equation}
    \langle\lambda|\left(\sqrt{z\zeta}\right)^{L_0}=\langle\lambda|\left(\left(\sqrt{z^*\zeta^*}\right)^{L_0}\right)^{\dagger}=\left(\sqrt{z\zeta}\right)^{\frac{\lambda^2}{2}}\langle\lambda|.
\end{equation}
Note that \eqref{eq:appADelta} is valid in the region $|u|<|y|<|u|^{-1}$ and $|u|<1$. The next step is to simplify \eqref{eq:appADelta}. For that we express the vertex operator $V(\psi,y)$ as exponential. Suppose 
\begin{equation}
    \psi=\left(\prod_{a=1}^M a^{j_a}_{-m_a}\right)|0\rangle.
\end{equation}
Then 
\begin{equation}
\begin{split}
    V(\psi,y)&=:\prod_{a=1}^M\frac{i}{(m_a-1)!}\frac{d^{m_a}X^{j_a}}{dy^{m_a}}:\\
    &=:\prod_{a=1}^M\mathscr{D}^{(m_a)}\sum_{n\in\mathbb{Z}}a^{j_a}_{n}y^{-n-1}:
\end{split}
\end{equation}
where 
\begin{equation}
   \mathscr{D}^{(m_a)}:=\frac{1}{(m_a-1)!}\frac{d^{m_a-1}}{dy^{m_a-1}}, 
\end{equation}
and $p^j=a_0^j$. One can now write this as an exponential. To do this we introduce the variables $\rho_a^{j_a}$ for $1\leq a\leq M$ and $1\leq j_a\leq 24$. If we make a fixed choice of $M$-tuple $(j_1, \dots, j_M)$ then, relative to this choice we define  
\begin{equation}
    \frac{\partial}{\partial\rho}:=\frac{\partial}{\partial \rho^{j_1}_1}\dots\frac{\partial}{\partial\rho^{j_M}_M}=\prod_{a=1}^M\frac{\partial}{\partial \rho^{j_a}_1}.
\end{equation}
We do not explicitly indicate the $M$-tuple $(j_1, \dots, j_M)$ in the notation 
$\frac{\partial}{\partial \rho}$ to avoid cluttering the notation. 
Next introduce vectors 
\begin{equation}
    \begin{split}
      &(\rho_{n}^<)^{k} =\sum_{a=1}^{M} \sqrt{n} \rho_{a}^{k} \mathscr{D}^{\left(m_{a}\right)} y^{n-1}=\sum_{a=1}^{M} \sqrt{n}\left(\begin{array}{c}
n-1 \\
m_{a}-1
\end{array}\right) \rho_{a}^{k} y^{n-m_{a}}, \\
&(\rho_{n}^>)^{k} =\sum_{a=1}^{M} \sqrt{n} \rho_{a}^{k} \mathscr{D}^{\left(m_{a}\right)} y^{-n-1}=\sum_{a=1}^{M} \sqrt{n}\left(\begin{array}{c}
-n-1 \\
m_{a}-1
\end{array}\right) \rho_{a}^{k} y^{-n-m_{a}}, \\
&(\rho_{0})^{k} =\sum_{a=1}^{M} \rho_{a}^{k} \mathscr{D}^{\left(m_{a}\right)} y^{-1}=\sum_{a=1}^{M}\left(\begin{array}{c}
-1 \\
m_{a}-1
\end{array}\right) \rho_{a}^{k} y^{-m_{a}} .  
    \end{split}
\end{equation}
and then define
\begin{equation}
    \begin{aligned}
\rho^{<} a^{\dagger} &=\sum_{n>0} \frac{\rho_{n}^{<} \cdot a_{-n}}{\sqrt{n}}, \quad \rho^{>} a=\sum_{n>0} \frac{\rho_{n}^{>} \cdot a_{n}}{\sqrt{n}}.
\end{aligned}
\end{equation}
Then we have 
\begin{equation}
   V(\psi,y)=\frac{\partial}{\partial\rho}\left.\left(e^{\rho^<a^{\dagger}}e^{p\cdot\rho_0}e^{\rho^>a}\right)\right|_{\{\rho_a^k=0\}}. 
\end{equation}
Next we have 
\begin{equation}
    \Delta(-u^{-1})=\sum_{n>0}{-\frac{1}{2}\choose n}\frac{(-u^{-1})^{-n}}{n}a_n\cdot p+\frac{1}{2} \sum_{n, m > 0}\left(\begin{array}{c}
-\frac{1}{2} \\
m
\end{array}\right)\left(\begin{array}{l}
\frac{1}{2} \\
n
\end{array}\right) \frac{(-u^{-1})^{-m-n}}{m+n} a_{m} \cdot a_{n}.
\end{equation}
Defining 
\begin{equation}
    T_{mn}=\frac{\sqrt{nm}}{n+m}\nu_n\nu_m,\quad \nu_n=(-u)^n{-\frac{1}{2}\choose n},
\end{equation}
we can write 
\begin{equation}
    \Delta(-u^{-1})=\widetilde{p}a+\frac{1}{2}aTa
\end{equation}
where 
\begin{equation}
    \widetilde{p}_n^j=\frac{p^j\nu_n}{\sqrt{n}},\quad \widetilde{p}a=\sum_{n>0}\widetilde{p}_n\cdot\frac{a_n}{\sqrt{n}},\quad aTa=\sum_{m,n>0}T_{mn}\frac{a_m}{\sqrt{m}}\cdot\frac{a_n}{\sqrt{n}} 
\end{equation}
Thus we have
\begin{equation}
\begin{split}
     \left\langle\lambda\left|e^{A(z)} V(\psi, w) e^{\overline{A}(\zeta)}\right| \lambda\right\rangle &=(-4u^{-1})^{-\lambda^2}(z\zeta)^{-\frac{1}{2}h_{\psi}}\frac{\partial}{\partial \rho}\left\langle\lambda\left|\exp\left(\widetilde{p}a+\frac{1}{2}aTa\right)\left(e^{\rho^<a^{\dagger}}e^{p\cdot\rho_0}e^{\rho^>a}\right)\right.\right.\\&\hspace{6cm}\times\left.\left.\left.\exp\left(\widetilde{p}a^{\dagger}+\frac{1}{2}a^{\dagger}Ta^{\dagger}\right)\right|\lambda\right\rangle\right|_{\{\rho_a^k=0\}}\\&=(-4u^{-1})^{-\lambda^2}(z\zeta)^{-\frac{1}{2}h_{\psi}}\frac{\partial}{\partial \rho}e^{\lambda\cdot\rho_0}\left\langle\lambda\left|\exp\left(\widetilde{\lambda}a+\frac{1}{2}aTa\right)\left(e^{\rho^<a^{\dagger}}e^{\rho^>a}\right)\right.\right.\\&\hspace{6cm}\times\left.\left.\left.\exp\left(\widetilde{\lambda}a^{\dagger}+\frac{1}{2}a^{\dagger}Ta^{\dagger}\right)\right|\lambda\right\rangle\right|_{\{\rho_a^k=0\}}
\end{split}
\end{equation}
where $\widetilde{\lambda}a$ is defined in exactly the same way as $\widetilde{p}a$ with $p^j$ replaced by $\lambda^j$. Now noting that $e^{\lambda\cdot q}|0\rangle=|\lambda\rangle$ and $q$ commutes with all other oscillators $a_n$, we have 
\begin{equation}
\begin{split}
     \left\langle\lambda\left|e^{A(z)} V(\psi, w) e^{\overline{A}(\zeta)}\right| \lambda\right\rangle &=(-4u^{-1})^{-\lambda^2}(z\zeta)^{-\frac{1}{2}h_{\psi}}\frac{\partial}{\partial \rho}e^{\lambda\cdot\rho_0}\left\langle 0\left|\exp\left(\widetilde{\lambda}a+\frac{1}{2}aTa\right)\left(e^{\rho^<a^{\dagger}}e^{\rho^>a}\right)\right.\right.\\&\hspace{6cm}\times\left.\left.\left.\exp\left(\widetilde{\lambda}a^{\dagger}+\frac{1}{2}a^{\dagger}Ta^{\dagger}\right)\right|0\right\rangle\right|_{\{\rho_a^k=0\}}
\end{split}
\end{equation} 
Normal ordering the oscillators, we can write 
\begin{equation}
    \left\langle 0\left|\exp\left(\widetilde{\lambda}a+\frac{1}{2}aTa\right)\left(e^{\rho^<a^{\dagger}}e^{\rho^>a}\right)\exp\left(\widetilde{\lambda}a^{\dagger}+\frac{1}{2}a^{\dagger}Ta^{\dagger}\right)\right|0\right\rangle=\left\langle 0\left|e^{\widetilde{\lambda}a}e^{\frac{1}{2}aTa}e^{\rho^<a^{\dagger}}e^{\rho^>a}e^{\widetilde{\lambda}a^{\dagger}}e^{\frac{1}{2}a^{\dagger}Ta^{\dagger}}\right|0\right\rangle
\end{equation}
The vacuum expectation value on the right hand side can be evaluated using coherent state techniques and has been done in \cite[Appendix D]{Dolan:1989vr}. We get 
\begin{equation}
\begin{split}
    \left\langle 0\left|e^{\widetilde{\lambda} a} e^{\frac{1}{2} a T a}e^{\rho^{<} a^{\dagger}} e^{\rho^{>}a}\right.\right.& \left.\left.e^{\frac{1}{2} a^{\dagger} T a^{\dagger}} e^{\bar{\lambda} a^{\dagger}}\right| 0\right\rangle\\&=\operatorname{det}\left(1-T^{2}\right)^{-\frac{1}{2}r} \exp \left(\widetilde{\lambda}(1-T)^{-1} \widetilde{\lambda}\right) \exp \left(\widetilde{\lambda}(1-T)^{-1}\left(\rho^{<}+\rho^{>}\right)\right) e^{h(u, y)}
\end{split}    
\end{equation}
where 
\begin{equation}
\begin{split}
    h(u, y)=\frac{1}{2}\left[\rho^{<} T\left(1-T^{2}\right)^{-1} T \rho^{>}+\rho^{>} T\left(1-T^{2}\right)^{-1} T \rho^{<}+\rho^{<}\left(1-T^{2}\right)^{-1} T \rho^{<}+\rho^{>}\left(1-T^{2}\right)^{-1} T \rho^{>}\right] .
\end{split}
\end{equation}
The determinant can be evaluated using group theory. The result is \cite[eq. (9.16)]{Dolan:1989vr}
\begin{equation}
    \operatorname{det}\left(1-T^{2}\right)^{-\frac{1}{2} r}=\left(1-u^{2}\right)^{-r / 8} \Theta_{3}(q^2)^{-r},
\end{equation}
where 
\begin{equation}\label{eq:quK}
    q=\exp \left\{-\pi \frac{K(\sqrt{1-u^2})}{K(u)}\right\}, 
\end{equation}
where $K(u)$ is a complete elliptic integral given by
\begin{equation}
    K(u)=\frac{1}{2} \int_{0}^{1} d x x^{-1 / 2}(1-x)^{-1 / 2}\left(1-u^{2} x\right)^{-1 / 2}=\frac{1}{2} \pi \Theta_{3}(q^2)^{2},
\end{equation}
Here $\Theta_3(q)$ is the Jacobi theta function \cite[Section 20.9]{NIST:DLMF} given by
\begin{equation}
\Theta_{3}(q)=\sum_{n\in\mathbb{Z}}q^{\frac{n^2}{2}}=\prod_{n=1}^{\infty}\left(1-q^{n}\right)\left(1+ q^{n-\frac{1}{2}} \right)^2.    
\end{equation}
Because we are working with $\vert \zeta \vert < \vert z \vert$ we have $\vert u\vert <1$. 
Also we have \cite[eq. (9.20)]{Dolan:1989vr} 
\begin{equation}
    \exp \left(\widetilde{\lambda}(1-T)^{-1} \widetilde{\lambda}\right) = (4u^{-1})^{\lambda^2}q^{\frac{1}{2}\lambda^2}.
\end{equation}
Next define 
\begin{equation}
    \lambda\cdot t(u,y)=\lambda\cdot\rho_0+\widetilde{\lambda}(1-T)^{-1}\left(\rho^{<}+\rho^{>}\right).
\end{equation}
Then we can easily see that 
\begin{equation}
    \begin{aligned}\left\langle\lambda\left|e^{A(z)} V(\psi, \omega) e^{\bar{A}(\zeta)}\right| \lambda\right\rangle=&(z \zeta)^{-\frac{1}{2} h_{\psi}} \frac{\partial}{\partial \rho}\left(1-u^{2}\right)^{-r/ 8} \Theta_{3}(q^2)^{-r} q^{\frac{1}{2} \lambda^{2}}\left.e^{\lambda \cdot t(u, y)} e^{h(u, y)}\right|_{\left\{\rho_{a}^{k}\right\}=0} . \end{aligned}
\end{equation}
This implies 
\begin{equation}
    \begin{aligned}\langle\mu|W(\chi, z) V(\psi, w) \overline{W}(\varphi, \zeta)| \rho\rangle=& \zeta^{-2 h_{\varphi}}(z \zeta)^{-\frac{1}{2} h_{\psi}} \sum_{\lambda \in \Lambda_{\mathrm{L}}}\left\langle\mu\left|e_{\lambda}\right| \chi\right\rangle\left\langle\overline{\varphi}\left|e_{\lambda}^{\dagger}\right| \varrho\right\rangle \\ & \times\left.\left(1-u^{2}\right)^{-r / 8} \Theta_{3}(q^2)^{-r} q^{\frac{1}{2} \lambda^{2}} \frac{\partial}{\partial \rho} e^{\lambda \cdot t(u, y)} e^{h(u, y)}\right|_{\left\{\rho^{k}_{a}\right\}=0} . \end{aligned}
    \label{eq:WWlocexp}
\end{equation}
We have kept the $r$ explicit in above formula but at the end we need to plug in $r=24$. We now need to show that the right hand side of \eqref{eq:WWlocexp} is invariant under $\varphi\leftrightarrow \chi$ and $z\leftrightarrow \zeta$  up to the factor $(-1)^{|\chi||\varphi|}=-1$.

To facilitate the exchange of $\varphi$ and $\chi$ we proceed as follows. 
We begin with the  relation \cite[Appendix C]{Dolan:1989vr}
\begin{equation}
    \sum_{\gamma\in\Gamma}\gamma_{ab}^{\dagger}\gamma_{cd}=2^{12}\delta_{ad}\delta_{bc}
\end{equation}
to transform the inner products in \eqref{eq:appinteqWVW}. Choose a basis for $\mathcal{S}$ and let us denote the column vectors for the states $\mu,\chi,\varphi,\varrho$ by the same symbol. We then have \begin{equation}
\begin{split}
    \left\langle\mu\left|e_{\lambda}\right| \chi\right\rangle\left\langle\overline{\varphi}\left|e_{\lambda}^{\dagger}\right| \varrho\right\rangle &=(e^{-i\pi h_{\varphi}})^*\left\langle\mu\left|e_{\lambda}\right| \chi\right\rangle\left\langle\theta M\varphi^*\left|e_{\lambda}^{\dagger}\right| \varrho\right\rangle\\&=-e^{i\pi h_{\varphi}}\left(\mu^{\dagger}\gamma_{\lambda}\chi\right)\left(\varphi^TM^*\gamma_{\lambda}^{\dagger}\varrho\right)\\&=-2^{-12}e^{i\pi h_{\varphi}}\sum_{\gamma_{\nu}\in\Gamma}\left(\mu^{\dagger}\gamma_{\nu}^{\dagger}\gamma_{\nu}\gamma_{\lambda}\chi\right)\left(\varphi^TM^*\gamma_{\lambda}^{\dagger}\varrho\right)\\&=-2^{-12}e^{i\pi h_{\varphi}}\sum_{\gamma_{\nu}\in\Gamma}\left(\mu^{\dagger}\gamma_{\nu}^{\dagger}\gamma_{\nu}\gamma_{\lambda}\chi\right)\left((M^*\gamma_{\lambda}^{\dagger}\varrho)^T\varphi\right)\\&=-2^{-12}e^{i\pi h_{\varphi}}\sum_{\gamma_{\nu}\in\Gamma}\left(\mu^{\dagger}\gamma_{\nu}^{\dagger}\varphi\right)\left((\varrho^T\gamma_{\lambda}^{*}M^*\gamma_{\nu}\gamma_{\lambda}\chi\right)\\&=-2^{-12}e^{i\pi h_{\varphi}}\sum_{\gamma_{\nu}\in\Gamma}\left(\mu^{\dagger}\gamma_{\nu}^{\dagger}\varphi\right)\left((\varrho^TM^*\gamma_{\lambda}\gamma_{\nu}\gamma_{\lambda}\chi\right)\\&=2^{-12}e^{i\pi h_{\varphi}}\sum_{\gamma_{\nu}\in\Gamma}\left(\mu^{\dagger}\gamma_{\nu}^{\dagger}\varphi\right)\left((\theta M^*\varrho)^T\gamma_{\lambda}\gamma_{\nu}\gamma_{\lambda}\chi\right)\\&=2^{-12}e^{i\pi (h_{\varphi}-h_{\varrho})}(-1)^{\frac{\lambda^2}{2}}\sum_{e_{\nu}\in\Gamma(\Lambda_{\mathrm{L}})}(-1)^{\lambda\cdot\nu}\left\langle\mu\left|e_{\nu}^{\dagger}\right|\varphi\right\rangle\left\langle\overline{\varrho}\left|e_{\nu}\right|\chi\right\rangle.
\end{split}  
\label{eq:fierztrans}
\end{equation} 
Thus we can write 
\begin{equation}
    \begin{aligned}
\langle\mu|W(\chi, z) V(\psi, w) \overline{W}(\varphi, \zeta)| \varrho\rangle
&=\zeta^{-2 h_{\varphi}}(z \zeta)^{-\frac{1}{2} h_{\psi}}e^{i\pi (h_{\varphi}-h_{\varrho})} 2^{-12} \sum_{e_{\nu} \in \Gamma(\Lambda_{\mathrm{L}})}\left\langle\mu\left|e_{\nu}^{\dagger}\right| \varphi\right\rangle\left\langle\overline{\varrho}\left|e_{\nu}\right| \chi\right\rangle \\
&\quad \times\left.\left(1-u^{2}\right)^{-r / 8} \Theta_{3}(q^2)^{-r} \frac{\partial}{\partial \rho}e^{h(u, y)} \Gamma(q, \nu)\right|_{\left\{\rho_{a}^{k}\right\}=0}
\end{aligned}
\end{equation}
where 
\begin{equation}
    \Gamma(q, \nu)=\sum_{\lambda \in \Lambda_{\mathrm{L}}}(-q)^{\frac{1}{2} \lambda^{2}}(-1)^{\lambda \cdot \nu} e^{\lambda \cdot t(u, y)}.
\end{equation}

We now need to exchange $z$ and $\zeta$, which is equivalent to changing  $u\to 1/u$.
Of course, this requires analytic continuation since $\vert 1/u \vert > 1$. 
Under $u\to 1/u$, $q$ does not transform simply. So we change variables to 
\begin{equation}\label{eq:q'uK}
    q'=\exp\left(-\pi\frac{K(u)}{K(\sqrt{1-u^2})}\right).
\end{equation}
This is because $q'\to-q'$ under $u\to 1/u$ \cite[eq. (9.27)]{Dolan:1989vr}. 
To change the variables from $q$ to $q'$ we need to express $\Gamma(q,\nu)$ in terms of $q'$. Note that if we write $q=e^{2\pi i\tau}$ and $q'=e^{2\pi i\tau'}$ then  
$\tau' =  -1/4\tau$.
One can now use the Poisson summation formula and get \cite[eq. (9.31)]{Dolan:1989vr} 
\begin{equation}
    \Gamma(q, \nu)=(-1)^{\frac{1}{2} \nu^{2}}\left(\frac{\ln q^{\prime}}{i \pi}\right)^{\frac{1}{2} r} \sum_{\lambda \in \Lambda_{\mathrm{L}}}\left(-q^{\prime}\right)^{\frac{1}{2} \lambda^{2}}(-1)^{\lambda \cdot \nu}\left(q^{\prime}\right)^{\lambda \cdot t(u, y) / i \pi}\left(q^{\prime}\right)^{-t(u, y)^{2} / 2 \pi^{2}} .
\end{equation}
Next we have the transformation of the theta function given by 
\begin{equation}
    \Theta_3(-1/\tau)=\sqrt{-i\tau}\Theta_3(\tau).
\end{equation}
where $\Theta_3(\tau):=\Theta_3(q)$ with $q=e^{2\pi i\tau}$.
Using this transformation, it is easy to see that 
\begin{equation}
    \Theta_3(q'^2)=\Theta_3(2\tau')=\Theta_3(-2/4\tau)=\sqrt{-2i\tau}\Theta_3(2\tau)=\sqrt{\frac{i}{2\tau'}}\Theta(q^2)=\Theta(q^2)\left(-\frac{\ln q^{\prime}}{ \pi}\right)^{-\frac{1}{2}} 
\end{equation}
where we used the fact that $q=e^{2\pi i\tau}\to q'=e^{2\pi i\tau'}$ under $\tau\to\tau'=-1/4\tau$. This gives us
%\cite[eq. (9.32)]{Dolan:1989vr} 
\begin{equation}
    \Theta_3(q^2)=\Theta_3(q'^2)\left(-\frac{\ln q^{\prime}}{ \pi}\right)^{\frac{1}{2}} 
\end{equation}

At this point we can return to the algebraic computations needed to 
exchange $\varphi$ and $\chi$. 
Doing a computation similar to \eqref{eq:fierztrans}, we get 
\begin{equation}
    2^{-12}(-1)^{\frac{1}{2} \lambda^{2}} \sum_{e_{\nu} \in \Gamma(\Lambda_{\mathrm{L}})}\left\langle\overline{\varphi}\left|e_{\nu}^{\dagger}\right| \overline{\mu}\right\rangle\left\langle\overline{\varrho}\left|e_{\nu}\right| \chi\right\rangle(-1)^{\lambda \cdot \nu}=\left\langle\overline{\varphi}\left|e_{\lambda}\right| \chi\right\rangle\left\langle\mu\left|e_{\lambda}^{\dagger}\right| \varrho\right\rangle e^{-i\pi(h_{\mu}-h_{\varrho})}
\end{equation}
To use this relation, we need to relate $\left\langle\overline{\varphi}\left|e_{\nu}^{\dagger}\right| \overline{\mu}\right\rangle$ with $\left\langle\mu\left|e_{\nu}^{\dagger}\right| \varphi\right\rangle$. Indeed we have 
\begin{equation}
    \begin{split}
        \left\langle\overline{\varphi}\left|e_{\nu}^{\dagger}\right| \overline{\mu}\right\rangle&=(e^{-i\pi h_{\varphi}})^* e^{-i\pi h_{\mu}}(M\varphi^*)^{\dagger}\gamma_{\nu}^{\dagger}M\mu^*\\&=e^{i\pi (h_{\varphi}-h_{\mu})}\varphi^TM^*\gamma_{\nu}^{\dagger}M\mu^*\quad\quad(M^*\gamma_{\nu}^{\dagger}M=\gamma_{\nu}^{T})\\&=e^{i\pi (h_{\varphi}-h_{\mu})}(-1)^{\frac{1}{2}\nu^2}\varphi^T(\gamma_{\nu}^T)^{-1}\mu^*\\&=e^{i\pi (h_{\varphi}-h_{\mu})}(-1)^{\frac{1}{2}\nu^2}(\gamma_{\nu}^{\dagger}\varphi)^T\mu^*\\&=e^{i\pi (h_{\varphi}-h_{\mu})}(-1)^{\frac{1}{2}\nu^2}\langle\mu|e_{\nu}^{\dagger}|\varphi\rangle
    \end{split}
\end{equation}
where we used the fact that $e_{\nu}$ is unitary.

Putting these facts together we can express our matrix element in the form:  
\begin{equation}
\begin{split}\langle\mu|W(\chi, z) V(\psi, w) \overline{W}(\varphi, \zeta)| \varrho\rangle &=\zeta^{-2 h_{\varphi}}(z \zeta)^{-\frac{1}{2} h_{\psi}} \frac{\partial}{\partial \rho} \sum_{\lambda \in \Lambda_{\mathrm{L}}}\left\langle\overline{\varphi}\left|e_{\lambda}\right| \chi\right\rangle\left\langle\mu\left|e_{\lambda}^{\dagger}\right| \varrho\right\rangle\left(q^{\prime}\right)^{\frac{1}{2} \lambda^{2}}\left(q^{\prime}\right)^{\lambda \cdot t(u, y) / i \pi} \\ & \times\left.i^{\frac{1}{2}r}\left(1-u^{2}\right)^{-r / 8} \Theta_{3}\left(q^{\prime 2}\right)^{-r} e^{h(u, y)}\left(q^{\prime}\right)^{-t(u, y)^{2} / 2 \pi^{2}}\right|_{\left\{\rho_{a}^{k}=0\right\}} \end{split}
\label{eq:loctwistfieldsvalue}
\end{equation} 
We now show that the right hand side of \eqref{eq:loctwistfieldsvalue} is invariant under $\varphi\leftrightarrow \chi$ and $z\leftrightarrow \zeta$ up to the factor $(-1)^{|\chi||\phi|}=-1$. Under this, $u\leftrightarrow 1/u$ and $q'\leftrightarrow -q'$. To check invariance under exchange we need the identities:  
\begin{equation}
    \begin{split}
\left\langle\overline{\varphi}\left|e_{\lambda}\right| \chi\right\rangle=\left\langle\overline{\chi}\left|e_{\lambda}\right| \varphi\right\rangle e^{i\pi (h_{\chi}+h_{\varphi})}(-1)^{\frac{1}{2} \lambda^{2}}\\
\left(1-u^{-2}\right)^{-r / 8} \Theta_{3}\left(-q^{\prime 2}\right)^{-r}=\left(1-u^{2}\right)^{-r / 8} \Theta_{3}\left(q^{\prime 2}\right)^{-r}(-1)^{r / 8}(\zeta / z)^{-r / 8}\\\Theta_3(-q'^2)=u^{1/2}\Theta_3(q'^2)
\label{eq:thetarels}
\end{split}
\end{equation}
We also note that   $\mathrm{e}^{h(u, y)}\left(q^{\prime}\right)^{-t(u, y)^{2} / 2 \pi^{2}}$ and $\left(q^{\prime}\right)^{\lambda \cdot t(u, y) / i \pi}$ are invariant under $u\leftrightarrow 1/u$ and $q'\leftrightarrow -q'$ \cite[Section 9, 10]{Dolan:1989vr} and $r=24$. Putting all these together we finally arrive at the 
desired identity:  
\begin{equation}
    \langle\mu|W(\chi, z) V(\psi, w) \overline{W}(\varphi, \zeta)| \varrho\rangle = -\langle\mu|W(\varphi, \zeta) V(\psi, w) \overline{W}(\chi,z)| \varrho\rangle. 
\end{equation}

It remains to prove the identites \eqref{eq:thetarels}. 
The last relation in eq.  \eqref{eq:thetarels} is proved using the fact that $\Theta_3(-q'^2)=\Theta_4(q'^2)$ where 
\begin{equation}
    \Theta_4(q')=\sum_{n\in\mathbb{Z}}(-1)^nq'^{\frac{n^2}{2}}=\prod_{n=1}^{\infty}\left(1-q'^{n}\right)\left(1- q'^{n-\frac{1}{2}} \right)^2
\end{equation}
and \cite{NIST:DLMF}
\begin{equation}
    \Theta_4(q'(u)^2)=u^{\frac{1}{2}}\Theta_3(q'(u)^2),
\end{equation}
where $q'(u)$ is given by eq.  \eqref{eq:q'uK}.
The middle relation in \eqref{eq:thetarels} follows from the usual modular transformation rules.
The first relation in eq.  \eqref{eq:thetarels} is similar to, but different from one proven in \cite{Dolan:1989vr}, because their context is slightly different. We prove the first relation as follows:  
\begin{equation}
\begin{split}
\left\langle\overline{\varphi}\left|e_{\lambda}\right| \chi\right\rangle &=\left(e^{-i\pi h_{\varphi}}\right)^*\left\langle \theta M(\varphi)|e_{\lambda}|\chi\right\rangle\\&=-e^{i\pi h_{\varphi}}\varphi^TM^*\gamma_{\lambda}\chi\\&=-e^{i\pi h_{\varphi}}\varphi^T\gamma_{\lambda}^*M^*\chi\quad\quad(M\gamma_{\lambda}^*=\gamma_{\lambda} M\text{ and }M^*=M^{-1})\\&=e^{i\pi (h_{\varphi}-h_{\chi})}\varphi^T\gamma_{\lambda}^*\left(e^{-i\pi h_{\chi}}\theta M\chi^*\right)^*\\&=e^{i\pi (h_{\varphi}+h_{\chi})}\varphi^T\gamma_{\lambda}^*\overline{\chi}^*\\&=e^{i\pi (h_{\varphi}-h_{\chi})}\varphi^T\gamma_{\lambda}^T(\gamma_{\lambda}^{-1})^T\gamma_{\lambda}^*\overline{\chi}^*\\&=e^{i\pi (h_{\varphi}-h_{\chi})}(\gamma_{\lambda}\varphi)^T(\gamma_{\lambda}^*)^2\overline{\chi}^*\\&=e^{i\pi (h_{\varphi}-h_{\chi})}(-1)^{\frac{\lambda^2}{2}}(\gamma_{\lambda}\varphi)^T\overline{\chi}^*\\&=e^{i\pi (h_{\varphi}-h_{\chi})}(-1)^{\frac{\lambda^2}{2}}\langle\overline{\chi}|e_{\lambda}|\varphi\rangle,
\end{split}    
\end{equation}
where we used the properties of the gamma matrices and the fact that they are unitary.
\end{proof}
\end{thm}

To complete the proof of locality, we need to promote the states $\chi$ and $\varphi$ to arbitrary fermionic states in $\mathscr{H}_T^-(\Lambda_{\mathrm{L}})$. 
\begin{lemma}
The locality relation $W(\chi,z)\overline{W}(\varphi,w)=-W(\varphi,w)\overline{W}(\chi,z)$ for arbitrary fermionic twisted states $\chi,\varphi$ holds if and only if
\begin{equation}
    W(\chi, z)\overline{W}(\varphi, \zeta)=-W(\varphi, \zeta) \overline{W}(\chi, z)
    \label{eq:WWbarlocrel}
\end{equation}
for $\chi,\varphi$ twisted ground states. 
\begin{proof}
Direct implication is obvious. We prove the converse. 
Since we already have the locality of $\widetilde{V}(\psi,z)\widetilde{V}(\chi,w)=\widetilde{V}(\chi,w)\widetilde{V}(\psi,z)$ for twisted state $\chi$ and untwisted state $\psi$, we have the duality relation
\begin{equation}
    \widetilde{V}(\psi,z)\widetilde{V}(\chi,w)=\widetilde{V}(\widetilde{V}(\psi,z-w)\chi,w).
\label{eq:VVdualitylocrel}
\end{equation}
Note that by Proposition \ref{prop:sqrtsingope}, the right hand side of from \eqref{eq:VVdualitylocrel} has square root singularity when $\chi\in\mathscr{H}_T^-(\Lambda_{\mathrm{L}})$ and $\psi\in\mathscr{H}^-(\Lambda_{\mathrm{L}})$. Fortunately we do not have to use  \eqref{eq:VVdualitylocrel} for $\psi\in\mathscr{H}^-(\Lambda_{\mathrm{L}})$. To prove this lemma, we only need $\psi\in\mathscr{H}^+(\Lambda_{\mathrm{L}})$. 
Using explicit matrix forms of $\widetilde{V}$  \eqref{eq:veropuntbb} and \eqref{eq:veroptbb}, we get 
\[
\begin{split}
V(\psi,z)\overline{W}(\chi,w)=\overline{W}(V_T(\psi,z-w)\chi,w)\\
V_T(\psi,z)W(\chi,w)=W(V_T(\psi,z-w)\chi,w)
\end{split}
\]
Acting on $W(\chi, z)  \overline{W}(\varphi, \zeta)| \varrho\rangle$ by $V_T(\psi,w)$ with $\psi\in\mathscr{H}^+(\Lambda_{\mathrm{L}})$ and using the above relations and Lemma \ref{lemma:VTpsichi}, we can generalise $\chi$ to arbitrary fermionic twisted state on both sides of \eqref{eq:WWbarlocrel}. This is because $V_T(\psi,z)$ maps fermions to bosons and vice-versa in the twisted sector if $\theta\psi=-\psi$ and preserve the parity if $\theta \psi=\psi$. Similarly acting with $V_T(\psi,w)$ on $W(\chi, z)  \overline{W}(\varphi, \zeta) \varrho$, we can generalise $\varphi$ to arbitrary fermionic twisted state on both sides of \eqref{eq:WWbarlocrel}. The statement of the lemma now follows. 
\end{proof}
\end{lemma}
%To prove the remaining locality relations, we note that the locality relation $W(\chi, z)\overline{W}(\varphi, \zeta)=W(\varphi, \zeta) \overline{W}(\chi, z)$ for bosonic twisted states $\chi$ and $|\varphi\rangle$ of conformal weights 2 has been proved in \cite[Section 11]{Dolan:1989vr}. Note that $\widetilde{V}(\widetilde{V}(\psi,z-w)\chi,w)$ does not have square root singularity when $|\psi\rangle\in\mathscr{H}^-(\Lambda_{\mathrm{L}})$ and $\chi\in\mathscr{H}^+(\Lambda_{\mathrm{L}})$ 
\end{appendix}
\noindent\textbf{Data availability:} Data sharing not applicable to this article as no datasets were generated or
analyzed during the current study.\\\\\noindent\textbf{Funding and/or Conflicts of interests/Competing interests:} The authors have no competing interests to declare that are relevant to the
contents of this article.
\printbibliography
\end{document}